\documentclass[11pt]{article}

\usepackage{amsmath, amssymb}

\usepackage{amsthm}

\usepackage{ifthen}

\usepackage{ifpdf}

\ifpdf
\usepackage[pdftex]{graphicx}
\usepackage[pdftex]{hyperref}
\else
\usepackage[dvips]{graphicx}
\fi

\ifthenelse{\isundefined{\NoGenerateLetterSize}}
{

\ifpdf
\pdfpageheight\paperheight
\pdfpagewidth\paperwidth
\fi
}
{}

\newcommand{\Comment}[1]{}

\long\def\LongVersion#1\LongVersionEnd{#1}
\long\def\ShortVersion#1\ShortVersionEnd{}

\LongVersion 
\LongVersionEnd 

\ShortVersion 
\ShortVersionEnd 

\ShortVersion 
\usepackage{times}
\ShortVersionEnd 

\newtheorem{theorem}{Theorem}[section]
\newtheorem{lemma}[theorem]{Lemma}
\newtheorem{observation}[theorem]{Observation}
\newtheorem{corollary}[theorem]{Corollary}
\newtheorem{proposition}[theorem]{Proposition}

\theoremstyle{definition}
\newtheorem{property}[theorem]{Property}
\theoremstyle{plain}

\Comment{
\newtheorem{theorem}{Theorem}
\newtheorem{lemma}[theorem]{Lemma}
\newtheorem{observation}[theorem]{Observation}
\newtheorem{corollary}[theorem]{Corollary}
\newtheorem{proposition}[theorem]{Proposition}

\theoremstyle{definition}
\newtheorem{property}[theorem]{Property}
\theoremstyle{plain}
}

\Comment{
\newtheorem{theorem}{Theorem}[section]
\newtheorem{lemma}{Lemma}[section]
\newtheorem{corollary}{Corollary}[section]

\newtheorem{proposition}{Proposition}[section]

\newtheorem{observation}{Observation}[section]

\theoremstyle{definition}

\theoremstyle{plain}
}

\newtheorem{GlobalTheorem}{Theorem}

\newtheorem*{claim*}{Claim}
\newtheorem*{observation*}{Observation}

\newenvironment{AvoidOverfullParagraph}[0]
{\sloppy\ignorespaces}
{\par\fussy\ignorespacesafterend}

\newcommand{\Station}[0]{\mathit{s}}
\newcommand{\cA}{\mathcal{A}}
\newcommand{\Power}[0]{\psi}
\newcommand{\Noise}[0]{\mathit{N}}
\newcommand{\Reals}[0]{\mathbb{R}}

\newcommand{\Grid}[0]{\mathit{G}}
\newcommand{\Area}[0]{\mathrm{area}}
\newcommand{\Perimeter}[0]{\mathrm{per}}
\newcommand{\SmallRadius}[0]{\delta}
\newcommand{\LargeRadius}[0]{\Delta}
\newcommand{\Boundary}[0]{\partial}
\newcommand{\DataStructure}[0]{\ensuremath{\mathtt{DS}}}
\newcommand{\QuerDS}[0]{\ensuremath{\mathtt{QDS}}}
\newcommand{\ColumnsCollection}[0]{\mathit{Q}}
\newcommand{\Energy}[0]{\mathrm{E}}
\newcommand{\GridSpace}[0]{\gamma}
\newcommand{\ReceptionZone}[0]{\mathcal{H}}
\newcommand{\Segment}[2]{\overline{{#1} \, {#2}}}
\newcommand{\Interference}[0]{\mathrm{I}}
\newcommand{\SINR}[0]{\mathrm{SINR}}
\newcommand{\Sign}[0]{\mathrm{sign}}
\newcommand{\Remainder}[0]{\mathrm{rem}}
\newcommand{\Var}[0]{\mathrm{SC}}
\newcommand{\Ball}[0]{\mathit{B}}
\newcommand{\HearPoly}[0]{\mathit{H}}
\newcommand{\QuerPoly}[0]{\mathit{Q}}
\newcommand{\PolyZone}[0]{\mathcal{Q}}
\newcommand{\Type}[0]{\mathrm{T}}
\newcommand{\MinDist}[0]{\kappa}

\newcommand{\dist}[1]{\mathrm{dist} ({#1})}
\newcommand{\RightMost}[0]{\mu_r}
\newcommand{\LeftMost}[0]{\mu_l}
\newcommand{\SmallRadiusBound}[0]{\tilde{\SmallRadius}}
\newcommand{\LargeRadiusBound}[0]{\tilde{\LargeRadius}}
\newcommand{\FatnessParameter}[0]{\varphi}
\newcommand{\hx}[0]{\mathit{\bar{r}}}
\LongVersion 
\newcommand{\UPN}[0]{uniform power network}
\LongVersionEnd 
\ShortVersion 
\newcommand{\UPN}[0]{UPN}
\ShortVersionEnd 
\LongVersion 
\newcommand{\Figure}[0]{Figure}
\LongVersionEnd 
\ShortVersion 
\newcommand{\Figure}[0]{Fig.}
\ShortVersionEnd 

\newcommand{\Paragraph}[1]{\vspace{-0.03in}\paragraph{#1}}
\def\dnsparagraph{\vspace{-10pt}\paragraph}

\begin{document}

\begin{titlepage}

\title{SINR Diagrams: \\
Towards Algorithmically Usable SINR Models of Wireless Networks}

\author{
Chen Avin\footnotemark[1]
\and
Yuval Emek\footnotemark[2]
\and
Erez Kantor\footnotemark[2]
\and
Zvi Lotker\footnotemark[1]
\and
David Peleg\footnotemark[2]
\and
Liam Roditty\footnotemark[3]
}

\date{\today}

\def\thefootnote{\fnsymbol{footnote}}

\footnotetext[1]{
\noindent
Department of Communication Systems Engineering, Ben Gurion University, 
Beer-Sheva, Israel.
E-mail:{\tt \{avin,zvilo\}@cse.bgu.ac.il}.
Z. Lotker was partially supported by a gift from Cisco research center}

\footnotetext[2]{
\noindent
Department of Computer Science and Applied Mathematics,
Weizmann Institute of Science, Rehovot, Israel.
E-mail: 
{\tt \{yuval.emek,erez.kantor,david.peleg\}@weizmann.ac.il}.
Supported in part by grants from the Minerva Foundation and
the Israel Ministry of Science.
}

\footnotetext[3]{
\noindent
Department of Computer Science, 
Bar Ilan University, Ramat-Gan, Israel.
E-mail: 
{\tt liam.roditty@gmail.com}.
}

\maketitle

\begin{abstract}
The rules governing the availability and quality of connections in a wireless 
network are described by \emph{physical} models such as the
\emph{signal-to-interference \& noise ratio (SINR)} model.
For a collection of simultaneously transmitting stations in the plane, it is
possible to identify a \emph{reception zone} for each station, consisting of
the points where its transmission is received correctly.
The resulting \emph{SINR diagram} partitions the plane into a reception zone
per station and the remaining plane where no station can be heard.

SINR diagrams appear to be fundamental to understanding the behavior of
wireless networks, and may play a key role in the development of suitable
algorithms for such networks, analogous perhaps to the role played by
Voronoi diagrams in the study of proximity queries and related issues in
computational geometry.
So far, however, the properties of SINR diagrams have not
been studied systematically, and most algorithmic studies in wireless
networking rely on simplified \emph{graph-based} models such as the \emph{unit
disk graph (UDG)} model, which conveniently abstract away interference-related
complications, and make it easier to handle algorithmic issues, but
consequently fail to capture accurately some important aspects of wireless
networks.

The current paper focuses on obtaining some basic understanding of SINR
diagrams, their properties and their usability in algorithmic applications.
Specifically, based on some algebraic properties of the polynomials defining
the reception zones we show that assuming uniform power transmissions, the
reception zones are convex and relatively well-rounded.
These results are then used to develop an efficient approximation algorithm
for a fundamental point location problem in wireless networks.
\end{abstract}

\thispagestyle{empty}

\end{titlepage}

\section{Introduction}
\label{section:Introduction}

\ShortVersion 
\Paragraph{Background.}
\ShortVersionEnd 
\LongVersion 
\subsection{Background}
\LongVersionEnd 
\LongVersion 
It is commonly accepted that traditional (wired, point-to-point) communication
networks are satisfactorily represented using a graph based model.
The question of whether a station $s$ is able to transmit a message 
to another station $s'$ depends on a single
(necessary and sufficient) condition, namely, that there be a wire 
connecting the two stations.
This condition is thus independent of the locations of the two stations, of
their other connections and activities, and of the locations, connections or
activities of other nearby stations\footnote{
Broadcast domain wired networks such as LANs are an exception, but even most
LANs are collections of point-to-point connections.
}.

In contrast, wireless networks are considerably harder to represent
faithfully, due to the fact that deciding whether a transmission 
by a station $s$ is successfully received by another station $s'$
is nontrivial, and depends on the positioning 
and activities of $s$ and $s'$, as well as on the positioning
and activities of other nearby stations, which might interfere with
the transmission and prevent its successful reception. 
Thus such a transmission from $s$ may reach $s'$ under certain circumstances 
but fail to reach it under other circumstances.
Moreover, the question is not entirely ``binary'', in the sense that
connections can be of varying quality and capacity.
\LongVersionEnd 
\ShortVersion 
Wireless networks are hard to represent faithfully,
due to the fact that deciding whether a transmission 
by a station $s$ is successfully received by another station $s'$
is nontrivial, and depends on the positioning and activities of $s$, $s'$,
and nearby stations that might interfere with
the transmission and prevent its successful reception. 
Thus such a transmission from $s$ may reach $s'$ under certain circumstances 
but fail to reach it under other circumstances.
Moreover, the question is not entirely ``binary'', in the sense that
connections can be of varying quality and capacity.
\ShortVersionEnd 

The rules governing the availability and quality of wireless connections
can be described by {\em physical} or {\em fading channel} models
(cf. \cite{PL95,B96,R96}).
Among those, the most commonly studied is the
{\em signal-to-interference \& noise ratio (SINR)} model.
In the SINR model, the energy of a signal fades with the distance 
to the power of the path-loss parameter $\alpha$. 
If the signal strength received by a device divided by the interfering
strength of other simultaneous transmissions (plus the fixed \emph{background
noise} \(\Noise\)) is above some \emph{reception threshold} $\beta$, then the
receiver successfully receives the message, otherwise it does not.
Formally, denote by $\dist{p,q}$ the Euclidean distance between $p$ and $q$,
and assume that each station \(\Station_i\) transmits with power \(\Power_i\).
(A \emph{uniform power network} is one where all stations transmit with the
same power.)
At an arbitrary point $p$, the transmission of station \(\Station_i\) is
correctly received if
\[
\frac{\Power_i \cdot \dist{p, s_i}^{-\alpha}}
{N + \sum_{j \neq i} \Power_j \cdot \dist{p, s_j}^{-\alpha}}
~ \geq ~ \beta ~ .
\]
Hence for a collection \( S = \{ \Station_0, \ldots ,\Station_{n - 1} \} \) of
simultaneously transmitting stations in the plane, it is possible to identify
with each station \(\Station_i\) a \emph{reception zone} \(\ReceptionZone_i\)
consisting of the points where the transmission of \(\Station_i\) is received
correctly.
It is believed that the path-loss parameter \( 2 \leq \alpha \leq 4 \), where
\( \alpha = 2 \) is the common ``textbook'' choice, and that the reception
threshold \( \beta \approx 6 \) (\(\beta\) is always assumed to be greater
than \(1\)).

To illustrate how reception depends on the locations and activities of other
stations, consider (the numerically generated) \Figure{}~\ref{fig:example}.
(Throughout, figures are deferred to the end of the Appendix.)
\Figure{}~\ref{fig:example}(A) depicts uniform stations \( \Station_1,
\Station_2, \Station_3 \) and their reception zones.
Point $p$ (represented as a solid black square) falls inside
\(\ReceptionZone_2\).
\Figure{}~\ref{fig:example}(B) depicts the same stations except station
\(\Station_1\) has moved, so that now $p$ does not receive any transmission.
\Figure{}~\ref{fig:example}(C) depicts the stations in the same positions as
\Figure{}~\ref{fig:example}(B), but now \(\Station_3\) is silent, and as a
result, the other two stations have larger reception zones, and $p$ receives
the message of \(\Station_1\).

\newcommand{\FigureExample}{
\begin{figure}[htbp]
\centering
\begin{tabular}{ccc}
\begin{minipage}{1.8in}
\includegraphics[width=1.8in]{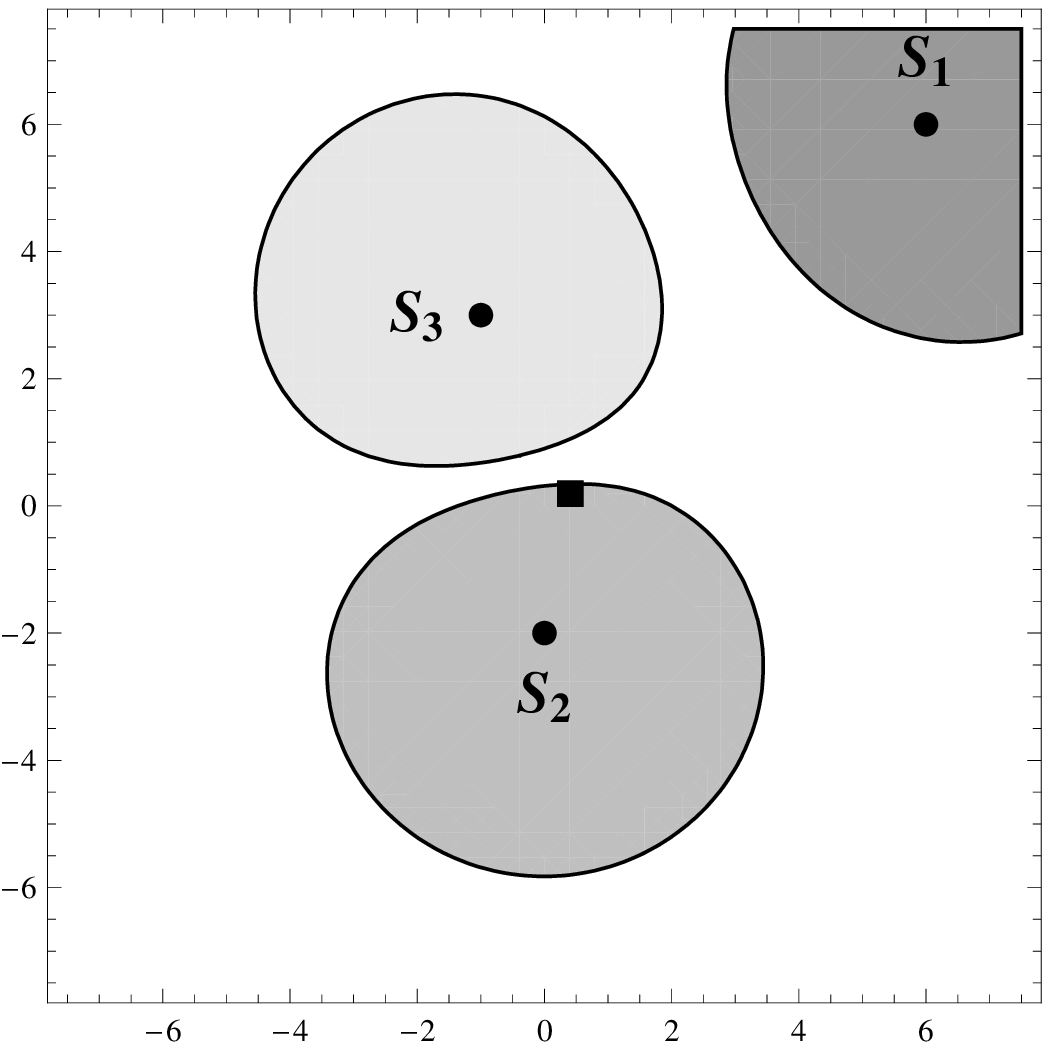}
\end{minipage}&
\begin{minipage}{1.8in}
\includegraphics[width=1.8in]{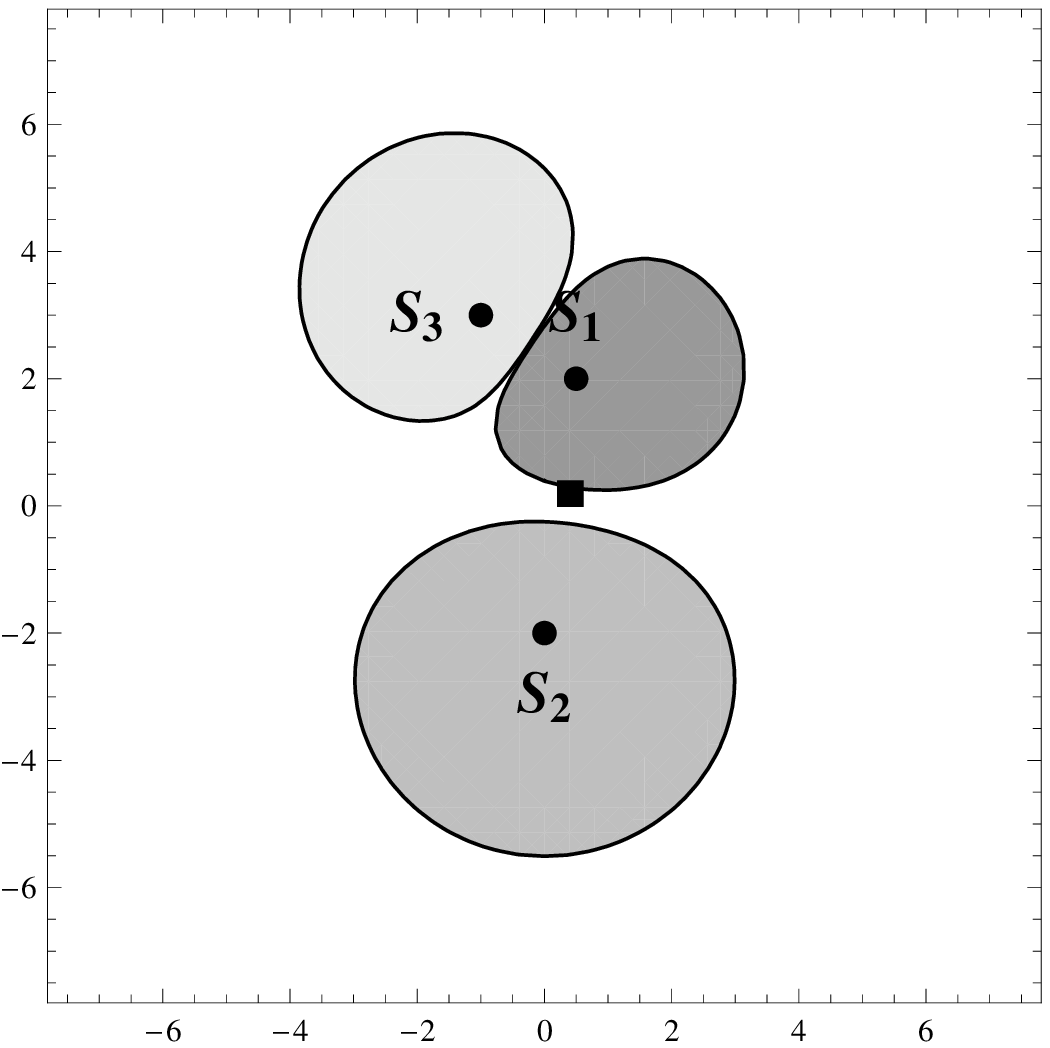}
\end{minipage}&
\begin{minipage}{1.8in}
\includegraphics[width=1.8in]{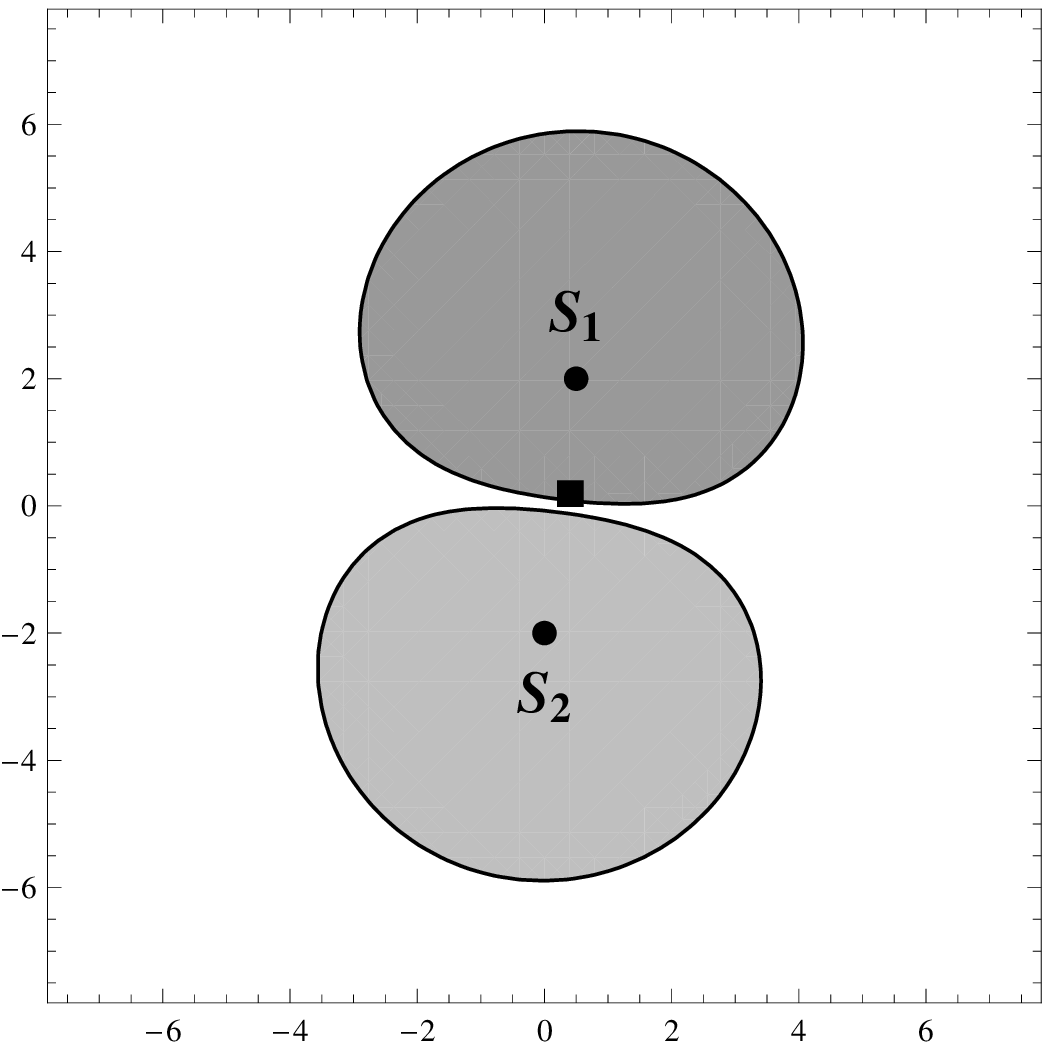}
\end{minipage} \\
\small (A) & \small (B) & \small (C)
\end{tabular}
\caption{An example of SINR diagram with three transmitters 
$s_1, s_2, s_3$ and one receiver denoted by the solid black square. 
(A) The receiver can hear $s_2$.
(B) Station $s_1$ moves and now the receiver cannot hear any transmission. 
(C) If, at the same locations as in (B), $s_3$ is silent, then the receiver 
can hear $s_1$.}
\label{fig:example}
\end{figure}
} 
\LongVersion 
\FigureExample{}
\LongVersionEnd 

\Comment{
Figure \ref{fig:non-convex-SINR} illustrates a more involved configuration
which may arise in the non-uniform case.
(Here, station $s_1$ transmits with a higher power level 
than $s_2$ and $s_3$.)

\newcommand{\FigureNonConvexSinr}{
\begin{figure}[htbp]
\centering
\begin{tabular}{cc}
\begin{minipage}{2.5in}
\includegraphics[width=2.5in]{figs/Power1}
\end{minipage}
\end{tabular}
\caption{Reception zones in the SINR model for a non-uniform power network.
Here, station $s_1$ transmits with a higher power level
than the other two stations.
The white crescent-like internal region
represents a ``no-reception'' area.}
\label{fig:non-convex-SINR}
\end{figure}
} 
\LongVersion 
\FigureNonConvexSinr{}
\LongVersionEnd 
} 

\Figure{}~\ref{fig:example} illustrates
\Comment{
Figures \ref{fig:example} and \ref{fig:non-convex-SINR} illustrate 
} 
a concept central to this paper, namely, the {\em SINR diagram}. 
An SINR diagram is a ``reception map'' characterizing the reception zones of
the stations, namely, partitioning the plane into $n$ reception zones
\(\ReceptionZone_i\), \( 0 \leq i \leq n - 1 \), and a zone
\(\ReceptionZone_{\emptyset}\) where no station can be heard.
In many scenarios the diagram changes dynamically with time, as the stations
may choose to transmit or keep silent, adjust their transmission power level,
or even change their location from time to time.

It is our belief that SINR diagrams are fundamental to understanding the
dynamics of wireless networks, and will play a key role in the development of
suitable algorithms for such networks, analogous perhaps to the role played
by Voronoi diagrams in the study of proximity queries and related issues in
computational geometry.
Yet, to the best of our knowledge, SINR diagrams have not been studied
systematically so far, from either geometric, combinatorial, 
or algorithmic standpoints.
In particular, in the SINR model it is not clear what shapes the reception
zones may take, and it is not easy to construct an SINR diagram even in a
static setting.

Taking a broader perspective, a closely related concern motivating 
this paper is that while a fair amount of research 
exists on the SINR model and other variants of the physical model, 
little has been done in such models in the {\em algorithmic} arena.
(Some recent exceptions are
\cite{GoOsWa07,GuKu00,KuWaZo03,Mo07,MWW06,MoWaZo06,RSWZ05}.)
The main reason for this is that SINR models are complex and hard 
to work with.
In these models it is even hard to decide some of the most
elementary questions on a given setting, and it is definitely more difficult
to develop communication or design protocols, 
prove their correctness and analyze their efficiency.

Subsequently, most studies of higher-layer concepts in wireless multi-hop 
\LongVersion 
networking, including issues such as transmission scheduling, 
frequency allocation, topology control, connectivity maintenance,
routing, and related design and communication tasks,
\LongVersionEnd 
\ShortVersion 
networking 
\ShortVersionEnd 
rely on simplified {\em graph-based} models rather than on the SINR model.
\LongVersion 
Graph-based models represent the network by a graph $G=(S,E)$ such that 
a station $s$ will successfully receive a message transmitted by a station $s'$
if and only if $s$ and $s'$ are neighbors in $G$ and $s$ does not have 
a concurrently transmitting neighbor in $G$.
\LongVersionEnd 
In particular, the model of choice for many protocol designers is
the {\em unit disk graph (UDG)} model \cite{Clark1990unit}.
In this model, also known as the \emph{protocol} model \cite{GuKu00},
\ShortVersion 
two stations are considered to be \emph{neighbors} if their Euclidean distance
is at most \(1\).
This defines the \emph{UDG graph}.
A silent station \(\Station\) successfully receives the message of a
transmitting station \(\Station'\) if \(\Station'\) is a neighbor of
\(\Station\) and no other neighbor of \(\Station\) transmits concurrently.
\ShortVersionEnd 
\LongVersion 
the stations are represented as points in the Euclidean plane, 
and the transmission of a station can be received by every other
station within a unit ball around it.
The \emph{UDG graph} is thus a graph whose vertices correspond to the
stations, with an edge connecting any two vertices whose corresponding
stations are at distance at most one from each other.
\LongVersionEnd 

Graph-based models are attractive for higher-layer protocol design, 
as they conveniently abstract away interference-related complications. 
\Comment{
For example, returning to the issue of obtaining a ``reception map'',
observe that a {\em UDG diagram} can be constructed in linear time
in a straightforward manner. 
} 
Issues of topology control, scheduling
and allocation are also handled more directly, since notions such as 
adjacency and overlap are easier to define and test, in turn making it
simpler to employ also some useful derived concepts such as 
domination, independence, clusters, and so on.
\LongVersion 
(Note also that the SINR model in itself is rather simplistic, as it assumes 
perfectly isotropic antennas and ignores environmental obstructions. 
These issues can be integrated into the basic SINR model, 
at the cost of yielding relatively complicated "SINR+" models,
even harder to use by protocol designers. 
In contrast, graph-based models naturally incorporate both 
directional antennas and terrain obstructions.)
\LongVersionEnd 
On the down side, it should be realized that graph-based models,
and in particular the UDG model, ignore or do not accurately capture 
a number of important physical aspects of real wireless networks.
In particular, such models oversimplify the physical laws of interference;
in reality, several nodes slightly outside the reception range
of a receiver station $v$ (which consequently are not adjacent to $v$
in the UDG graph) might still generate enough cumulative interference 
to prevent $v$ from successfully receiving a message from a sender station 
adjacent to it in the UDG graph; see \Figure{}~\ref{fig:sinr_udg_3} for an
example.
Hence the UDG model might yield a ``false positive'' indication of reception.
\newcommand{\FigureSinrUdgThree}{
\begin{figure}[htbp]
\centering
\begin{tabular}{cc}
\begin{minipage}{1.8in}
\includegraphics[width=1.8in]{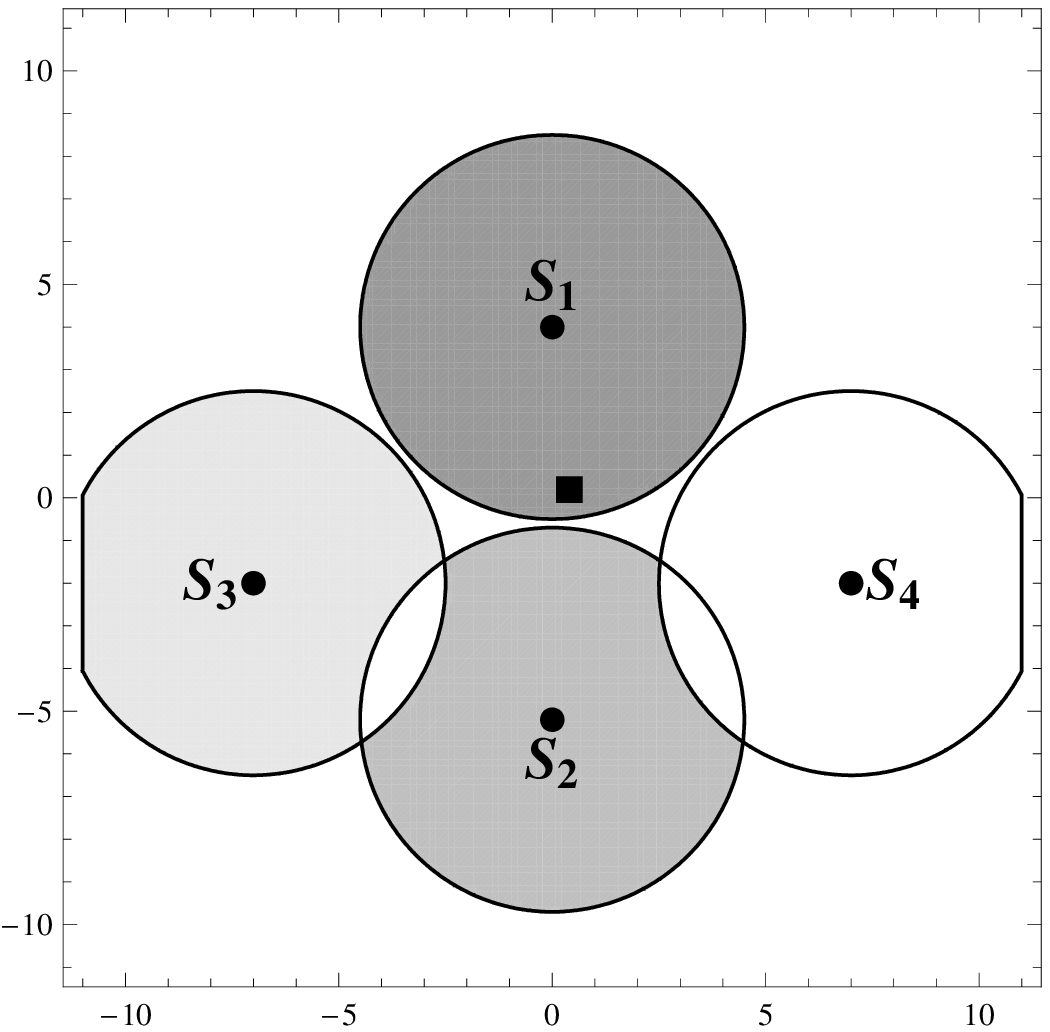}
\end{minipage}&
\begin{minipage}{1.8in}
\includegraphics[width=1.8in]{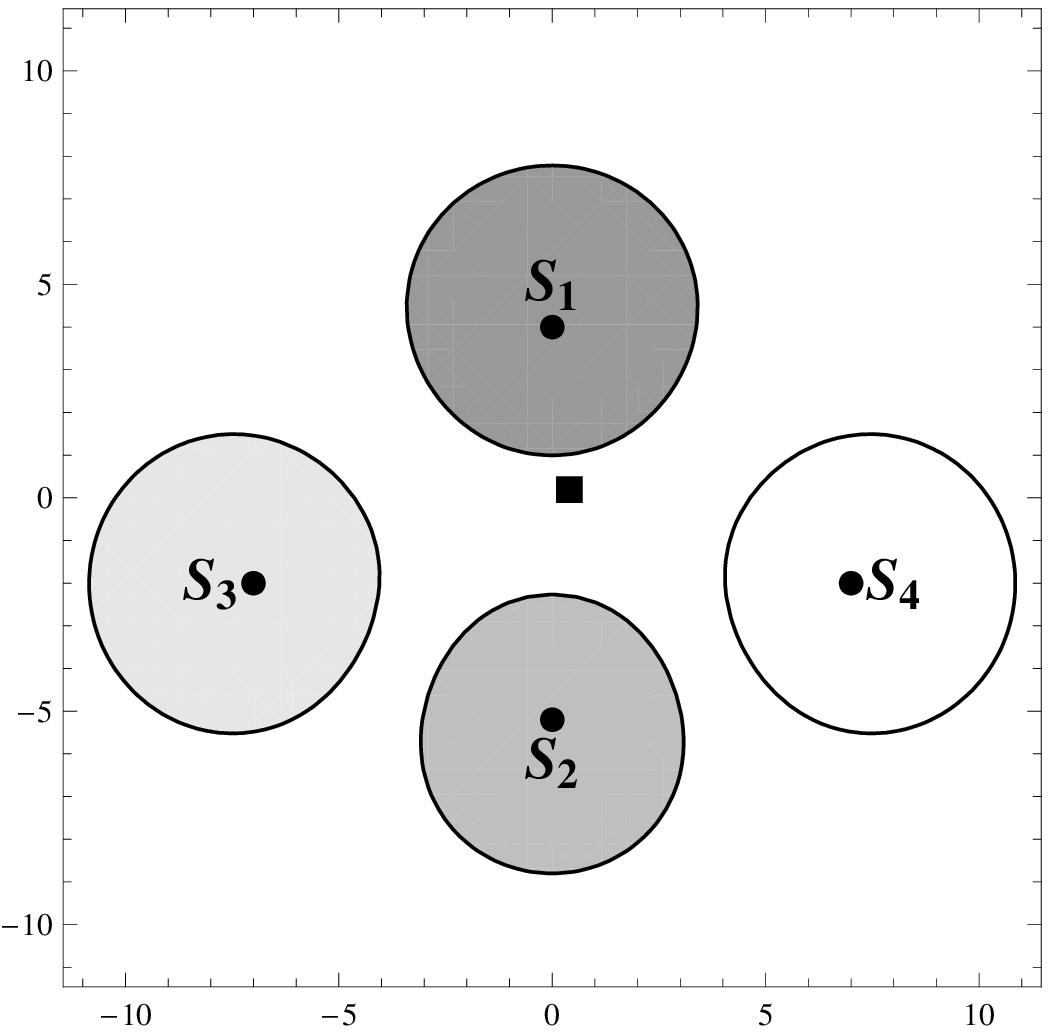}
\end{minipage} \\
\small (A) - UDG Diagram & \small (B) - SINR Diagram
\end{tabular}
\caption{Cumulative interference in the UDG and SINR models. 
(A) UDG diagram: $p$ can hear $s_1$.
(B) SINR diagram: the cumulative 
interference of stations $s_2, s_3, s_4$ prevents $p$ from hearing $s_1$.}
\label{fig:sinr_udg_3}
\end{figure}
} 
\LongVersion 
\FigureSinrUdgThree{}
\LongVersionEnd 
Conversely, a simultaneous transmission by two or more neighbors should not
always end in collision and loss of the message; in reality this depends
on other factors, such as the relative distances and the relative strength
of the transmissions.
\LongVersion 
We illustrate some of these scenarios in 
Figures \ref{fig:sinr_udg_1}, \ref{fig:sinr_udg_2}, where we compare 
the reception zones of the UDG and SINR models with four 
transmitting stations $s_1, s_2, s_3, s_4$ and one receiver $p$ 
(represented as a solid black square). 
In \Figure{}~\ref{fig:sinr_udg_1} only station $s_1$ transmits,
and all others remain silent, so the diagrams are the same 
and $p$ can hear $s_1$ in both models.
Figure \ref{fig:sinr_udg_2} illustrates the next three steps of gradually
adding $s_2, s_3$ and $s_4$ to the transmitting set.
When both $s_1, s_2$ 
transmit simultaneously, $p$ cannot hear any station in the UDG model, 
but it does hear $s_1$ in the SINR model (cases (A) and (B) respectively). 
\LongVersionEnd 
Hence in this case the UDG model yields a ``false negative'' indication.
\LongVersion 
When $s_3$ joins the transmitting stations, $p$ still cannot hear any station 
in the UDG model, but now it can hear station $s_3$ in the SINR model 
(cases (C) and (D)).
In step 4, when $s_4$ starts to transmit as well, 
the effect varies again across the two models (cases (E) and (F)).
\LongVersionEnd 

\newcommand{\FigureSinrUdgOne}{
\begin{figure}[htbp]
\centering
\begin{tabular}{cc}
\begin{minipage}{2.2in}
\includegraphics[width=2.2in]{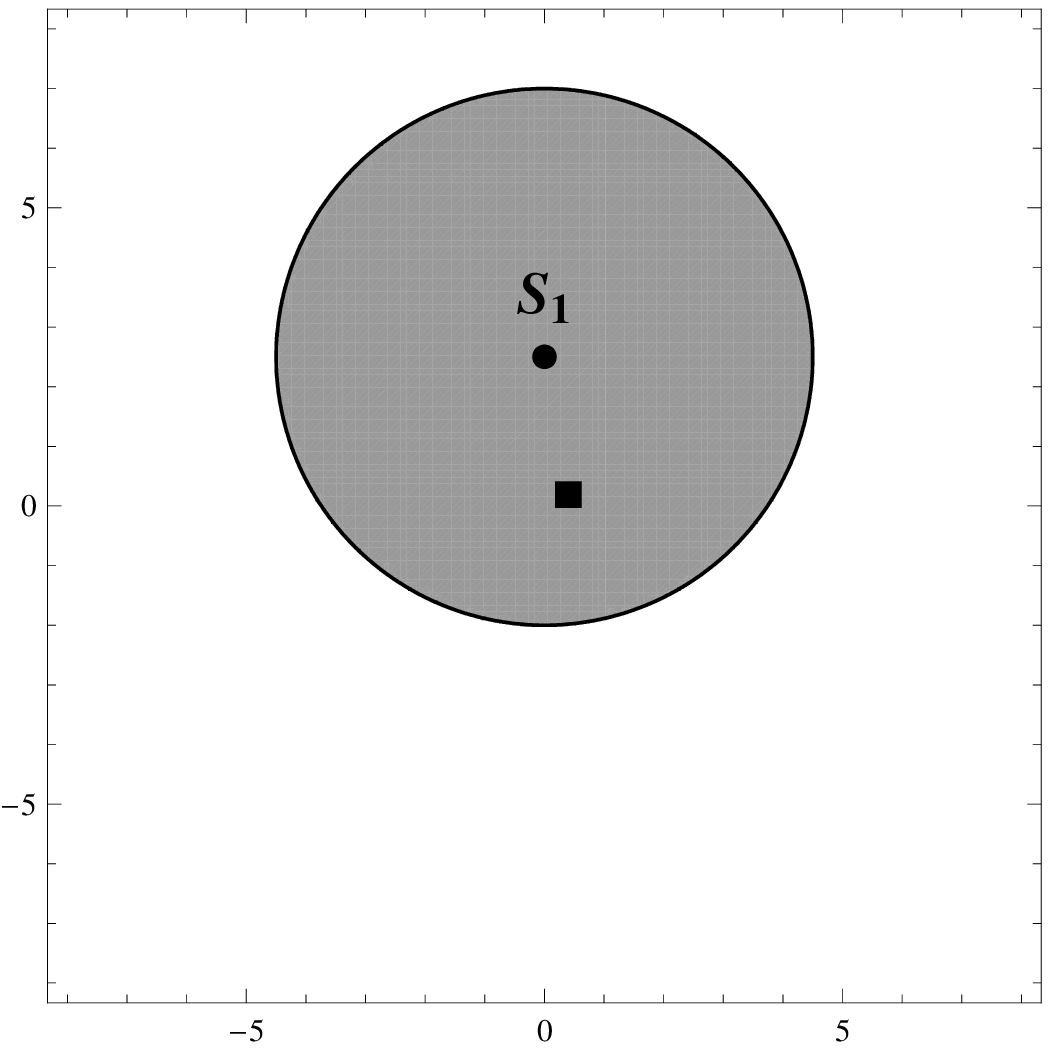}
\end{minipage}&
\begin{minipage}{2.2in}
\includegraphics[width=2.2in]{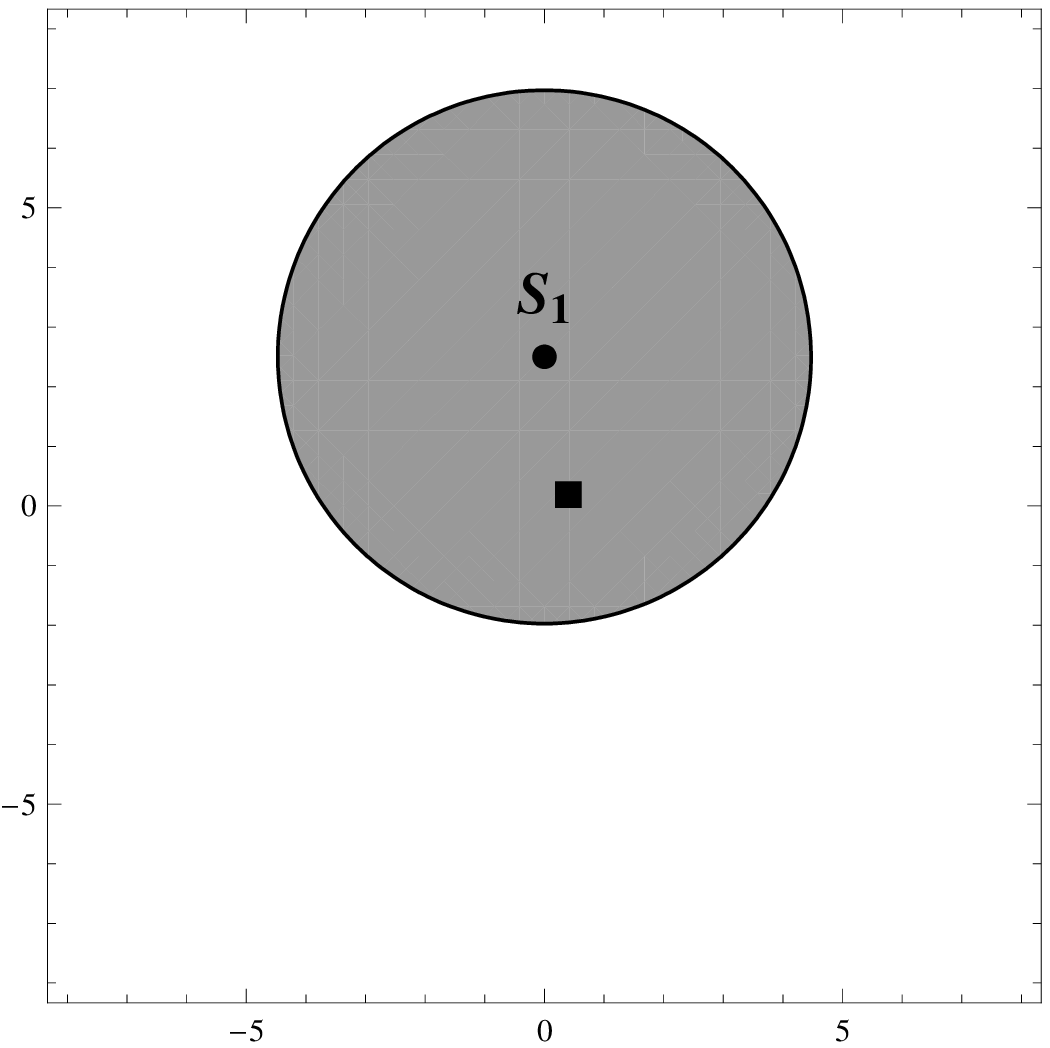}
\end{minipage}\\
\small (A) - UDG Diagram & \small (B) - SINR Diagram
\end{tabular}
\caption{Reception zones in the UDG and SINR models. 
In step 1 only $s_1$ transmits, so the reception zones are the same. 
}
\label{fig:sinr_udg_1}
\end{figure}
} 
\LongVersion 
\FigureSinrUdgOne{}
\LongVersionEnd 

\newcommand{\FigureSinrUdgTwo}{
\begin{figure}[htbp]
\centering
\begin{tabular}{cc}
\begin{minipage}{2.2in}
\includegraphics[width=2.2in]{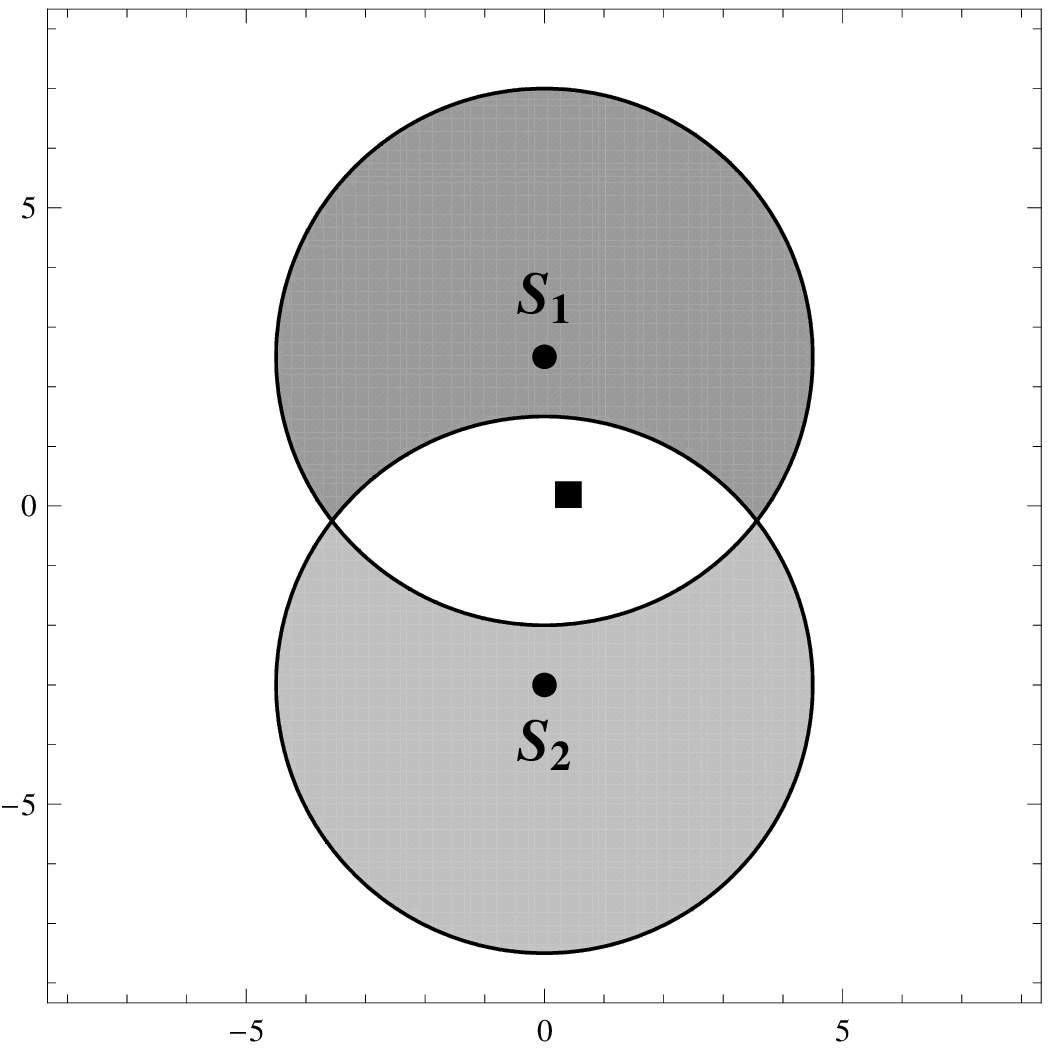}
\end{minipage}&
\begin{minipage}{2.2in}
\includegraphics[width=2.2in]{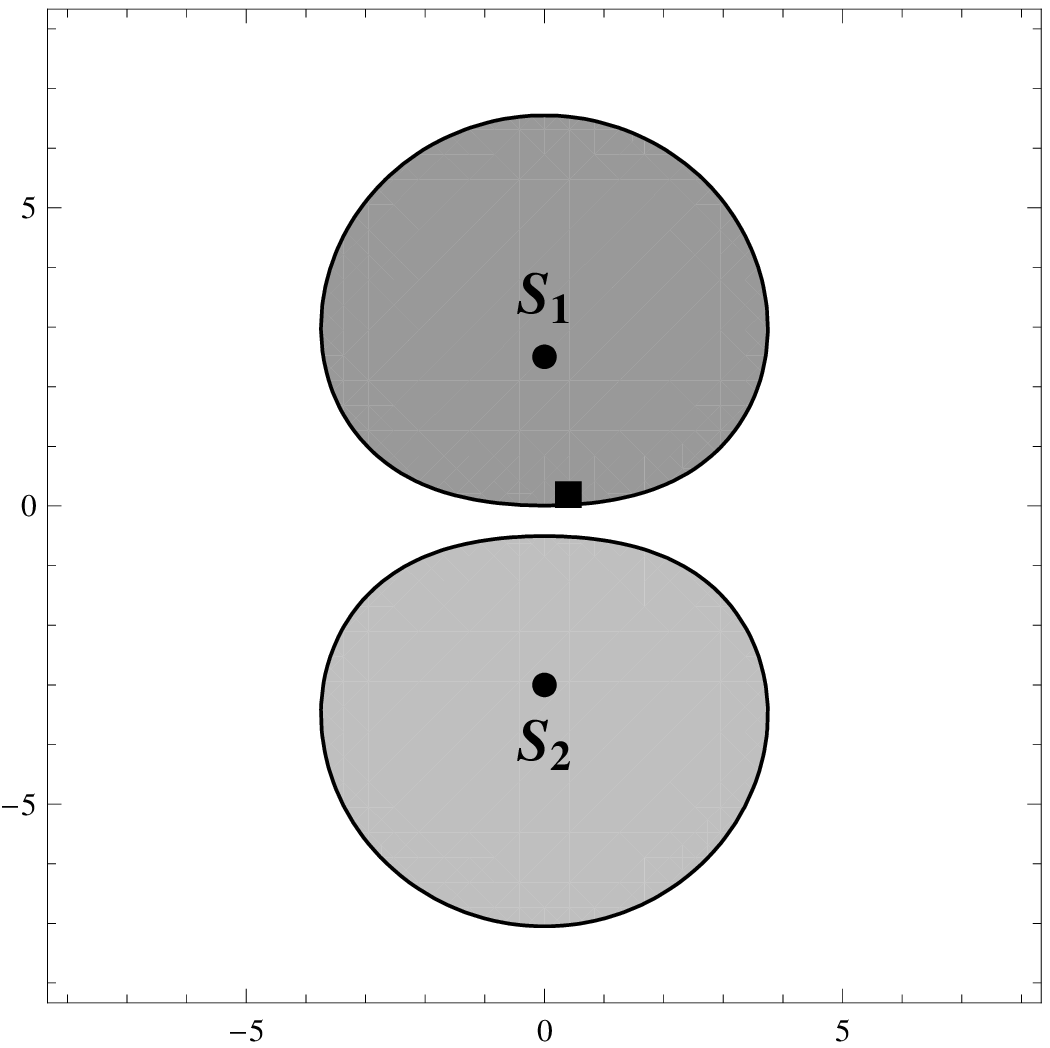}
\end{minipage}\\
\small (A) - UDG Diagram & \small (B) - SINR Diagram\\
\begin{minipage}{2.2in}
\includegraphics[width=2.2in]{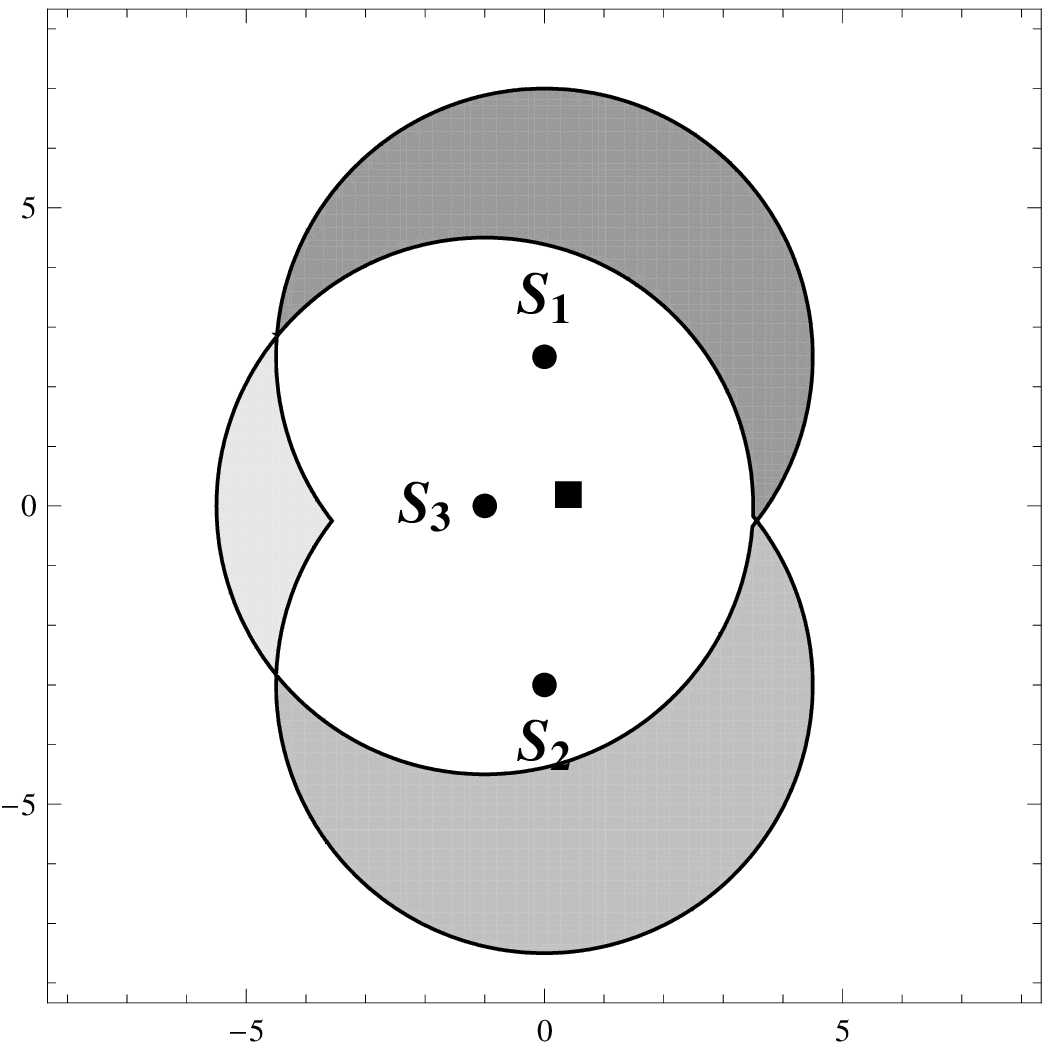}
\end{minipage}&
\begin{minipage}{2.2in}
\includegraphics[width=2.2in]{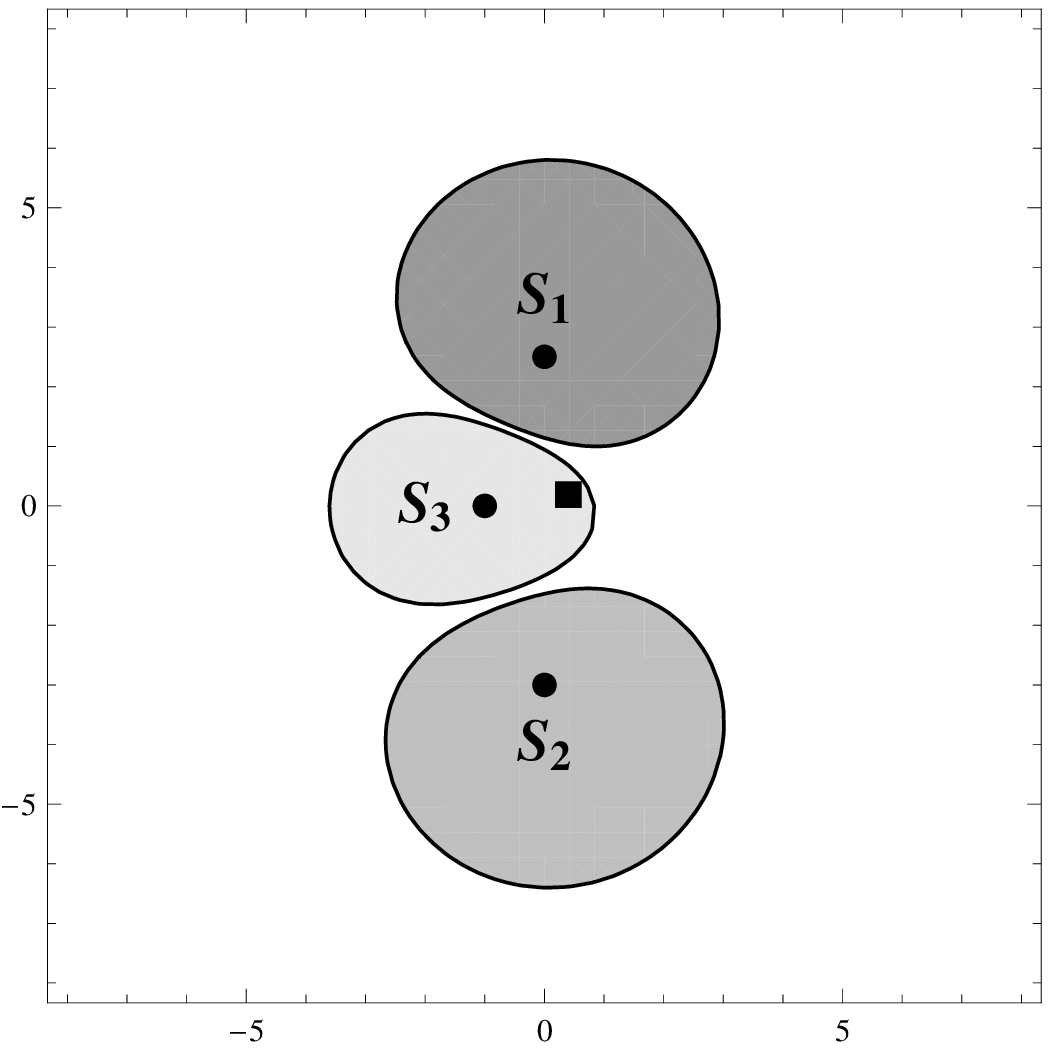}
\end{minipage} \\
\small (C) - UDG Diagram & \small (D) - SINR Diagram \\
\begin{minipage}{2.2in}
\includegraphics[width=2.2in]{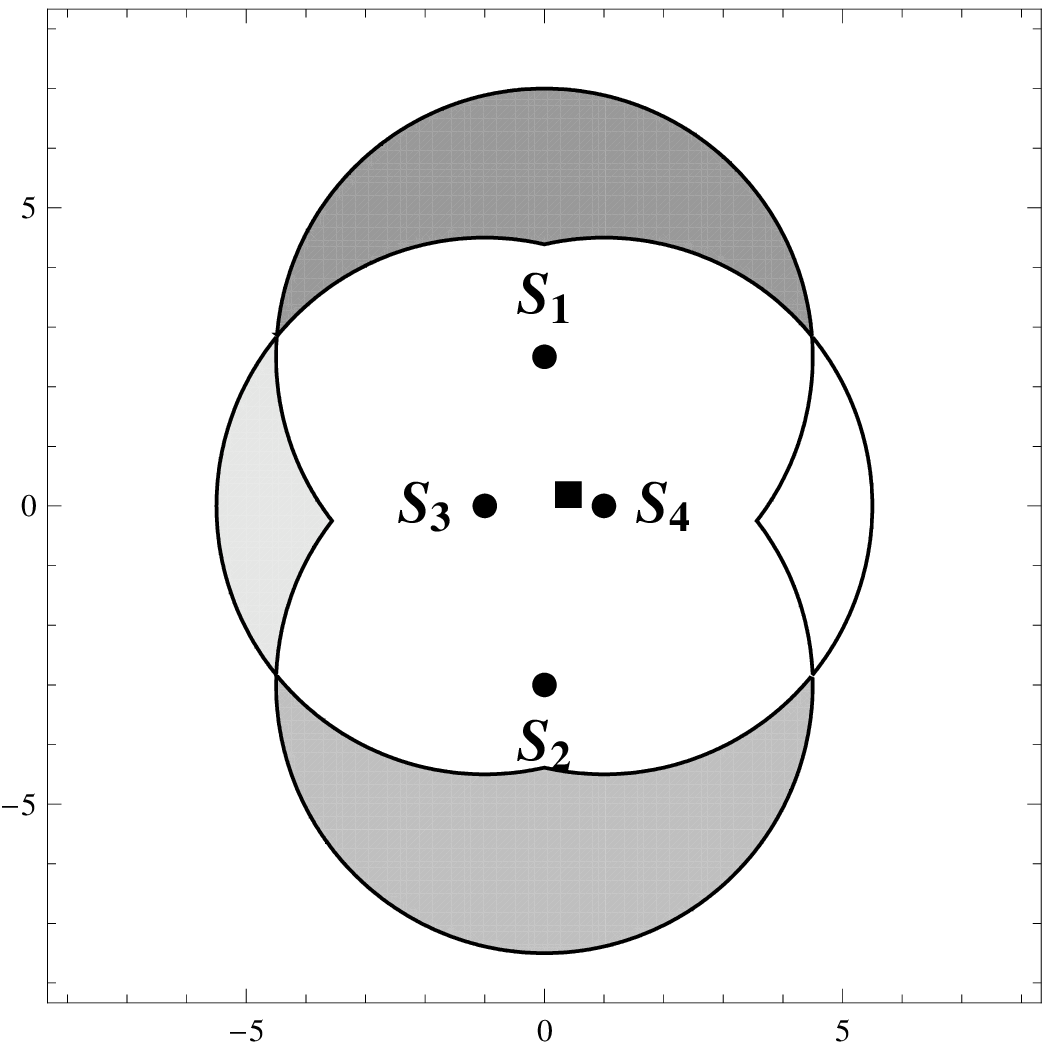}
\end{minipage}&
\begin{minipage}{2.2in}
\includegraphics[width=2.2in]{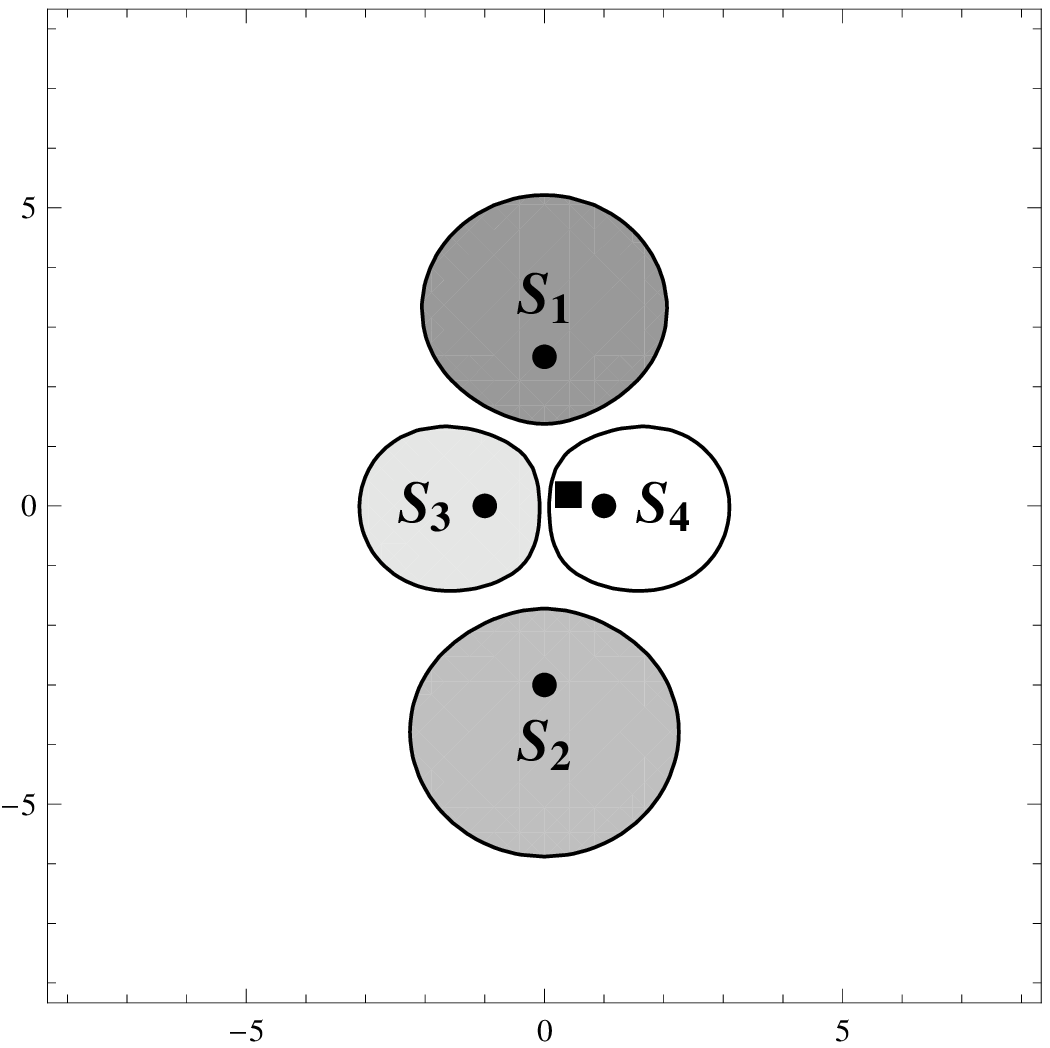}
\end{minipage}\\
\small (E) - UDG Diagram & \small (F) - SINR Diagram\\
\end{tabular}
\caption{Reception zones in the UDG and SINR models.
Steps 2-4 add stations
$s_2, s_3, s_4$, one at a time. 
(A)-(B): adding $s_2$.
(C)-(D): adding $s_3$
(E)-(F): adding $s_4$}
\label{fig:sinr_udg_2}
\end{figure}
} 
\LongVersion 
\FigureSinrUdgTwo{}
\LongVersionEnd 

\Comment{
An additional difference between the models that can be illustrated
by \Figure{}~\ref{fig:sinr_udg_2} concerns frequency allocation.
Suppose we wish to cover the maximum possible area by the four stations,
namely, the union of the individual reception zones 
(reachable by each station when transmitting alone), 
using the minimum number of frequencies.
This can be done in the UDG model using three frequencies, 
letting $s_1$ and $s_2$ use the same frequency and $s_3$ and $s_4$ use 
two additional frequencies.
In contrast, in the SINR model one must use 
four distinct frequencies.
In general, in the presence of $n$ stations, 
it is known that $O(\log n)$ frequencies are sufficient for attaining maximum 
in the UDG model, while in the SINR model, $n$ frequencies are necessary.   
} 

In summary, while the existing body of literature on models and algorithms 
for wireless networks represents a significant base containing
a rich collection of tools and techniques, the state of affairs 
described above leaves us in the unfortunate situation 
where the more practical graph-based models (such as the UDG model) 
are not sufficiently accurate, and the more accurate SINR model is not
well-understood and therefore difficult for protocol designers.
Hence obtaining a better understanding of the SINR model, and consequently
bridging the gap between this physical model and the graph based models
may have potentially significant (theoretical and practical) implications.
This goal is the central motivation behind the current paper.

\ShortVersion 
\dnsparagraph{Related work.}
\ShortVersionEnd 
\LongVersion 
\subsection{Related work}
\LongVersionEnd 
\ShortVersion 
Some recent studies aim at achieving a better understanding of the SINR model.
In their seminal work~\cite{GuKu00}, Gupta and Kumar analyzed the capacity of
wireless networks in the physical and protocol models.
Moscibroda~\cite{Mo07} analyzed the worst-case capacity of wireless networks,
making no assumptions on the deployment of nodes in the plane, as opposed to
almost all the previous work on this problem.
\ShortVersionEnd 
\LongVersion 
Some recent studies aim at achieving a better understanding of the SINR model.
In particular, in their seminal work~\cite{GuKu00}, Gupta and Kumar analyzed
the capacity of wireless networks in the physical and protocol models.
Moscibroda~\cite{Mo07} analyzed the worst-case capacity of wireless networks,
making no assumptions on the deployment of nodes in the plane,
as opposed to almost all the previous work on wireless network capacity.
\LongVersionEnd 

Thought provoking experimental results presented in \cite{MWW06} show that 
even basic wireless stations can achieve communication patterns that 
are impossible in graph-based models. 
Moreover, the paper presents certain situations in which it is possible 
to apply routing / transport schemes that may break the theoretical 
throughput limits of any protocol which obeys the laws 
of a graph-based model. 
\Comment{
For example, Figure \ref{fig:non-convex-SINR} (which extends to two dimensions 
a 1-dimensional example of \cite{MWW06}) illustrate a possible setting 
in the SINR model, in which $p_1$ can hear $s_1$ and $p_2$ can hear $s_2$ 
simultaneously.
In contrast, in the UDG model there can be no assignment
of transmission powers that will achieve this effect.
} 

Another line of research, in which known results from the UDG model 
are analyzed under the SINR model, includes \cite{MoWaZo06}, 
which studies the problem of topology control in the SINR model, 
and \cite{GoOsWa07}, where impossibility results were proven in the SINR model 
for scheduling.

\Comment{
Transmission scheduling (so as to prevent collisions) is a central issue 
in wireless communication theory.
The scheduling complexity of basic network structures, namely 
strongly connected networks, is studied in \cite{MW06,MoWaZo06}. 
It is shown that adjusting the transmission power gives 
an exponential advantage over uniform or linear power assignment schemes. 
This gives an interesting complement to the more pessimistic bounds 
for the capacity in wireless networks \cite{GuKu00}. 
A measure called disturbance, capturing the intrinsic difficulty 
of finding a short schedule for a problem instance, is defined in \cite{MOW07},
along with an algorithm that achieves provably efficient performance in any 
network and request setting that exhibits a low disturbance. 
For the special case of many-to-one communication with data aggregation 
in relaying nodes, a scaling law describing the achievable rate 
in arbitrarily deployed sensor networks is derived in \cite{Mo07}. 
It is shown that for a large number 
of aggregation functions, a sustainable rate of $1/\log^2 n$ can be achieved.

In the context of routing, the problem of constructing end-to-end schedules 
for a given set of routing requests such that the delay is minimized
is studied in \cite{CKMPS07}.
That is, each node is assigned 
a distinct power level, the paths for all requests are determined, 
and all message transmissions are scheduled to guarantee 
successful reception in the SINR model. 
In this setting, \cite{CKMPS07} presents a polynomial-time algorithm 
with provable worst-case performance for the problem. 
} 

More elaborate graph-based models may employ two separate graphs, 
a connectivity graph $G_c=(S,E_c)$ and an interference graph $G_i=(S,E_i)$,
such that a station $s$ will successfully receive a message transmitted 
by a station $s'$ if and only if $s$ and $s'$ are neighbors in the 
connectivity graph $G_c$ and $s$ does not have a concurrently transmitting 
neighbor in the interference graph $G_i$.
Protocol designers often consider special cases of this more general model. 
For example, it is sometimes assumed that $G_i$ is $G_c$ 
augmented with all edges between 2-hop neighbors in $G_c$.
Similarly, a variant of the UDG model handling transmissions and interference
separately, named the {\em Quasi Unit Disk Graph (Q-UDG)} model, 
was introduced in \cite{KuWaZo03}.
In this model, two concentric circles are associated with each station, 
the smaller representing its reception zone and the larger representing
its area of interference.
An alternative interference model, also based on the UDG model, 
is proposed in \cite{RSWZ05}.

\ShortVersion 
\dnsparagraph{Our results.}
\ShortVersionEnd 
\LongVersion 
\subsection{Our results}
\label{section:OurResults}
\LongVersionEnd 
\ShortVersion 
A fundamental issue in wireless network modeling involves characterizing the
reception zones of the stations and constructing the reception diagram.
The current paper aims at gaining a better understanding of this issue in the
SINR model,
\ShortVersionEnd 
\LongVersion 
As mentioned earlier, a fundamental issue in wireless network modeling
involves characterizing the reception zones of the stations
and constructing the reception diagram.
The current paper aims at gaining a better understanding of this issue 
in the SINR model, 
\LongVersionEnd 
and as a consequence, deriving some algorithmic results. 
In particular, we consider the structure of 
\LongVersion 
reception zones in 
\LongVersionEnd 
SINR diagrams corresponding to uniform power networks with path-loss
parameter \( \alpha = 2 \) and examine two specific properties of interest,
namely, the {\em convexity} and {\em fatness}\footnote{
The notion of fatness has received a number of non-equivalent technical
definitions, all aiming at capturing the same intuition, namely, absence of
long, skinny or twisted parts.
In this paper we say that the reception zone of station \(\Station_i\) is
\emph{fat} if the ratio between the radii of the smallest ball centered at
\(\Station_i\) that completely contains the zone and the largest ball centered
at \(\Station_i\) that is completely contained by it is bounded by some
constant.
Refer to
\LongVersion 
Section~\ref{section:GeometricNotions}
\LongVersionEnd 
\ShortVersion 
Section~\ref{section:Preliminaries}
\ShortVersionEnd 
for a formal definition.
} of the reception zones.
Apart from their theoretical interest, these properties are also of
considerable practical significance, as obviously, having reception zones
that are non-convex, or whose shape is arbitrarily skewed, twisted or skinny, 
might complicate the development of protocols for various design 
and communication tasks. 

Our first result, which turns out to be surprisingly less trivial than we may
have expected, is cast in Theorem~\ref{gtheorem:Convexity}.
This theorem is proved in
\LongVersion 
Section~\ref{section:Convex}
\LongVersionEnd 
\ShortVersion 
Section~\ref{section:ConvexFatness}
\ShortVersionEnd 
by a complex analysis of the polynomials defining the reception zones, based
on combining several observations with Sturm's condition for counting real
roots.

\begin{GlobalTheorem} \label{gtheorem:Convexity}
The reception zones in an SINR diagram of a uniform power network with
path-loss parameter \( \alpha = 2 \) and reception threshold \( \beta > 1
\) are convex.
\end{GlobalTheorem}

Note that our convexity proof still holds when \( \beta = 1 \).
In contrast, when \( \beta < 1 \), the reception zones of a uniform power
network are not necessarily convex.
This phenomenon is illustrated in (the numerically generated)
\Figure{}~\ref{figure:NonConvex}.
\newcommand{\FigureNonConvex}{
\begin{figure}
\centering
\LongVersion 
\includegraphics[width=0.4\textwidth]{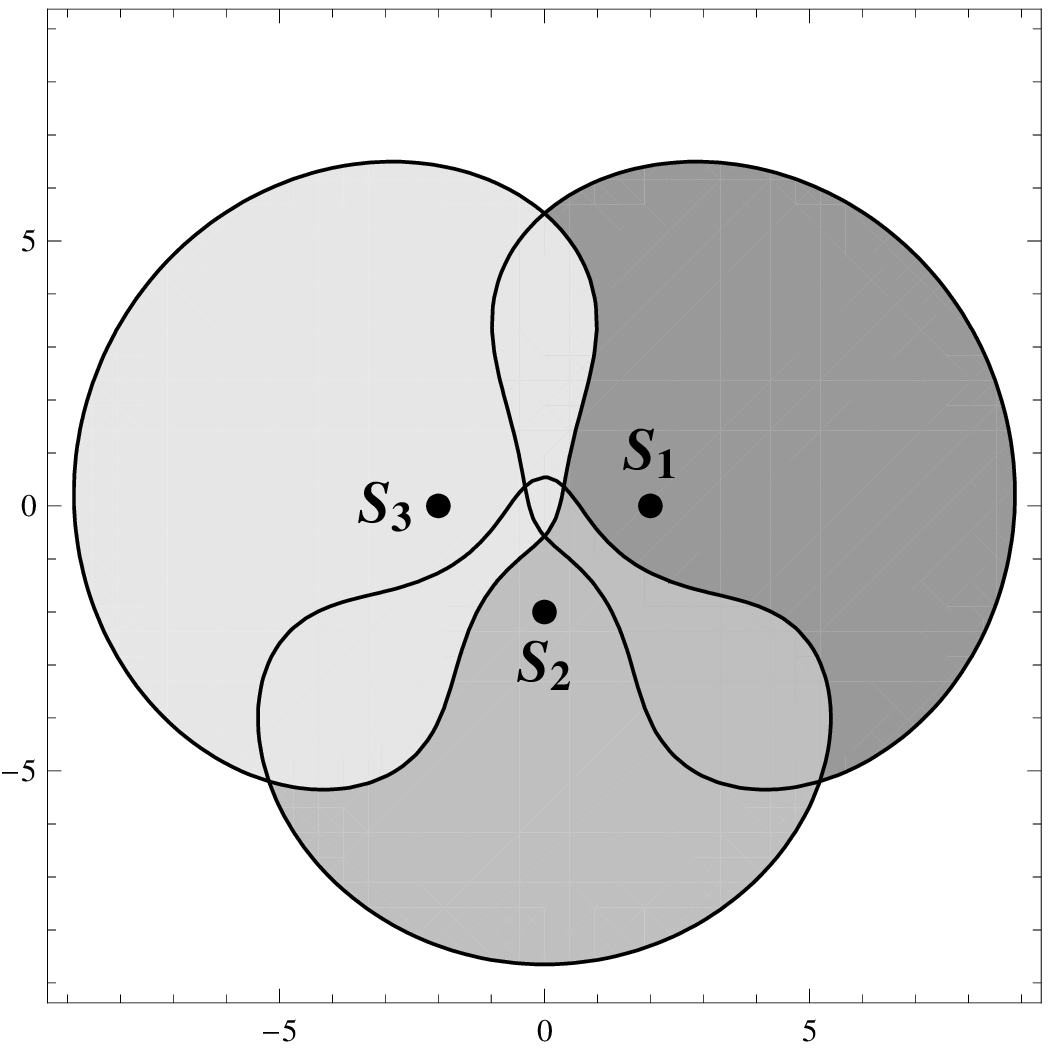}
\LongVersionEnd 
\ShortVersion 
\includegraphics[width=0.22\textwidth]{figs/non-convex}
\ShortVersionEnd 
\caption{\label{figure:NonConvex}
A uniform power network with path-loss parameter \( \alpha = 2 \), reception
threshold \( \beta = 0.3 \), and background noise \( \Noise = 0.05 \).
The reception zones are clearly non-convex.
}
\end{figure}
} 
\LongVersion 
\FigureNonConvex{}
\LongVersionEnd 
\ShortVersion 
We then establish an additional attractive property of the reception zones
(also proved in Section~\ref{section:ConvexFatness}), which in a certain sense
lends support to the model of {\em Quasi Unit Disk Graphs} suggested by Kuhn
et al. in~\cite{KuWaZo03}.
\ShortVersionEnd 
\LongVersion 
We then establish an additional attractive property of the reception zones.
\LongVersionEnd 

\begin{GlobalTheorem} \label{gtheorem:Fatness}
The reception zones in an SINR diagram of a uniform power network with
path-loss parameter \( \alpha = 2 \) and reception threshold \( \beta > 1 \)
are fat.
\end{GlobalTheorem}

\LongVersion 
Theorem~\ref{gtheorem:Fatness} is proved in Section~\ref{section:Fatness}.
In a certain sense, this result lends support to the model of {\em Quasi Unit
Disk Graphs} suggested by Kuhn et al. in~\cite{KuWaZo03}.
\LongVersionEnd 

Armed with this characterization of the reception zones, we turn to a basic
algorithmic task closely related to SINR diagrams, namely, answering {\em
point location queries}.
We address the following natural question: given a point in the
plane, which reception zone contains this point (if any)?
For UDG, this problem can be dealt with using known techniques
(cf. \cite{AE99,AHL90}).
For arbitrary (non-unit) disk graphs, the problem is already harder, 
as the direct reduction to the technique of \cite{AHL90} no longer works.
In the SINR model the problem becomes even harder. 
A naive solution will require computing the signal to interference \& noise 
ratio for each station, yielding time $O(n^2)$.
A more efficient ($O(n)$ time) 
querying algorithm can be based, for example, on the observation that there is
a unique candidate \( \Station_i \in S \) whose transmission may be received
at $p$, namely, the one whose Voronoi cell contains $p$ in the Voronoi diagram
defined for $S$.
However, it is not known if a sublinear query time can be obtained.
This problem can in fact be thought of as part of a more general one, namely, 
point location over a general set of objects defined by polynomials
and satisfying some ``niceness'' properties.
Previous work on the problem dealt with Tarski cells, namely, objects whose
boundaries are defined by a constant number of polynomials of constant degree
\cite{CEGS91,AE99}.
In contrast, SINR reception zones are defined by polynomials of degree
proportional to \(n\).

Consider the SINR diagram of a uniform power network with path-loss parameter
\( \alpha = 2 \) and reception threshold \( \beta > 1 \) and fix some
performance parameter \( 0 < \epsilon < 1 \).
The following theorem is proved in
Section~\ref{section:ApproximatePointLocationQueries}
(refer to \Figure{}~\ref{figure:ApproxQuery} for illustration).

\begin{GlobalTheorem} \label{gtheorem:ApproximatePointLocationQueries}
A data structure \DataStructure{} of size \( O (n \epsilon^{-1}) \) is
constructed in \( O (n^{3} \epsilon^{-1}) \) preprocessing time.
This data structure essentially partitions the Euclidean plane into disjoint
zones \( \Reals^{2} = \bigcup_{i = 0}^{n - 1} \ReceptionZone_{i}^{+} \cup
\ReceptionZone^{-} \cup \bigcup_{i = 0}^{n - 1} \ReceptionZone_{i}^{?} \) such
that for every \( 0 \leq i \leq n - 1 \): \\
(1) \( \ReceptionZone_{i}^{+} \subseteq \ReceptionZone_i \); \\
(2) \( \ReceptionZone^{-} \cap \ReceptionZone_i = \emptyset \); and \\
(3) \(\ReceptionZone_{i}^{?}\) is bounded and its area is at most an
\(\epsilon\)-fraction of the area of \(\ReceptionZone_i\). \\
Given a query point \( p \in \Reals^{2} \), \DataStructure{} identifies
the zone in \( \{ \ReceptionZone_{i}^{+} \}_{i = 0}^{n - 1} \cup \{
\ReceptionZone^{-} \} \cup \{ \ReceptionZone_{i}^{?} \}_{i = 0}^{n - 1} \) to
which \(p\) belongs, in time \( O (\log n) \).
\end{GlobalTheorem}

\newcommand{\FigureApproxQuery}{
\begin{figure}
\begin{center}
\LongVersion 
\includegraphics[width=0.5\textwidth]{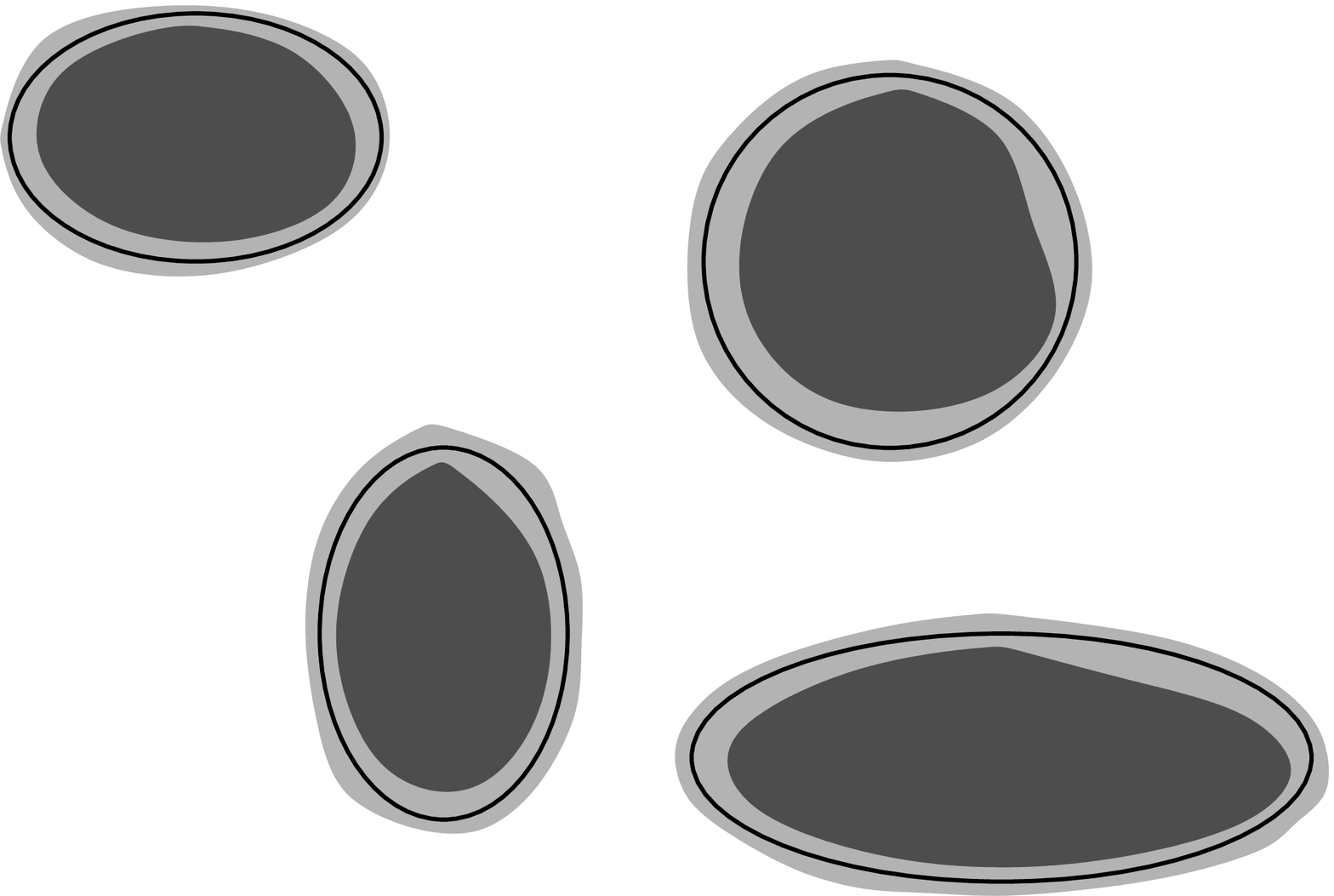}
\LongVersionEnd 
\ShortVersion 
\includegraphics[width=0.4\textwidth]{figs/approx-query}
\ShortVersionEnd 
\end{center}
\caption{\label{figure:ApproxQuery}
The reception zones \(\ReceptionZone_i\) (enclosed by the bold lines) and the
partition of the plane to disjoint zones \(\ReceptionZone_{i}^{+}\) (dark
gray), \(\ReceptionZone_{i}^{?}\) (light gray), and \(\ReceptionZone^{-}\)
(the remaining white).
}
\end{figure}
}
\LongVersion 
\FigureApproxQuery{}
\LongVersionEnd 

\ShortVersion 
\dnsparagraph{Open problems.}
\ShortVersionEnd 
\LongVersion 
\subsection{Open Problems}
\label{section:OpenProblems}
\LongVersionEnd 
Various extensions of our original setting may be considered.
For instance, it may be of interest to study SINR diagrams in $d>2$
dimensions, or for path-loss parameter $\alpha>2$.

Our results concern wireless networks with uniform power transmissions.
General wireless networks are harder to deal with.
For instance, the point location problem becomes considerably more difficult
when different stations are allowed to use different transmission energy,
since in this case, the appropriate graph-based model is no longer a unit-disk
graph but a ({\em directed}) general disk graph, based on disks of arbitrary
radii.
An even more interesting case is the {\em variable power} setting, where the
stations can adjust their transmission energy levels from time to time.

The problems discussed above become harder in a dynamic setting, and in
particular, if we assume the stations are mobile, and extending our approach
to the dynamic and mobile settings are the natural next steps.

\section{Preliminaries}
\label{section:Preliminaries}

\ShortVersion 
\dnsparagraph{Geometric notions.}
\ShortVersionEnd 
\LongVersion 
\subsection{Geometric notions}
\label{section:GeometricNotions}
\LongVersionEnd 
\LongVersion 
We consider the Euclidean plane \(\Reals^2\).
The \emph{distance} from point \(p\) to point \(q\) is denoted by \( \dist{p,
q} = \dist{q, p} = \| q - p \| \).
A \emph{ball} of radius \(r\) centered at point \(p\) is the set of all points
at distance at most \(r\) from \(p\), denoted by \( \Ball(p, r) = \{ q \in
\Reals^2 \mid \dist{p, q} \leq r \} \).
We say that point \( p \in \Reals^2 \) is \emph{internal} to the point set
\(P\) if there exists some \( \epsilon > 0 \) such that \( \Ball(p, \epsilon)
\subseteq P \).
\LongVersionEnd 
\ShortVersion 
In the Euclidean plane \(\Reals^2\), the \emph{distance} between points \(p,
q\) is denoted by \( \dist{p, q} = \dist{q, p} = \| q - p \| \).
A \emph{ball} of radius \(r\) centered at point \(p\) is the set of all points
at distance at most \(r\) from \(p\), denoted by \( \Ball(p, r) = \{ q \in
\Reals^2 \mid \dist{p, q} \leq r \} \).
Point \( p \in \Reals^2 \) is \emph{internal} to the point set
\(P\) if there exists some \( \epsilon > 0 \) such that \( \Ball(p, \epsilon)
\subseteq P \).
(For some fundamental notions in point set topology see
Appendix~\ref{appendix:TopologyOverview}.)
\ShortVersionEnd 

\newcommand{\TopologyOverview}{
Consider some point set \(P\).
\(P\) is said to be an \emph{open set} if all points \( p \in P \) are
internal points.
\(P\) is said to be a \emph{closed set} if the complement of \(P\) is an open
set.
If there exists some real \(r\) such that \( \dist{p, q} \leq r \) for every
two points \( p, q \in P \), then \(P\) is said to be \emph{bounded}.
A \emph{compact} set is a set which is both closed and bounded.
The \emph{closure} of \(P\) is the smallest closed set containing \(P\).
The \emph{boundary} of \(P\), denoted by \(\Boundary P\), is the intersection
of the closure of \(P\) and the closure of its complement.
A \emph{connected set} is a point set that cannot be partitioned to two
non-empty subsets such that each of the subsets has no point in common with
the closure of the other.
} 
\LongVersion 
\TopologyOverview{}
\LongVersionEnd 
\LongVersion 
We refer to the closure of an open bounded connected set as a \emph{thick}
set.
By definition, every thick set is compact.
\LongVersionEnd 

A point set \(P\) is said to be \emph{convex} if the segment
\(\Segment{p}{q}\) is contained in \(P\) for every two points \( p, q \in P
\).
The point set \(P\) is said to be \emph{star-shaped} \cite{BCKO08} with
respect to point \( p \in P \) if the segment \(\Segment{p}{q}\) is contained
in \(P\) for every point \( q \in P \).
\ShortVersion 
Note that 
\ShortVersionEnd 
\LongVersion 
Clearly, convexity is stronger than the star-shape property in the sense that
\LongVersionEnd 
a convex point set \(P\) is star-shaped with respect to any point \( p \in P
\);
the converse is not necessarily true.
\ShortVersion 
We refer to the closure of an open bounded connected set as a \emph{thick}
set.
By definition, every thick set is compact.
We have the following.
\ShortVersionEnd 
\LongVersion 
For thick sets we have the following necessary and sufficient condition for
convexity.
\LongVersionEnd 

\begin{lemma} \label{lemma:ThickSetsConvexity}
A thick set \(P\) is convex if and only if every line intersects \( \Boundary
P \) at most twice.
\end{lemma}

We frequently use the term \emph{zone} to describe a point set with some
``niceness'' properties.
Unless stated otherwise, a zone refers to the union of an open connected set
and some subset of its boundary.
(A thick set is a special case of a zone.)
It may also refer to a single point or to the finite union of zones.
Given some bounded zone \(Z\), we denote the \emph{area} and
\emph{perimeter} of \(Z\) (assuming that they are well defined) by
\(\Area(Z)\) and \(\Perimeter(Z)\), respectively.
Let \(Z\) be a non-empty bounded zone and let \(p\) be some internal point
of \(Z\).
Denote
\[
\SmallRadius(p, Z) = \sup \{ r > 0 \mid Z \supseteq \Ball(p, r) \} ~ ,
\qquad
\LargeRadius(p, Z) = \inf \{ r > 0 \mid Z \subseteq \Ball(p, r) \} ~ ,
\]
and define the \emph{fatness parameter} of \(Z\) with respect to \(p\) to be
the ratio of \(\LargeRadius(p, Z)\) and \(\SmallRadius(p, Z)\), denoted by
\( \FatnessParameter(p, Z) = \LargeRadius(p, Z) / \SmallRadius(p, Z) \).
(See \Figure{}~\ref{figure:Fatness}.)
\ShortVersion 
The zone \(Z\) is \emph{fat} w.r.t. \(p\) if \(\FatnessParameter(p, Z) \leq c
\) for some constant $c>0$.
\ShortVersionEnd 
\LongVersion 
The zone \(Z\) is said to be \emph{fat} with respect to \(p\) if
\(\FatnessParameter(p, Z)\) is bounded by some constant.
\LongVersionEnd 

\newcommand{\FigureFatness}{
\begin{figure}
\begin{center}
\LongVersion 
\includegraphics[width=0.35\textwidth]{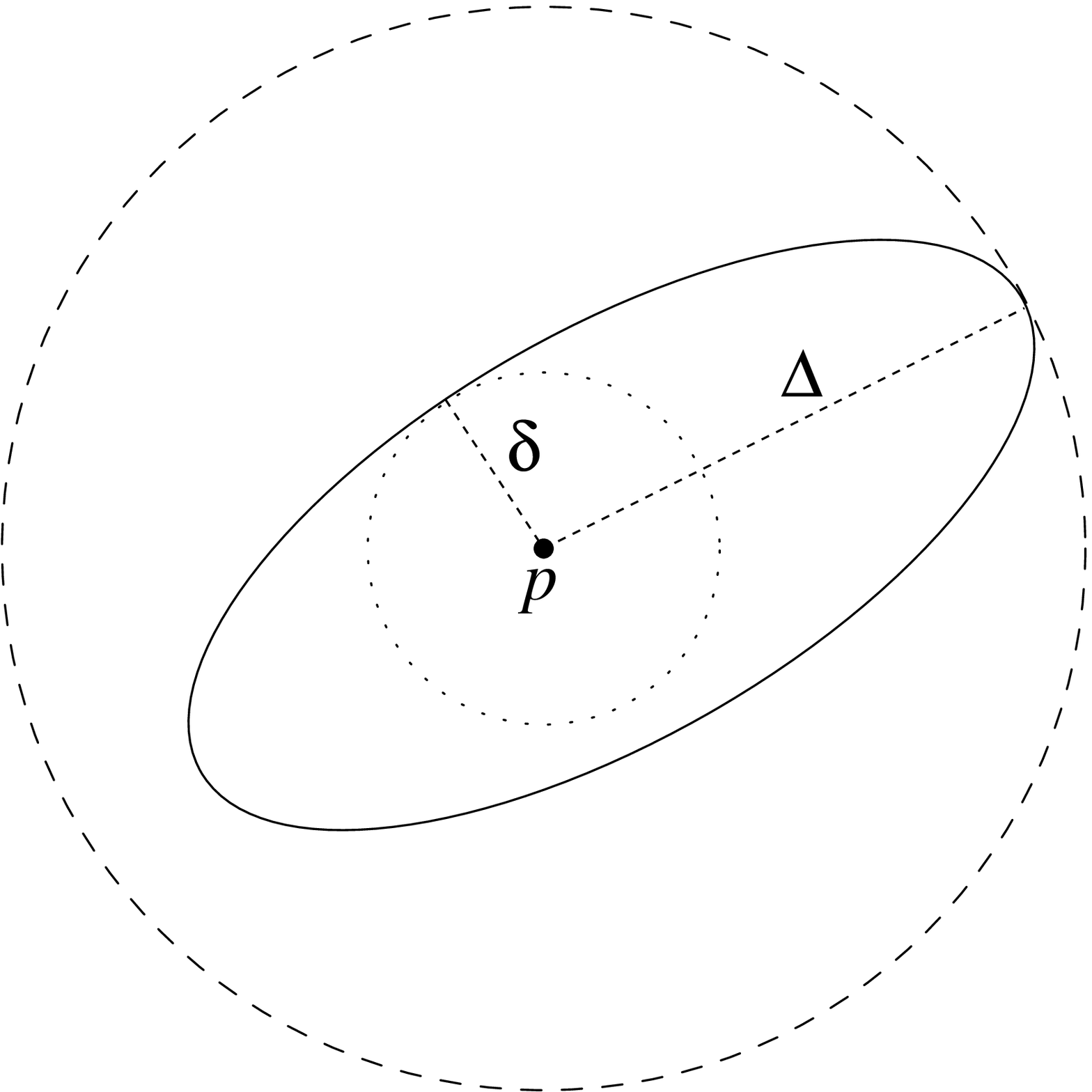}
\LongVersionEnd 
\ShortVersion 
\includegraphics[width=0.25\textwidth]{figs/fatness}
\ShortVersionEnd 
\end{center}
\caption{\label{figure:Fatness}
The zone \(Z\) (enclosed by the solid line) with the ball defining
\(\SmallRadius(p, Z)\) (dotted line) and the ball defining \(\LargeRadius(p,
Z)\) (dashed line).
}
\end{figure}
} 
\LongVersion 
\FigureFatness{}
\LongVersionEnd 

\ShortVersion 
The \emph{separation line} of two points \(p_1\) and \(p_2\) in the plane
is the set of points \( \{ q \mid \dist{p_1, q} = \dist{p_2, q} \} \).
\ShortVersionEnd 
\LongVersion 
Consider some two points \( p_1, p_2 \) in the plane.
The set of points \(q\) that satisfy \( \dist{p_1, q} = \dist{p_2, q} \) form
a line referred to as the \emph{separation line} of \(p_1\) and \(p_2\).
\LongVersionEnd 

\ShortVersion 
\begin{AvoidOverfullParagraph}
\ShortVersionEnd 
\ShortVersion 
\dnsparagraph{Wireless networks.}
\ShortVersionEnd 
\LongVersion 
\subsection{Wireless networks}
\label{section:WirelessNetworks}
\LongVersionEnd 
We consider a wireless network \( \cA = \langle S, \Power, \Noise,
\beta \rangle \), where \( S = \{ \Station_0, \Station_1, \dots, \Station_{n -
1} \} \) is a set of transmitting \emph{radio stations} embedded in the
Euclidean plane, \(\Power\) is an assignment of a positive real
\emph{transmitting power} \(\Power_i\) to each station \(\Station_i\), \(
\Noise \geq 0 \) is the \emph{background noise}, and \( \beta \geq 1 \) is a
constant that serves as the \emph{reception threshold} (will be explained
soon).
For the sake of notational simplicity, \(\Station_i\) also refers to the point
\((a_i, b_i)\) in the plane where the station \(\Station_i\) resides.
The network is assumed to contain at least two stations, i.e., \( n \geq 2 \).
We say that \(\cA\) is a
\LongVersion 
\emph{uniform power network}
\LongVersionEnd 
\ShortVersion 
\emph{uniform power network (UPN)}
\ShortVersionEnd 
if \( \psi = \bar{1} \), namely, if \( \Power_i = 1 \) for every \(i\).
\ShortVersion 
\end{AvoidOverfullParagraph}
\ShortVersionEnd 

The \emph{energy} of station \(\Station_i\) at point \( p \neq \Station_i \)
is defined to be \( \Energy_{\cA}(\Station_i, p) = \Power_i \cdot
\dist{\Station_i, p}^{-2} \).
The \emph{energy} of a station set \( T \subseteq S \) at point \(p\), where
\( p \neq \Station_i \) for every \( i \in T \), is defined to be \(
\Energy_{\cA}(T, p) = \sum_{i \in T} \Energy_{\cA}(\Station_i, p) \).
Fix some station \(\Station_i\) and consider some point \( p \notin S \).
We define the \emph{interference} to \(\Station_i\) at point \(p\) 
to be the energies at \(p\) of all stations other than \(\Station_i\), denoted
\( \Interference_{\cA}(\Station_i, p) = \Energy_{\cA}(S - \{\Station_i\}, p)
\).
The \emph{signal to interference \& noise ratio (SINR)} of \(\Station_i\)
at point \(p\) is defined as
\begin{equation}
\label{eq:Def-SINR}
\SINR_{\cA}(\Station_i, p)
~ = ~ \frac{\Energy_{\cA}(\Station_i, p)}{\Interference_{\cA}(\Station_i, p) +
\Noise}
~ = ~ \frac{\Power_i \cdot \dist{\Station_i, p}^{-2}}{\sum_{j \neq i} \Power_j
\cdot \dist{\Station_j, p}^{-2} + \Noise} ~ .
\end{equation}
Observe that \(\SINR_{\cA}(\Station_i, p)\) is always positive since the
transmitting powers and the distances of the stations from \(p\) are always
positive and the background noise is non-negative.
When the network \(\cA\) is clear from the context, we may omit it and write
simply \(\Energy(\Station_i, p)\), \(\Interference(\Station_i, p)\), and
\(\SINR(\Station_i, p)\).

The fundamental rule of the SINR model is that the transmission of station
\(\Station_i\) is received correctly at point \( p \notin S \) if and only if
its SINR at \(p\) is not smaller than the reception threshold of the network,
i.e., \( \SINR(\Station_i, p) \geq \beta \).
If this is the case, then we say that \(\Station_i\) is \emph{heard} at \(p\).
We refer to the set of points that hear station \(\Station_i\) as the
\emph{reception zone} of \(\Station_i\), defined 
\ShortVersion 
(somewhat tediously)
\ShortVersionEnd 
as
\[
\ReceptionZone_i
= \{ p \in \Reals^2 - S \mid \SINR(\Station_i, p) \geq \beta \} \cup
\{\Station_i\} ~ .
\]
\ShortVersion 
(This definition is necessary as \(\SINR(\Station_i, \cdot)\)
is not defined at any point in \(S\) and in particular, at \(\Station_i\)
itself.)
\ShortVersionEnd 
\LongVersion 
This admittedly tedious definition is necessary as \(\SINR(\Station_i, \cdot)\)
is not defined at any point in \(S\) and in particular, at \(\Station_i\)
itself.
\LongVersionEnd 

Consider station \(\Station_0\) and an arbitrary point \( p = (x, y) \in
\Reals^2 \).
By rearranging the expression in (\ref{eq:Def-SINR}), we correlate the
fundamental rule of the SINR model to the \(2\)-variate polynomial
\(\HearPoly(x, y)\)
\ShortVersion 
s.t. \(\Station_0\) is heard at \(p\) iff
\ShortVersionEnd 
\LongVersion 
so that \(\Station_0\) is heard at \(p\) if and only if
\LongVersionEnd 
\begin{align*}
\HearPoly(x, y)
= ~ & \beta \left[ \sum_{i > 0} \Power_i \cdot \prod_{j \neq i} \left( (a_j -
x)^2 + (b_j - y)^2 \right) + \Noise \cdot \prod_{i} \left( (a_i - x)^2 + (b_i
- y)^2 \right) \right] \\
& - \Power_0 \cdot \prod_{i > 0} \left( (a_i - x)^2 + (b_i - y)^2 \right) 
~ \leq ~ 0 ~ .
\end{align*}
Note that this condition holds even for points \( p \in S \).
Consequently, we can rewrite \( \ReceptionZone_0 = \{ (x, y) \in \Reals^2 \mid
\HearPoly(x, y) \leq 0 \} \), where the boundary points of
\(\ReceptionZone_0\) are exactly the roots of \(\HearPoly(x, y)\).
In general, the polynomial \(\HearPoly(x, y)\) has degree \( 2 n \); the
degree is \( 2 n - 2 \) if the background noise \( \Noise = 0 \).
This polynomial plays a key role in the analysis carried out in
Section~\ref{section:ThreeStations}.

A \UPN{} \( \cA = \langle S, \bar{1}, \Noise, \beta \rangle \)
is said to be \emph{trivial} if \( |S| = 2 \), \( \Noise = 0 \), and \( \beta
= 1 \).
Note that for \( i = 0, 1 \), the reception zone \(\ReceptionZone_i\) of
station \(\Station_i\) in a trivial \UPN{} is the half-plane
consisting of all points whose distance to \(\Station_i\) is not greater than
their distance to \( \Station_{1 - i} \).
In particular, \(\ReceptionZone_i\) is unbounded.
For non-trivial networks, we have the following observation that relies on the
fact that \(\SINR(\Station_i, \cdot)\) is a continuous function in \( \Reals^2
- S \).

\begin{observation} \label{observation:ReceptionSetBasicProperties}
Let \( \cA = \langle S, \bar{1}, \Noise, \beta \rangle \) be a non-trivial
\UPN{}.
Then the reception zone \(\ReceptionZone_i\) is compact for every \(
\Station_i \in S \).
Moreover, every point in \(\ReceptionZone_i\) is closer to \(\Station_i\) than
it is to any other station in \(S\) (i.e., \(\ReceptionZone_i\) is strictly
contained in the Voronoi cell of \(\Station_i\) in the Voronoi diagram of
\(S\)).
\end{observation}

\LongVersion 
Next, we state a simple but important lemma that will be useful in our later
arguments.
\LongVersionEnd 

\begin{lemma} \label{lemma:Transformation}
Let \( f : \Reals^2 \rightarrow \Reals^2 \) be a mapping consisting of
rotation, translation, and scaling by a factor of \( \sigma > 0 \).
Consider some network \( \cA = \langle S, \Power, \Noise, \beta \rangle \) and
let \( f(\cA) = \langle f(S), \Power, \Noise / \sigma^2, \beta \rangle \),
where \( f(S) = \{ f(\Station_i) \mid \Station_i \in S \} \).
Then for every station \(\Station_i\) and for all points \( p \notin S \), we
have \( \SINR_{\cA}(\Station_i, p) ~ = ~ \SINR_{f(\cA)}(f(\Station_i), f(p))
\).
\end{lemma}

\LongVersion 
\section{Convexity of the reception zones}
\label{section:Convex}
\LongVersionEnd 
\ShortVersion 
\section{Convexity and fatness of the reception zones}
\label{section:ConvexFatness}
\ShortVersionEnd 
In this section we consider the SINR diagram of a uniform power network \( \cA
= \langle S, \bar{1}, \Noise, \beta \rangle \) and establish
\LongVersion 
Theorem~\ref{gtheorem:Convexity}.
\LongVersionEnd 
\ShortVersion 
Theorems \ref{gtheorem:Convexity} and \ref{gtheorem:Fatness}.
\ShortVersionEnd 
As all stations admit the same transmitting power, it is sufficient to focus
on \(\Station_0\) and to prove that the reception zone \(\ReceptionZone_0\) is
\LongVersion 
convex.
\LongVersionEnd 
\ShortVersion 
convex and fat.
The fatness property is established in Section~\ref{subsection:Fatness}.
\ShortVersionEnd 
\LongVersion 
We shall do so by considering
\LongVersionEnd 
\ShortVersion 
For the convexity proof, we consider
\ShortVersionEnd 
some arbitrary two points \( p_1, p_2 \in \Reals^2 \) and
\LongVersion 
arguing
\LongVersionEnd 
\ShortVersion 
argue
\ShortVersionEnd 
that if \(\Station_0\) is heard at \(p_i\) for \( i = 1, 2 \), then
\(\Station_0\) is heard at all points in the segment \( \Segment{p_1}{p_2}
\).
This argument is established in three steps.

First, as a warmup, we prove that \(\ReceptionZone_0\) is star-shaped
with respect to \(\Station_0\).
This proof, presented in Section~\ref{section:StarShape}, establishes our
argument for the case that \(p_1\) and \(p_2\) are colinear with
\(\Station_0\).
Next, we prove that in the absence of a background noise (i.e., \( \Noise = 0
\)), if \( p_i \in \ReceptionZone_0 \) for \( i = 1, 2 \), then \(
\Segment{p_1}{p_2} \subseteq \ReceptionZone_0 \).
This proof, presented in Section~\ref{section:NoNoiseConvex}, relies on the
analysis of a special case of a network consisting of only three stations
which is analyzed in Section~\ref{section:ThreeStations} and in a sense,
constitutes the main technical challenge of this paper.
Finally,
\LongVersion 
in Section~\ref{section:AddingNoise}
\LongVersionEnd 
we reduce the convexity proof of a \UPN{} with \(n\) stations and arbitrary
background noise, to that of a \UPN{} with \( n + 1 \) stations and no
background noise.
\ShortVersion 
This reduction is deferred to Appendix~\ref{appendix:AddingNoise}.
\ShortVersionEnd 
While the analyses in
\LongVersion 
Sections \ref{section:NoNoiseConvex} and \ref{section:AddingNoise}
\LongVersionEnd 
\ShortVersion 
Section~\ref{section:NoNoiseConvex} and Appendix~\ref{appendix:AddingNoise}
\ShortVersionEnd 
are consistent with some ``physical intuition'', the proof presented in
Section~\ref{section:ThreeStations} is based purely on algebraic arguments.

\subsection{Star-shape}
\label{section:StarShape}
\LongVersion 
In this section we consider a \UPN{} \( \cA = \langle S, \bar{1}, \Noise, \beta
\rangle \) and show that the reception zone \(\ReceptionZone_0\) is
star-shaped with respect to the station \(\Station_0\).
In fact, we prove a slightly stronger lemma.
\LongVersionEnd 
\ShortVersion 
In this section we consider a \UPN{} \( \cA = \langle S, \bar{1}, \Noise, \beta
\rangle \) and show that the reception zone \(\ReceptionZone_0\) is
star-shaped w.r.t. the station \(\Station_0\).
In fact, we establish the following slightly stronger lemma (proof is deferred
to Appendix~\ref{appendix:ProofLemmaStarConvex}).
\ShortVersionEnd 

\begin{lemma} \label{lemma:StarConvex}
Consider some point \( p \in \Reals^2 \).
If \( \SINR(\Station_0, p) \geq 1 \), then \( \SINR(\Station_0, q) >
\SINR(\Station_0, p) \) for all internal points \(q\) in the segment
\(\Segment{\Station_0}{p}\).
\end{lemma}
\newcommand{\ProofLemmaStarConvex}{
We consider two disjoint cases.
First, suppose that there exists some station \(\Station_i\), \( i > 0
\), such that \( \Energy(\Station_i, p) = \Energy(\Station_0, p) \).
The assumption that \( \SINR(\Station_0, p) \geq 1 \)
necessitates, by (\ref{eq:Def-SINR}), that \( \Noise = 0 \),  \( n = 2 \) 
(which means that \( i = 1 \)), and \( \SINR(\Station_0, p) = 1 \).
Therefore \( \dist{\Station_0, p} = \dist{\Station_1, p} \) and for all
internal points \(q\) in the segment \(\Segment{\Station_0}{p}\), we have \(
\dist{\Station_0, q} < \dist{\Station_1, q} \).
Thus \( \SINR(\Station_0, q) > 1 \) and the assertion holds.

Now, suppose that \( \Energy(\Station_i, p) < \Energy(\Station_0, p) \) for
every \( i > 0 \), which means that \( \dist{\Station_i, p} >
\dist{\Station_0, p} \) for every \( i > 0 \).
By Lemma~\ref{lemma:Transformation}, we may assume without loss of generality
that \( \Station_0 = (0, 0) \) and \( p = (-1, 0) \).
Consider some station \(\Station_i\), \( i > 0 \).
Note that if \(\Station_i\) is not located on the positive half of the
horizontal axis, then we can relocate it to a new location \(\Station'_i\) on
the positive half of the horizontal axis by rotating it around \(p\) so that \(
\dist{\Station'_i, p} = \dist{\Station_i, p} \) and \( \dist{\Station'_i, q}
\leq \dist{\Station_i, q} \) for all points \( q \in \Segment{\Station_0}{p}\)
(see \Figure{}~\ref{figure:StarShape}).
We can repeat this process with every station \(\Station_i\), \( i > 0 \),
until all stations are located on the positive half of the horizontal
axis without decreasing the interference at any point \( q \in
\Segment{\Station_0}{p} \).
Therefore it is sufficient to establish the assertion under the assumption that
\( \Station_i = (a_i, 0) \), where \( a_i > 0 \), for every \( i > 0 \).

Let \( q = (-x, 0) \) for some \( x \in (0, 1] \).
We can express the SINR function of \(\Station_0\) at \(q\) as 
\LongVersion 
\[
\SINR(\Station_0, q)
~ = ~ \frac{x^{-2}}{\sum_{i > 0} (a_i + x)^{-2} + \Noise} ~ .
\]
\LongVersionEnd 
\ShortVersion 
\( \SINR(\Station_0, q) = x^{-2} / 
\left(\sum_{i > 0} (a_i + x)^{-2} + \Noise \right).\)
\ShortVersionEnd 
In this context, it will be more convenient to consider the reciprocal of the
SINR function,
\LongVersion 
\[
f(x)
= \sum_{i > 0} \left( \frac{x}{a_i + x} \right)^2 + x^2 \cdot \Noise ~ ,
\]
\LongVersionEnd 
\ShortVersion 
\( f(x) = \sum_{i > 0} x^2 / (a_i + x)^2 + x^2 \cdot \Noise ,\)
\ShortVersionEnd 
so that we have to prove that \( f(x) < f(1) \) for all \( x \in (0, 1) \).
The assertion follows since \( \frac{d f(x)}{d x} = 2 x \cdot \sum_{i > 0}
\frac{a_i}{(a_i + x)^3} + 2 x \cdot \Noise \) is positive when \( x \in (0, 1]
\).
} 
\LongVersion 
\begin{proof}
\ProofLemmaStarConvex{}
\end{proof}
\LongVersionEnd 

\newcommand{\FigureStarShape}{
\begin{figure}
\begin{center}
\includegraphics[width=0.5\textwidth]{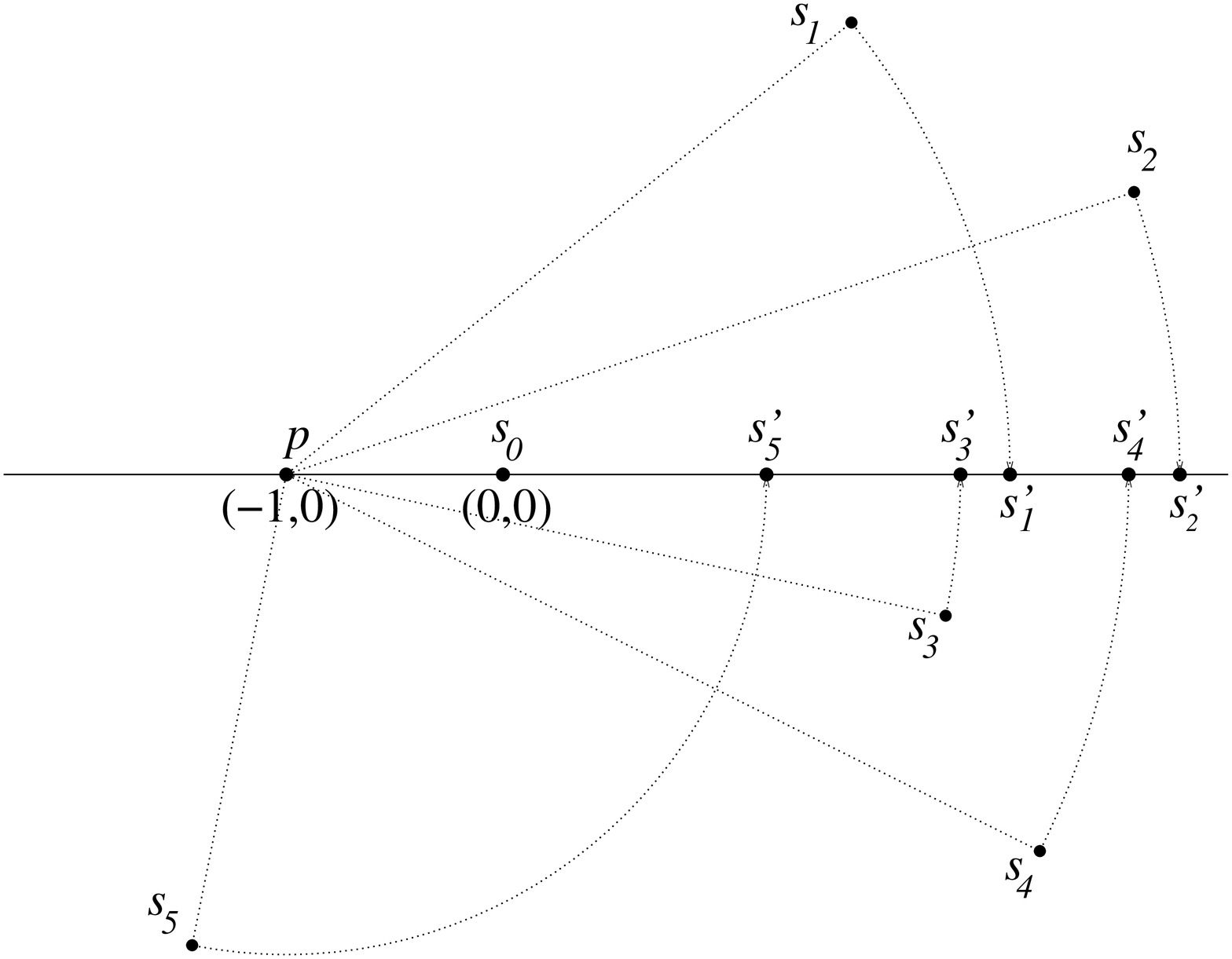}
\end{center}
\caption{\label{figure:StarShape}
Relocating stations \(\Station_i\), \( i > 0 \).
}
\end{figure}
} 
\LongVersion 
\FigureStarShape{}
\LongVersionEnd 

\ShortVersion 
Consider a non-trivial \UPN{} \( \cA = \langle S, \bar{1}, \Noise, \beta
\rangle \) and suppose that \( \Station_0 \neq \Station_j \) for every \( j >
0 \), i.e., the location of \(\Station_0\) is not shared by other stations.
Lemma~\ref{lemma:StarConvex} implies that the point set \( \ReceptionZone'_0 =
\{ p \in \Reals^2 - S \mid \SINR_{\cA}(\Station_0, p) > \beta \} \cup \{
\Station_0 \} \) is star-shaped w.r.t. \(\Station_0\), and in 
particular, connected.
Moreover, since \(\SINR\) is a continuous function in \( \Reals^2 - S \), it
follows that \(\ReceptionZone'_0\) is an open set.
As \(\ReceptionZone_0\) is the closure of \(\ReceptionZone'_0\), we have the
following.
\ShortVersionEnd 
\LongVersion 
Consider a non-trivial \UPN{} \( \cA = \langle S, \bar{1}, \Noise, \beta
\rangle \) and suppose that \( \Station_0 \neq \Station_j \) for every \( j >
0 \), that is, the location of \(\Station_0\) is not shared by any other
station.
Lemma~\ref{lemma:StarConvex} implies that the point set \( \ReceptionZone'_0 =
\{ p \in \Reals^2 - S \mid \SINR_{\cA}(\Station_0, p) > \beta \} \cup \{
\Station_0 \} \) is star-shaped with respect to \(\Station_0\), and in
particular, connected.
Moreover, since \(\SINR\) is a continuous function in \( \Reals^2 - S \), it
follows that \(\ReceptionZone'_0\) is an open set.
As \(\ReceptionZone_0\) is the closure of \(\ReceptionZone'_0\), we have the
following corollary.
\LongVersionEnd 

\begin{corollary} \label{corollary:ThickSet}
\ShortVersion 
In a nontrivial network, if \(\Station_0\)'s location is not shared by other
stations, then \(\ReceptionZone_0\) is a thick set.
\ShortVersionEnd 
\LongVersion 
In a non-trivial network, if the location of \(\Station_0\) is not shared by
any other station, then \(\ReceptionZone_0\) is a thick set.
\LongVersionEnd 
\end{corollary}

\subsection{Three stations with no background noise}
\label{section:ThreeStations}
In this section we analyze the special case of the \UPN{} \( \cA_3 = \langle S,
\bar{1}, \Noise, \beta \rangle \), where \( S = \{\Station_0, \Station_1,
\Station_2\} \), \( \Noise = 0 \), and \( \beta = 1 \).
Our goal is to establish the following lemma, which constitutes the main
technical challenge in the course of proving Theorem~\ref{gtheorem:Convexity}.

\begin{lemma} \label{lemma:ThreeStations}
The reception zone \(\ReceptionZone_0\) of station \(\Station_0\) in \(\cA_3\)
is convex.
\end{lemma}

Lemma~\ref{lemma:ThreeStations} clearly holds if \( \Station_j = \Station_0 \)
for some \( j \in \{1, 2\} \),  as this implies that \( \ReceptionZone_0 =
\{\Station_0\} \).
So, in what follows we assume that no other station shares the location of
\(\Station_0\).
By Corollary~\ref{corollary:ThickSet} we know that \(\ReceptionZone_0\) is a
thick set.
Lemma~\ref{lemma:ThickSetsConvexity} can now be employed to establish
Lemma~\ref{lemma:ThreeStations}.
To do that, it is required to show that every line intersects \(\Boundary
\ReceptionZone_0 \) at most twice.

Consider an arbitrary line \(L\) in \(\Reals^2\).
We claim that \(L\) and \(\Boundary \ReceptionZone_0\) have no more than
two intersection points.
\ShortVersion 
If \( \Station_0 \in L \), then the claim holds due to
Lemma~\ref{lemma:StarConvex}.
Hence in the remainder of this section, assume that \(\Station_0 \notin L\).
Recall (see Section~\ref{section:Preliminaries}) that point \( (x, y) \in
\Reals^2 \) is on the boundary of \(\ReceptionZone_0\) iff it is a root of the
polynomial
\ShortVersionEnd 
\LongVersion 
First, note that if \( \Station_0 \in L \), then the claim holds due to
Lemma~\ref{lemma:StarConvex}.
Hence in the remainder of this section we assume that \(\Station_0 \notin L\).
Recall that point \( (x, y) \in \Reals^2 \) is on the boundary of
\(\ReceptionZone_0\) if and only if it is a root of the polynomial
\LongVersionEnd 
\begin{align}
\HearPoly(x, y)
= ~ & \left( (a_0 - x)^2 + (b_0 - y)^2 \right) \left( (a_1 - x)^2 + (b_1 - y)^2
+ (a_2 - x)^2 + (b_2 - y)^2 \right)
\label{equation:GeneralExpressionPolynomial} \\
& - \left( (a_1 - x)^2 + (b_1 - y)^2 \right) \left( (a_2 - x)^2 + (b_2 - y)^2
\right) \nonumber
\end{align}
\LongVersion 
(see Section~\ref{section:WirelessNetworks}),
\LongVersionEnd 
so it is sufficient to prove that the projection of \(\HearPoly(x, y)\) on the
line \(L\) has at most two distinct real roots.

Employing Lemma~\ref{lemma:Transformation}, we may assume that \(\Station_0\)
is located at the origin and that \(L\) is the line \( y = 1 \).
By substituting \( y = 1 \) into (\ref{equation:GeneralExpressionPolynomial})
and rearranging the resulting expression, we get
\begin{align*}
\HearPoly(x)
= ~ & \left( x^2 + 1 \right) \left( (a_1 - x)^2 + (b_1 - 1)^2 + (a_2 - x)^2 +
(b_2 - 1)^2 \right) \\
& - \left( (a_1 - x)^2 + (b_1 - 1)^2 \right) \left( (a_2 - x)^2 + (b_2 - 1)^2
\right) \\
= ~ & x^4 + (2 - 4 a_1 a_2) x^2 + \left( 2 a_2 a_1^2 + 2 a_2^2 a_1 + 2 (1 -
b_2)^2 a_1 - 2 a_1 + 2 a_2 (1 - b_1)^2 - 2 a_2 \right) x \\
& + a_1^2 - a_1^2 a_2^2 + a_2^2 - a_2^2 (1 - b_1)^2 + (1 - b_1)^2 - a_1^2 (1 -
b_2)^2 \\
& - (1 - b_1)^2 (1 - b_2)^2 + (1 - b_2)^2 ~ ,
\end{align*}
so that \((x, 1)\) is on the boundary of \(\ReceptionZone_0\) if and only if
\(x\) is a root of \(\HearPoly(x)\).

Our goal in the remainder of this section is to show that \(\HearPoly(x)\)
has at most two distinct real roots, and towards this goal we will first
invest some effort in simplifying this polynomial.
As a first step we show that we can restrict our attention to the case where
both \(\Station_1\) and \(\Station_2\) are in the first quarter above the line
\( y = 1 \), that is, \( a_j > 0 \) for \( j = 1, 2 \) and \( b_j \geq 1 \)
for \( j = 1, 2 \).
The latter restriction is trivial due to the symmetry of interference along
the line \( y = 1 \), which implies that if \( b_j < 1 \) for some \( j \in
\{1, 2\} \), then relocating \(\Station_j\) in \( (a_j, 1 + |1 - b_j|) \) does
not affect the interference at \(q\) for all points \(q\) on the line \( y = 1
\), and in particular, does not affect the number of simple real roots of
\(\HearPoly(x)\).
For the former restriction we
\LongVersion 
prove the following proposition.
\LongVersionEnd 
\ShortVersion 
establish the following proposition whose proof is deferred to
Appendix~\ref{appendix:ProofPropositionOppositeSign}.
\ShortVersionEnd 

\begin{proposition} \label{proposition:OppositeSign}
If \( \Sign(a_1) \cdot \Sign(a_2) \neq 1 \), then \(\HearPoly(x)\) has at
most two distinct real roots.
\end{proposition}
\newcommand{\ProofPropositionOppositeSign}{
We write \( \HearPoly(x) = x^4 + A x^2 + B x + C \) for coefficients \(A, B,
C\) depending on \( a_1, b_1, a_2, b_2 \), where \( A = 2 - 4 a_1 a_2 \).
Let \( \HearPoly'(x) = 4 x^3 + 2 A x + B \) be the derivative of
\(\HearPoly(x)\).
The polynomial \(\HearPoly'(x)\) is a cubic polynomial, thus, it has at
least one real root.
If it has exactly one real root, then \(\HearPoly(x)\) has exactly one
extreme point and at most two distinct real roots.
A cubic polynomial \( c_3 x^3 + c_2 x^2 + c_1 x + c_0 \) with real
coefficients has one real root when its discriminant
\[
\Delta
~ = ~ c_{1}^{2} c_{2}^{2} - 4 c_0 c_{2}^{3} - 4 c_{1}^{3} c_3 + 18 c_0 c_1 c_2
c_3 - 27 c_{0}^{2} c_{3}^{2}
\]
is negative.
In the case of \(\HearPoly'(x)\) we have \( \Delta = - 128 A^3 - 432 B^2 \).
The assertion follows from the observation that if \( \Sign(a_1) \neq
\Sign(a_2) \) or if \( a_1 = a_2 = 0 \), then \( A > 0 \) and \(\Delta\) is
negative.
} 
\LongVersion 
\begin{proof}
\ProofPropositionOppositeSign{}
\end{proof}
\LongVersionEnd 

By Proposition~\ref{proposition:OppositeSign}, we may hereafter assume that 
\( \Sign(a_1) = \Sign(a_2) \neq 0 \).
The case where \( a_j < 0 \) for \( j = 1, 2 \) is redundant, since relocating
station \(\Station_j\) in \((-a_j, b_j)\) turns \(\HearPoly(x)\) into
\(\HearPoly(-x)\), and in particular, does not affect the number of distinct
real roots of the polynomial.
Therefore, in what follows we assume that \( a_j > 0 \) and \( b_j \geq 1 \)
for \( j = 1, 2 \).
\LongVersion 
\par
\LongVersionEnd 
Our next step is to rewrite \(\HearPoly(x)\) as
\begin{equation}
\label{eq:Px-by-Qx}
\HearPoly(x) = (x^2 + 1)^2 - J(x) ~ ,
\end{equation}
where
\begin{align*}
J(x)
= & 4 a_2 a_1 x^2 - \left( 2 a_2 a_1^2 + 2 a_2^2 a_1 + 2 b_2^2 a_1 - 4 b_2 a_1
+ 2 a_2 b_1^2 - 4 a_2 b_1 \right) x \\
& + a_2^2 a_1^2 + b_2^2 a_1^2 - 2 b_2 a_1^2 + a_2^2 b_1^2 + b_1^2 b_2^2 -
2 b_1 b_2^2 - 2 a_2^2 b_1 - 2 b_1^2 b_2 + 4 b_1 b_2
\end{align*}
is a polynomial of degree \(2\).
Under the assumption that \(a_1\) and \(a_2\) are positive, \(J(x)\) has the
following (not necessarily distinct) real roots:
\[
r_1 = \frac{a_1^2 + (b_1 - 2) b_1}{2 a_1} \qquad \text{and} \qquad
r_2 = \frac{a_2^2 + (b_2 - 2) b_2}{2 a_2} ~ .
\]
Perhaps surprisingly, the root \(r_j\) corresponds to the x-coordinate of the
intersection point of the line \( y = 1 \) and the separation line \(L_j\) 
of \(\Station_0\) and \(\Station_j\) for \( j = 1, 2 \) (see
\Figure{}~\ref{figure:Roots}).
To validate this observation, note that the point \((x, y)\) is on \(L_j\) 
if and only if
\LongVersion 
\[
(x - a_j)^2 + (y - b_j)^2 = x^2 + y^2 ~,
\]
\LongVersionEnd 
\ShortVersion 
\( (x - a_j)^2 + (y - b_j)^2 = x^2 + y^2 , \)
\ShortVersionEnd 
or equivalently, if and only if
\LongVersion 
\[
a_{j}^{2} + b_{j}^{2} = 2 (a_j x + b_j y) ~ .
\]
\LongVersionEnd 
\ShortVersion 
\( a_{j}^{2} + b_{j}^{2} = 2 (a_j x + b_j y) . \)
\ShortVersionEnd 
Fixing \( y = 1 \), we get that \( x = \frac{a_j^2 + (b_j - 2) b_j}{2 a_j} =
r_j \).

\newcommand{\FigureRoots}{
\begin{figure}
\begin{center}
\includegraphics[width=0.5\textwidth]{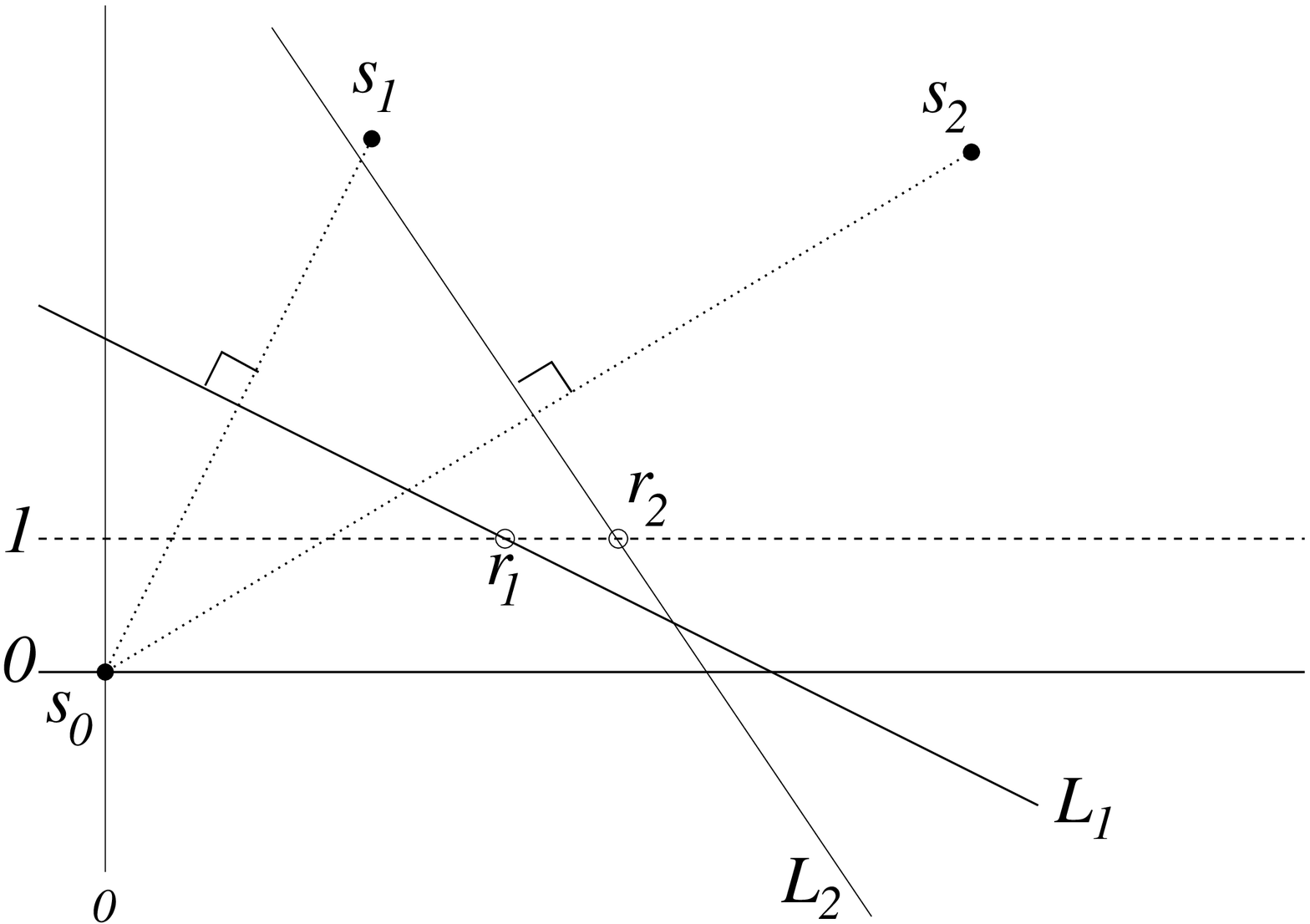}
\end{center}
\caption{\label{figure:Roots}
The point \((r_j, 1)\) is on the separation line \(L_j\) of \(\Station_0\) and
\(\Station_j\).
}
\end{figure}
} 
\LongVersion 
\FigureRoots{}
\LongVersionEnd 

Moreover, since \(x\) has a negative coefficient when \(L_j\) is expressed 
as \( y = -\frac{a_j}{b_j} x + a_{j}^{2} + b_{j}^{2} \), we conclude
that the point \( (x', 1) \) is as close to \(\Station_j\) at least as it is
to \(\Station_0\) for all \( x' \geq r_j \).
The following corollary follows since the real roots of \(\HearPoly(x)\)
correspond to points on the boundary of \(\ReceptionZone_0\), and since
\(\Station_0\) is heard at all such points.

\begin{corollary} \label{corollary:NoLargeRoots}
The real \(x'\) is not a root of the polynomial \(\HearPoly(x)\) for any \(
x' \geq \min \{ r_1, r_2\} \).
\end{corollary}

Incidentally, let us comment without proof that the point \(\min\{r_1, r_2\}\)
is the intersection point of the line \( y = 1 \) and the boundary of the
Voronoi cell of \(\Station_0\) in the Voronoi diagram of \(S\).

Fix \( \hx = (r_1 + r_2) / 2 \) and define the shifted variable \( z = x - \hx
\).
Since \(\hx\) is the center of the parabola \(J(x)\), it follows that when
expressing \(J(x)\) in terms of the shifted variable \(z\), we get the form \(
J(z) = \gamma z^2 + \delta \), where \( \gamma > 0 \) since the leading
coefficient of \(J(x)\) is positive, and \( \delta \leq 0 \) since \(J(x)\)
has at least one real root.
We can now express \(\HearPoly(x)\), as represented in (\ref{eq:Px-by-Qx}),
in terms of the shifted variable \(z\), introducing the polynomial
\[
\widehat{\HearPoly}(z) ~ = ~ \left( (z + \hx)^2 + 1 \right)^2 - \gamma z^2 -
\delta ~ ,
\]
which is obviously much simpler.
Clearly, \(\HearPoly(x)\) and \(\widehat{\HearPoly}(z)\) have the same
number of distinct real roots.

\LongVersion 
\begin{AvoidOverfullParagraph}
\LongVersionEnd 
\ShortVersion 
In what follows, we employ \emph{Sturm's Theorem}
\ShortVersionEnd 
\LongVersion 
In what follows, we employ \emph{Sturm's Theorem}, to be explained next,
\LongVersionEnd 
in order to bound the number of distinct real roots of
\(\widehat{\HearPoly}(z)\).
Consider some degree \(n\) polynomial \( P(x) = \alpha_n x^n + \cdots +
\alpha_1 x + \alpha_0 \) over the reals.
The \emph{Sturm sequence} of \(P(x)\) is a sequence of polynomials denoted by
\( P_0(x), P_1(x), \dots, P_m(x) \), where \( P_0(x) = P(x) \) and \( P_1(x) =
P'(x) \), and \( P_i(x) = -\Remainder(P_{i - 2}(x) / P_{i - 1}(x))\) for \( i >
1 \).
This recursive definition terminates at step \(m\) such that \(
-\Remainder(P_{m - 1}(x) / P_{m}(x)) = 0 \).
Since the degree of \(P_i(x)\) is at most \( n - i \), we conclude that \( m
\leq n \).
Define \(\Var_{P}(t)\) to be the number of sign changes in the sequence \(
P_{0}(t), P_{1}(t), \dots, P_{m}(t) \).
We are now ready to state the following theorem attributed to Jacques Sturm,
1829 (cf. \cite{BPR03}).
\LongVersion 
\end{AvoidOverfullParagraph}
\LongVersionEnd 

\begin{theorem}[Sturm's condition] \label{theorem:Sturm}
Consider two reals \(a, b\), where \( a < b \) and neither of them is a root
of \(P(x)\).
Then the number of distinct real roots of \(P(x)\) in the interval \( (a, b)
\) is \( \Var_{P}(a) - \Var_{P}(b) \).
\end{theorem}

We will soon show that\footnote{
We write \(\Var_{\widehat{\HearPoly}}(\infty)\) and
\(\Var_{\widehat{\HearPoly}}(-\infty)\) as a convenient shorthand for \(
\lim_{z \rightarrow \infty} \Var_{\widehat{\HearPoly}}(z) \) and \( \lim_{z
\rightarrow -\infty} \Var_{\widehat{\HearPoly}}(z) \), respectively.
} \( \Var_{\widehat{\HearPoly}}(-\infty) \leq 3 \) and \(
\Var_{\widehat{\HearPoly}}(\infty) \geq 1 \), hence \(
\Var_{\widehat{\HearPoly}}(-\infty) - \Var_{\widehat{\HearPoly}}(\infty)
\leq 2 \).
Therefore, by Theorem~\ref{theorem:Sturm}, we conclude that
\(\widehat{\HearPoly}(z)\) has at most two distinct real roots.
It is sufficient for our purposes to consider the first three polynomials in
the Sturm sequence of \(\widehat{\HearPoly}(z)\):
\begin{align*}
\widehat{\HearPoly}_{0}(z) = & z^4 + 4 \hx z^3 + \left( 6 \hx^2 - \gamma + 2
\right) z^2 + \left( 4 \hx^3 + 4 \hx \right) z + \hx^4 + 2 \hx^2 - \delta + 1
\\
\widehat{\HearPoly}_{1}(z) = & 4 z^3 + 12 \hx z^2 + 2 \left( 6 \hx^2 -
\gamma + 2 \right) z + 4 \hx^3 + 4 \hx \\
\widehat{\HearPoly}_{2}(z) = & (\gamma / 2 - 1) z^2 - \hx (2 + \gamma / 2) z
- \hx^2 - 1 + \delta ~ .
\end{align*}

\begin{proposition} \label{proposition:SignChangesPlus}
The polynomial \(\widehat{\HearPoly}(z)\) satisfies \(
\Var_{\widehat{\HearPoly}}(\infty) \geq 1 \).
\end{proposition}
\begin{proof}
We first argue that \(\widehat{\HearPoly}(z)\) does not have any root in the
interval \( [0, \infty) \).
This argument holds due to Corollary~\ref{corollary:NoLargeRoots} since by the
definition of \( z = x - \hx \), \( z \geq 0 \) implies \( x \geq \hx \geq
\min \{ r_1, r_2 \} \).
Therefore Theorem~\ref{theorem:Sturm} guarantees that \(
\Var_{\widehat{\HearPoly}}(\infty) = \Var_{\widehat{\HearPoly}}(0) \).
Now, the sign of \(\widehat{\HearPoly}_{0}(0)\) is positive while the sign
of \(\widehat{\HearPoly}_{2}(0)\) is negative, so there must be at least one
sign change when the Sturm sequence of \(\widehat{\HearPoly}(z)\) is
evaluated on \(0\), hence \( \Var_{\widehat{\HearPoly}}(\infty) =
\Var_{\widehat{\HearPoly}}(0) \geq 1 \).
\end{proof}

\begin{proposition} \label{proposition:SignChangesMinus}
The polynomial \(\widehat{\HearPoly}(z)\) satisfies \(
\Var_{\widehat{\HearPoly}}(-\infty) \leq 3 \).
\end{proposition}
\begin{proof}
First note that there are at most five polynomials in the Sturm sequence of
\(\widehat{\HearPoly}(z)\), hence \(\Var_{\widehat{\HearPoly}}(-\infty)\)
cannot be greater than \(4\).
Suppose, towards deriving contradiction, that \(
\Var_{\widehat{\HearPoly}}(-\infty) = 4 \).
This implies that there are exactly \(5\) polynomials in the Sturm sequence of
\(\widehat{\HearPoly}(z)\) and the degree of
\(\widehat{\HearPoly}_{i}(z)\) is \( 4 - i \) for every \( 0 \leq i \leq 4
\).
Clearly, both \(\Var_{\widehat{\HearPoly}}(-\infty)\) and
\(\Var_{\widehat{\HearPoly}}(\infty)\) depend solely on the signs of the
leading coefficients of the polynomials in the Sturm sequence of
\(\widehat{\HearPoly}(z)\).
Since \( \Sign(\widehat{\HearPoly}_{0}(-\infty)) = 1 \), we must have \(
\Sign(\widehat{\HearPoly}_{i}(-\infty)) = -1 \) for \( i = 1, 3 \) and \(
\Sign(\widehat{\HearPoly}_{i}(-\infty)) = 1 \) for \( i = 2, 4 \).
As \( \Sign(\widehat{\HearPoly}_{i}(\infty)) =
\Sign(\widehat{\HearPoly}_{i}(-\infty)) \) for \(i = 0, 2, 4 \) and \(
\Sign(\widehat{\HearPoly}_{i}(\infty)) =
-\Sign(\widehat{\HearPoly}_{i}(-\infty)) \) for \( i = 1, 3 \), we conclude
that \( \Sign(\widehat{\HearPoly}_{i}(\infty)) = 1 \) for every \(0 \leq i
\leq 4\).
Therefore \( \Var_{\widehat{\HearPoly}}(\infty) = 0 \), in contradiction to
Proposition~\ref{proposition:SignChangesPlus}.
\end{proof}

To conclude, combining Propositions \ref{proposition:SignChangesPlus} and
\ref{proposition:SignChangesMinus} with Theorem~\ref{theorem:Sturm}, we get
that \(\widehat{\HearPoly}(z)\) has at most two distinct real roots, and
thus \(\HearPoly(x)\) has at most two distinct real roots.
It follows that every line has at most two intersection points with
\( \Boundary \ReceptionZone_0 \), which completes the proof of
Lemma~\ref{lemma:ThreeStations}.

\subsection{Convexity with $n$ stations and no background noise}
\label{section:NoNoiseConvex}
\ShortVersion 
We now return
\ShortVersionEnd 
\LongVersion 
In this section we return 
\LongVersionEnd 
to a \UPN{} \( \cA = \langle S, \bar{1}, \Noise,
\beta \rangle \) with an arbitrary number of stations and with an arbitrary
reception threshold \( \beta \geq 1 \), but still, with no background noise
(i.e., \( \Noise = 0 \)), and establish the convexity of \( \ReceptionZone_0
\).

\begin{lemma} \label{lemma:NoBackgroundNoise}
The reception zone \(\ReceptionZone_0\) of station \(\Station_0\) in \(\cA\)
is convex.
\end{lemma}

Lemma~\ref{lemma:NoBackgroundNoise} is proved by induction on the number
of stations in the network, \( n = |S| \).
For the base of the induction, we consider the case where \( n = 2 \).
The theorem clearly holds if \(\Station_0\) and \(\Station_1\) share the same
location, as this implies that \( \ReceptionZone_0 = \{\Station_0\} \).
Furthermore, if \( \beta = 1 \), which means that \(\cA\) is trivial, then
\(\ReceptionZone_0\) is a half-plane and in particular, convex.
So, in what follows we assume that \( \Station_0 \neq \Station_1 \) and that
\( \beta > 1 \).

Corollary~\ref{corollary:ThickSet} implies that \(\ReceptionZone_0\) is a thick
set, thus, by Lemma~\ref{lemma:ThickSetsConvexity}, it is sufficient to argue
that every line \(L\) has at most two intersection points with \( \Boundary
\ReceptionZone_0 \).
If \( \Station_0 \in L \), then the argument holds due to
Lemma~\ref{lemma:StarConvex}.
If \( \Station_0 \notin L \), then \( \frac{\Energy(\Station_0,
q)}{\Interference(\Station_0, q)} = \beta \) is essentially a quadratic
equation, thus it has at most two real solutions which correspond to at most
two intersection points of \(L\) and \( \Boundary \ReceptionZone_0 \).

The inductive step of the proof of Lemma~\ref{lemma:NoBackgroundNoise} is
more involved.
We consider some arbitrary two points \( p_1, p_2 \in \ReceptionZone_0 \) and
prove that \( \Segment{p_1}{p_2} \subseteq \ReceptionZone_0 \).
Informally, we will show that if there exist at least two stations other than
\(\Station_0\), then we can discard one station and relocate the rest so that
the interference at \(p_i\) remains unchanged for \( i = 1, 2 \) and the
interference at \(q\) does not decrease for all points \( q \in
\Segment{p_1}{p_2} \).
By the inductive hypothesis, the segment \(\Segment{p_1}{p_2}\) is contained in
\(\ReceptionZone_0\) in the new setting, hence it is also contained in
\(\ReceptionZone_0\) in the original setting.
This idea relies on the following lemma.

\LongVersion 
\begin{AvoidOverfullParagraph}
\LongVersionEnd 
\begin{lemma} \label{lemma:DiscardStation}
Consider the stations \( \Station_0, \Station_1, \Station_2 \) and some
distinct two points \( p_1, p_2 \in \Reals^2 \).
If \( \Energy(\Station_0, p_i) \geq \Energy(\{\Station_1, \Station_2\}, p_i)
\) for \( i = 1, 2 \), then there exists a location \( \Station^{*} \in
\Reals^2 \) such that \\
(1) \( \Energy(\Station^{*}, p_i) = \Energy(\{\Station_1, \Station_2\}, p_i)
\) for \( i = 1, 2 \); and \\
(2) \( \Energy(\Station^{*}, q) \geq \Energy(\{\Station_1, \Station_2\}, q)
\), for all points \(q\) in the segment \(\Segment{p_1}{p_2}\).
\end{lemma}
\LongVersion 
\end{AvoidOverfullParagraph}
\LongVersionEnd 
%
%
\newcommand{\ProofPropositionBallBoundariesIntersect}{
By Lemma~\ref{lemma:Transformation}, we may assume that \( p_1 = (0, 0) \) and
\( p_2 = (c ,0) \) for some positive \(c\).
Since \(\Station_0\) must be in both \(\Ball_1\) and \(\Ball_2\), it follows
that the two balls cannot be disjoint.
We establish the claim by showing that \(\Ball_2\) is not contained in
\(\Ball_1\).
Let us define a new \UPN{} \(\cA'\) consisting of the stations \(\Station_1\),
\(\Station_2\), and \( \Station' = (c + \rho_2, 0) \) with no background
noise.
The points \(p_1\) and \(p_2\) are colinear with the station \(\Station'\),
hence Lemma~\ref{lemma:StarConvex} may be employed to conclude that
\begin{equation} \label{equation:SinrInequality}
\SINR_{\cA'}(\Station', p_1) < \SINR_{\cA'}(\Station', p_2) ~ .
\end{equation}

The construction of \(\cA'\) guarantees that \( \SINR_{\cA'}(\Station', p_2) =
\Energy(\Station', p_2) / \Energy(\{\Station_1, \Station_2\}, p_2) = 1 \).
On the other hand, if \( \Ball_2 \subseteq \Ball_1 \), then \(\Station'\) is
in \(\Ball_1\) (see \Figure{}~\ref{figure:B2IsInsideB1}), and thus \(
\Energy(\Station', p_1) \geq \Energy(\{\Station_1, \Station_2\}, p_1) \) which
means that \( \SINR_{\cA'}(\Station', p_1) \geq 1 \), in contradiction to
inequality~(\ref{equation:SinrInequality}).
Therefore \( \Boundary \Ball_1 \) and \( \Boundary \Ball_2 \) must intersect.
} 
\begin{proof}
Let \( \rho_i = 1 / \sqrt{\Energy(\{\Station_1, \Station_2\}, p_i)} \) and let
\(\Ball_i\) be a ball of radius \(\rho_i\) centered at \(p_i\) for \( i = 1, 2
\).
It is easy to verify that \(\Ball_i\) consists of all station locations
\(\Station\) such that \( \Energy(\Station, p_i) \geq \Energy(\{\Station_1,
\Station_2\}, p_i) \).
Assume without loss of generality that \( \rho_1 \geq \rho_2 \).
\ShortVersion 
The proof of the following proposition is deferred to
Appendix~\ref{appendix:ProofPropositionBallBoundariesIntersect}.
\ShortVersionEnd 

\begin{proposition} \label{proposition:BallBoundariesIntersect}
\( \Boundary \Ball_1 \) and \( \Boundary \Ball_2 \) intersect.
\end{proposition}
\LongVersion 
\begin{proof}
\ProofPropositionBallBoundariesIntersect{}
\end{proof}
\LongVersionEnd 

Let \(\Station^{*}\) be an intersection point of \( \Boundary \Ball_1 \) and
\( \Boundary \Ball_2 \) (see \Figure{}~\ref{figure:B1AndB2Intersect}).
We now show that \(\Station^{*}\) satisfies the assertion of 
Lemma~\ref{lemma:DiscardStation}.
Note that \( \Energy(\Station, p_i) = \Energy(\{\Station_1, \Station_2\},
p_i) \) for any station \(\Station\) located at \(\Boundary \Ball_i\), thus
\(\Station^{*}\) produces the desired energy at \(p_i\) for \( i = 1, 2 \),
that is, \( \Energy(\Station^{*}, p_i) = \Energy(\{\Station_1, \Station_2\},
p_i) \).
\LongVersion 
\par
\LongVersionEnd 
Consider a \UPN{} \(\cA^{*}\) consisting of the stations \(\Station^{*}\),
\(\Station_1\), and \(\Station_2\) with no background noise.
We have \( \SINR_{\cA^{*}}(\Station^{*}, p_i) = 1 \) for \( i = 1, 2 \).
Therefore, Lemma~\ref{lemma:ThreeStations} guarantees that
\( \SINR_{\cA^{*}}(\Station^{*}, q) \geq 1 \) for all points 
\( q \in \Segment{p_1}{p_2} \) which means that 
\( \Energy(\Station^{*}, q) \geq \Energy(\{\Station_1, \Station_2\}, q) \).
The assertion follows, completing the proof of 
Lemma~\ref{lemma:DiscardStation}.
\end{proof}

\newcommand{\FigureBTwoIsInsideBOne}{
\begin{figure}
\begin{center}
\includegraphics[width=0.3\textwidth]{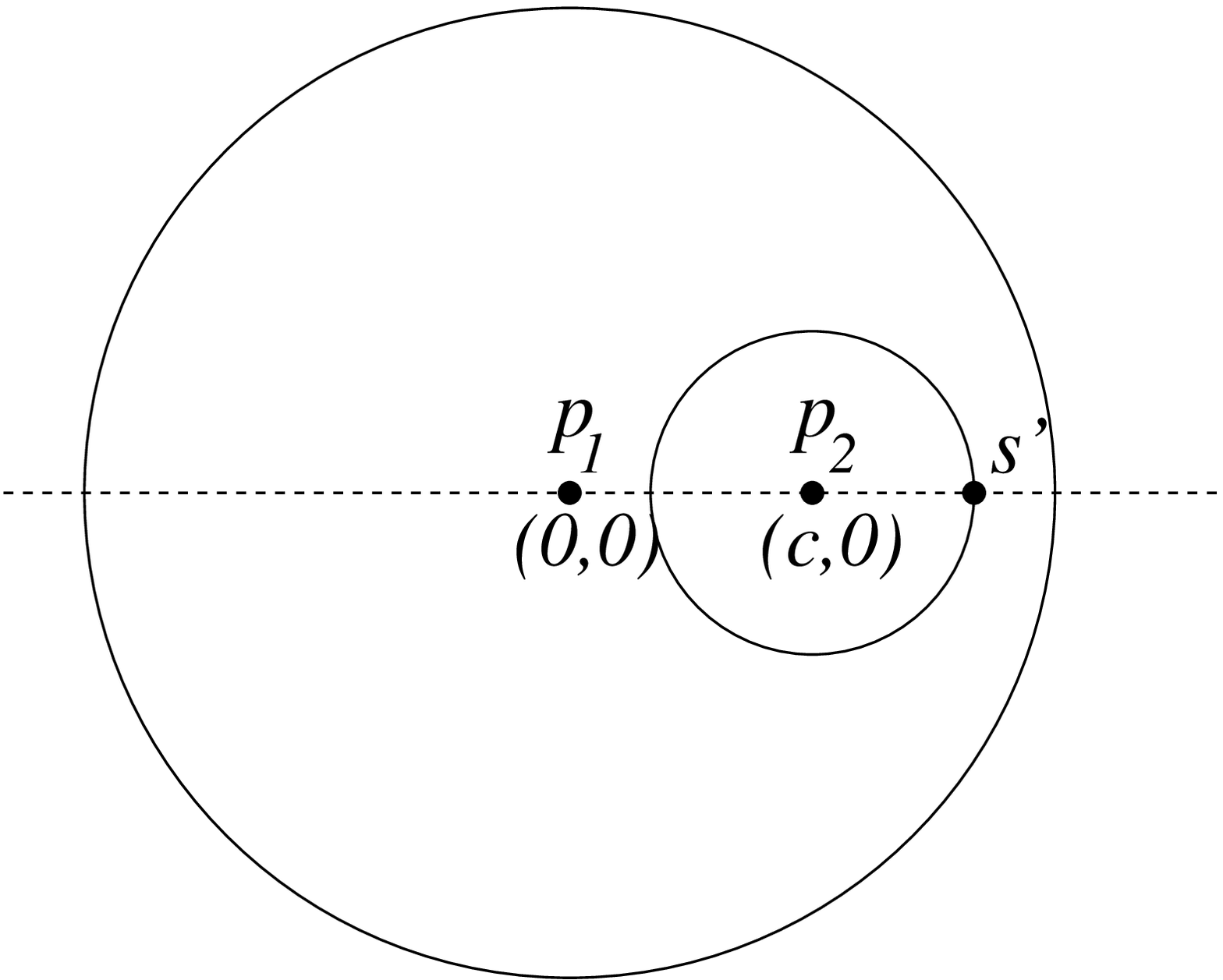}
\end{center}
\caption{\label{figure:B2IsInsideB1}
\(\Ball_2\) is strictly contained in \(\Ball_1\).
}
\end{figure}
} 
\LongVersion 
\FigureBTwoIsInsideBOne{}
\LongVersionEnd 

\newcommand{\FigureBOneAndBTwoIntersect}{
\begin{figure}
\begin{center}
\includegraphics[width=0.3\textwidth]{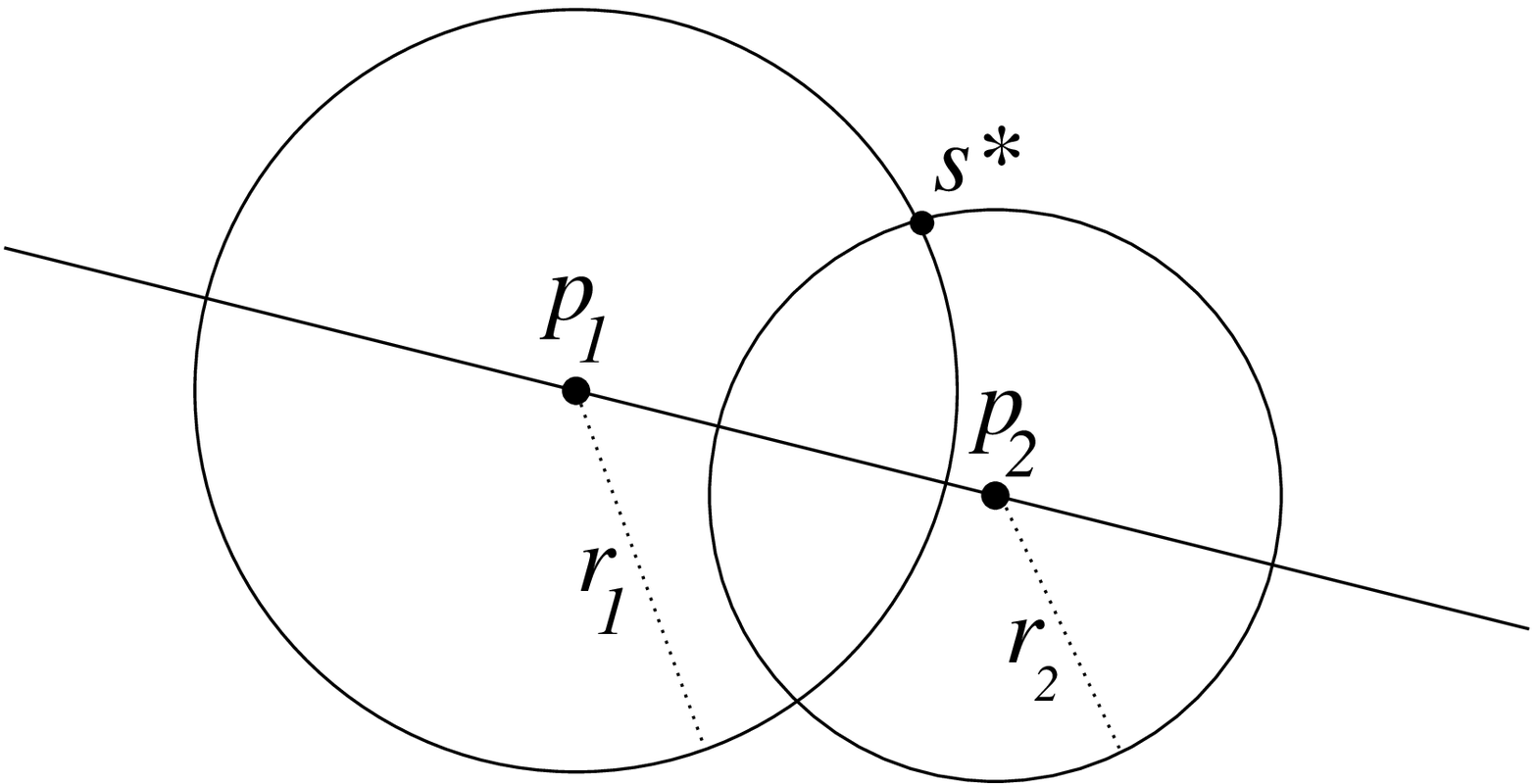}
\end{center}
\caption{\label{figure:B1AndB2Intersect}
\(\Station^{*}\) is at distance \(\rho_i\) from \(p_i\) for \( i = 1, 2 \).
}
\end{figure}
} 
\LongVersion 
\FigureBOneAndBTwoIntersect{}
\LongVersionEnd 

We now turn to describe the inductive step in the proof of
Lemma~\ref{lemma:NoBackgroundNoise}.
Assume by induction that the assertion of the theorem holds for \( n \geq 2 \)
stations, i.e., that in a \UPN{} with \( n \geq 2 \) stations and no background
noise we have \( \Segment{p_1}{p_2} \subseteq \ReceptionZone_0 \) for every \(
p_1, p_2 \in \ReceptionZone_0 \).
Now consider a \UPN{} \(\cA\) with \( n + 1 \) stations \( \Station_0, \dots,
\Station_n \) and no background noise.
Let \( p_1, p_2 \in \ReceptionZone_0 \).
Suppose that \( \Station_1 \) is closest to, say, \(p_1\) among all stations
\( \Station_1, \dots, \Station_n \).
Since \( p_1, p_2 \in \ReceptionZone_0 \), we know that \( \Energy(\Station_0,
p_i) > \Energy(\{\Station_1, \Station_2\}, p_i) \) for \( i = 1, 2 \).
Lemma~\ref{lemma:DiscardStation} guarantees that there exists a station
location \( \Station^{*} \in \Reals^2 \) such that
(1) \( \Energy(\Station^{*}, p_i) = \Energy(\{\Station_1, \Station_2\}, p_i)
\) for \( i = 1, 2 \); and
(2) \( \Energy(\Station^{*}, q) \geq \Energy(\{\Station_1, \Station_2\}, q)
\), for all points \(q\) in the segment \(\Segment{p_1}{p_2}\).

Note that the station location \(\Station^{*}\) must differ from
\(\Station_0\).
This is because \( \Energy(\Station^{*}, p_i) = \Energy(\{\Station_1,
\Station_2\}, p_i) \) while \( \Energy(\Station_0, p_i) >
\Energy(\{\Station_1, \Station_2\}, p_i) \) for \( i = 1, 2 \), thus \(
\dist{\Station^{*}, p_i} > \dist{\Station_0, p_i} \).

\ShortVersion 
\begin{AvoidOverfullParagraph}
\ShortVersionEnd 
Consider the \(n\)-station \UPN{} \(\cA^{*}\) obtained from \(\cA\) by
replacing \(\Station_1\) and \(\Station_2\) with a single station located at
\(\Station^{*}\) (see \Figure{}~\ref{figure:InductiveStep}).
Note that \( \Interference_{\cA^{*}}(\Station_0, p_i) =
\Interference_{\cA}(\Station_0, p_i) \) for \( i=1,2 \) and \(
\Interference_{\cA^{*}}(\Station_0, q) \geq \Interference_{\cA}(\Station_0, q)
\) for all points \( q \in \Segment{p_1}{p_2} \), hence \(
\SINR_{\cA^{*}}(\Station_0, p_i) = \SINR_{\cA}(\Station_0, p_i) \) for \( i =
1, 2 \) and \( \SINR_{\cA^{*}}(\Station_0, q) \leq \SINR_{\cA}(\Station_0, q)
\).
By the inductive hypothesis, \( \SINR_{\cA^{*}}(\Station_0, q) \geq \beta \)
for all points \( q \in \Segment{p_1}{p_2} \), therefore \(
\SINR_{\cA}(\Station_0, q) \geq \beta \) and \(\Station_0\) is heard at \(q\)
in \(\cA\).
It follows that every \( q \in \Segment{p_1}{p_2} \) belongs to \(
\ReceptionZone_0 \) in \(\cA\), which establishes the assertion and completes
the proof of Lemma~\ref{lemma:NoBackgroundNoise}.
\ShortVersion 
\end{AvoidOverfullParagraph}
\ShortVersionEnd 

\newcommand{\FigureInductiveStep}{
\begin{figure}
\begin{center}
\includegraphics[width=0.3\textwidth]{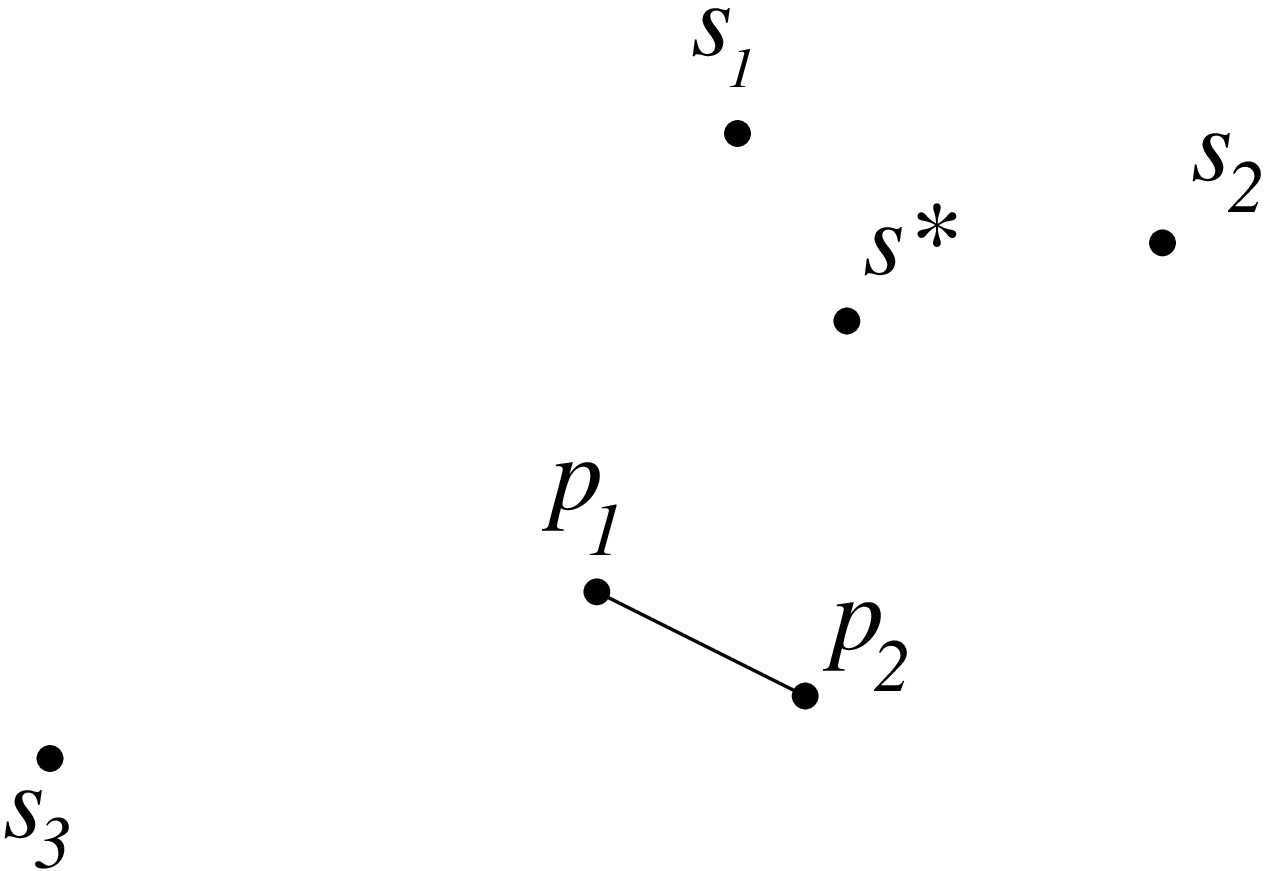}
\end{center}
\caption{\label{figure:InductiveStep}
\(\cA^{*}\) is obtained from \(\cA\) by removing stations \(\Station_1\) and
\(\Station_2\) and introducing station \(\Station^{*}\).
}
\end{figure}
} 
\LongVersion 
\FigureInductiveStep{}
\LongVersionEnd 

\LongVersion 
\subsection{Adding background noise}
\label{section:AddingNoise}
\LongVersionEnd 
\newcommand{\SectionAddingNoise}{
Our goal in this section is to show that the reception zones in a \UPN{} \(
\cA = \langle S, \bar{1}, \Noise, \beta \rangle \), where \( \Noise > 0 \), are
convex, thus establishing Theorem~\ref{gtheorem:Convexity}.
Let \(p_1\) and \(p_2\) be some points in \(\Reals^2\) and suppose that
\(\Station_0\) is heard at \(p_1\) and \(p_2\) in \(\cA\).
Let \(\Ball_1\) and \(\Ball_2\) be the balls of radius \( 1 / \sqrt{\Noise} \)
centered at \(p_1\) and \(p_2\), respectively.
Note that \( \SINR_{\cA}(\Station_0, p_i) \geq \beta \geq 1 \) implies that \(
\Energy(\Station_0, p_i) > \Noise \), thus \( \dist{\Station_0, p_i} < 1 /
\sqrt{\Noise} \) for \( i = 1, 2 \).
Therefore \( \dist{p_1, p_2} < 2 / \sqrt{\Noise} \) and \( \Boundary \Ball_1
\) and \( \Boundary \Ball_2 \) must intersect.

We construct an \( (n + 1) \)-station \UPN{} \(\cA'\) from \(\cA\) by locating
a new station \(\Station_n\) (with transmitting power \( \Power_n = 1 \) like
all other stations) in an intersection point of \( \Boundary \Ball_1 \) and \(
\Boundary \Ball_2 \) and omitting the background noise (see
\Figure{}~\ref{figure:AddingBackgroundNoise}).
Clearly, \( \Energy(\Station_n, p_i) = \Noise \) for \( i = 1, 2 \).
In particular, this means that \( \Station_n \neq \Station_0 \) as \(
\Energy(\Station_0, p_i) > \Noise \).
Since \( \dist{\Station_n, p_1} = \dist{\Station_n, p_2} = 1 / \sqrt{\Noise}
\), it follows that \( \dist{\Station_n, q} \leq 1 / \sqrt{\Noise} \) for all
points \( q \in \Segment{p_1}{p_2} \), hence \( \Energy(\Station_n, q) \geq
\Noise \).
Therefore \( \SINR_{\cA'}(\Station_0, p_i) = \SINR_{\cA}(\Station_0, p_i) \)
for \( i = 1, 2 \) and \( \SINR_{\cA}(\Station_0, q) \geq
\SINR_{\cA'}(\Station_0, q) \) for all points \( q \in \Segment{p_1}{p_2} \).
Since \(\cA'\) has no background noise, we may employ
Lemma~\ref{lemma:NoBackgroundNoise} to conclude that \(
\SINR_{\cA'}(\Station_0, q) \geq \beta \) for all points \( q \in
\Segment{p_1}{p_2} \). The assertion follows.
This completes the proof of Theorem~\ref{gtheorem:Convexity}.

\newcommand{\FigureAddingBackgroundNoise}{
\begin{figure}
\begin{center}
\includegraphics[width=0.3\textwidth]{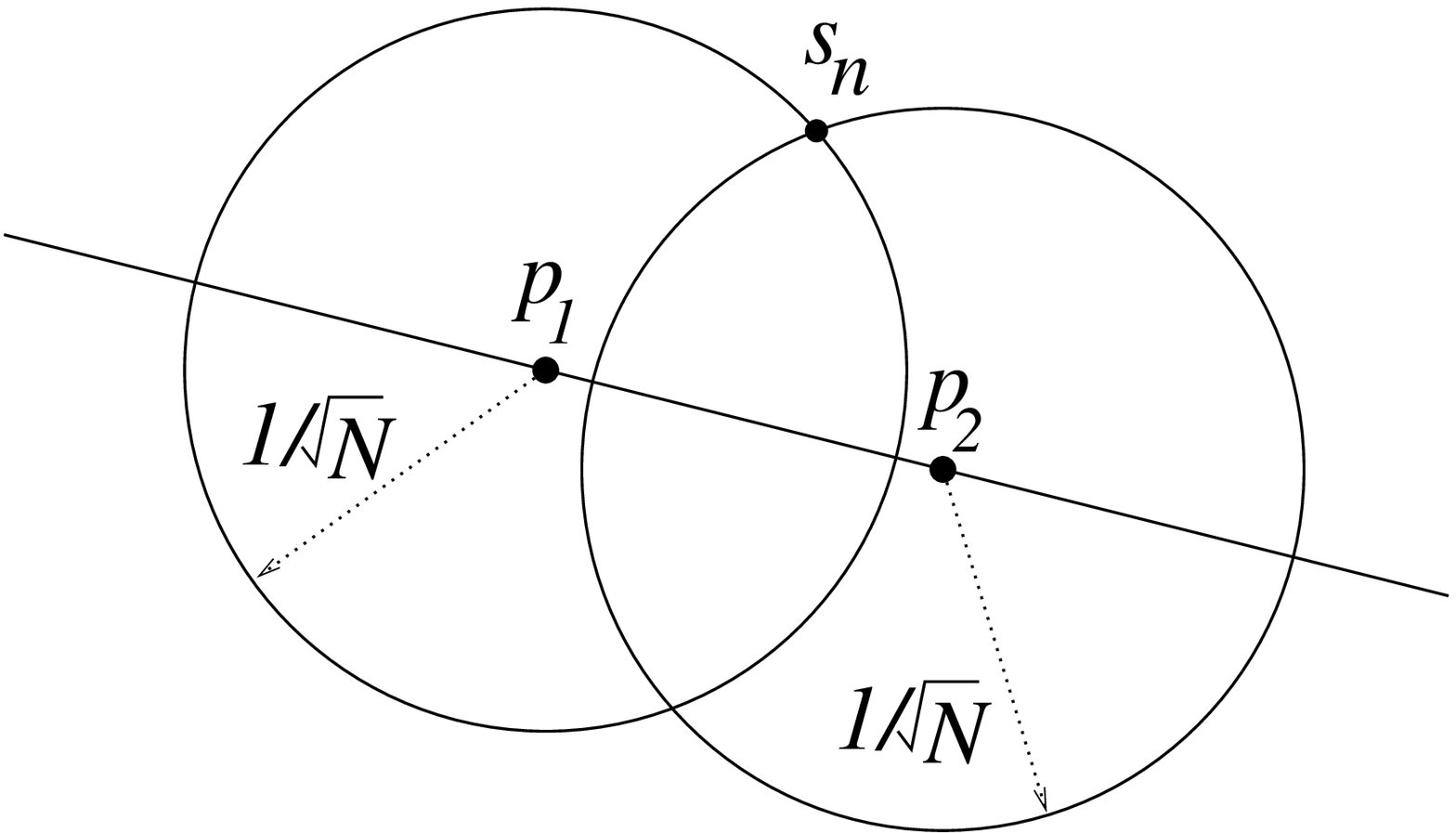}
\end{center}
\caption{\label{figure:AddingBackgroundNoise}
\(\cA'\) is obtained from \(\cA\) by omitting the background noise \(\Noise\)
and introducing station \(\Station_n\).
}
\end{figure}
} 
\LongVersion 
\FigureAddingBackgroundNoise{}
\LongVersionEnd 
} 
\LongVersion 
\SectionAddingNoise{}
\LongVersionEnd 

\ShortVersion 
\subsection{The fatness of the reception zones}
\label{subsection:Fatness}
In this section we develop a deeper understanding of the ``shape'' of the
reception zones by analyzing their fatness.
Consider a \UPN{} \( \cA = \langle S, \bar{1}, \Noise, \beta \rangle \), where
\( S = \{\Station_0, \dots, \Station_{n- 1}\} \) and\footnote{
Unlike the convexity proof
\LongVersion 
presented in Section~\ref{section:Convex}
\LongVersionEnd 
which holds for any \( \beta \geq 1 \), the analysis presented in the current
section is only suitable for \(\beta\) being a constant strictly greater than
\(1\).
In fact, when \( \beta = 1 \), the fatness parameter is not necessarily
defined (think of a trivial network).
} \(
\beta > 1 \) is a constant.
We focus on \(\Station_0\) and assume that its location is not shared by any
other station (otherwise, the reception zone \( \ReceptionZone_0 =
\{\Station_0\} \)).
In Appendix~\ref{section:FatnessExplicit} we establish explicit bounds on
\(\LargeRadius(\Station_0, \ReceptionZone_0)\) and \(\SmallRadius(\Station_0,
\ReceptionZone_0)\), showing the following.

\begin{theorem} \label{theorem:ExplicitBounds}
In a uniform energy network \( \cA = \langle S, \bar{1}, \Noise, \beta \rangle
\), where \( S = \{\Station_0, \dots, \Station_{n - 1}\} \) and \( \beta > 1
\) is a constant, if the minimum distance from \(\Station_0\) to any other
station is \( \MinDist > 0 \), then
\[
\SmallRadius(\Station_0, \ReceptionZone_0) \geq \frac{\MinDist}{\sqrt{\beta (n
- 1 + \Noise \cdot \MinDist^2)} + 1}
\quad \text{and} \quad
\LargeRadius(\Station_0, \ReceptionZone_0) \leq \frac{\MinDist}{\sqrt{\beta (1
+ \Noise \cdot \MinDist^2)} - 1} ~ .
\]
\end{theorem}

These bounds imply that \( \FatnessParameter(\Station_0, \ReceptionZone_0) = O
(\sqrt{n}) \). This is improved in Appendix~\ref{section:FatnessImplicit},
where we show that \( \FatnessParameter(\Station_0, \ReceptionZone_0) = O (1)
\), thus establishing Theorem~\ref{gtheorem:Fatness}.
\ShortVersionEnd 

\LongVersion 
\section{The fatness of the reception zones}
\label{section:Fatness}
In Section~\ref{section:Convex} we showed that the reception zone of each
station in a \UPN{} is convex.
In this section we develop a deeper understanding of the ``shape'' of the
reception zones by analyzing their fatness.
Consider a \UPN{} \( \cA = \langle S, \bar{1}, \Noise, \beta \rangle \), where
\( S = \{\Station_0, \dots, \Station_{n- 1}\} \) and\footnote{
Unlike the convexity proof presented in Section~\ref{section:Convex} which
holds for any \( \beta \geq 1 \), the analysis presented in the current
section is only suitable for \(\beta\) being a constant strictly greater than
\(1\).
In fact, when \( \beta = 1 \), the fatness parameter is not necessarily
defined (think of a trivial network).
} \(
\beta > 1 \) is a constant.
We focus on \(\Station_0\) and assume that its location is not shared by any
other station (otherwise, the reception zone \( \ReceptionZone_0 =
\{\Station_0\} \)).
In Section~\ref{section:FatnessExplicit} we establish explicit bounds on
\(\LargeRadius(\Station_0, \ReceptionZone_0)\) and \(\SmallRadius(\Station_0,
\ReceptionZone_0)\).
These bounds imply that \( \FatnessParameter(\Station_0, \ReceptionZone_0) = O
(\sqrt{n}) \).
This is improved in Section~\ref{section:FatnessImplicit},
where we show that \( \FatnessParameter(\Station_0, \ReceptionZone_0) = O (1)
\), thus establishing Theorem~\ref{gtheorem:Fatness}.
\LongVersionEnd 

\newcommand{\SectionFatnessPartA}{
\subsection{Explicit bounds}
\label{section:FatnessExplicit}
Our goal in this section is to establish an explicit lower bound on
\(\SmallRadius(\Station_0, \ReceptionZone_0)\) and an explicit upper bound on
\(\LargeRadius(\Station_0, \ReceptionZone_0)\).
Since \(\ReceptionZone_0\) is compact and convex, it follows that there exists
some points \( q_{\SmallRadius}, q_{\LargeRadius} \in \Boundary
\ReceptionZone_0 \) such that \( \dist{\Station_0, q_{\SmallRadius}} =
\SmallRadius(\Station_0, \ReceptionZone_0) \) and \( \dist{\Station_0,
q_{\LargeRadius}} = \LargeRadius(\Station_0, \ReceptionZone_0) \).
In fact, we may redefine \(\SmallRadius(\Station_0, \ReceptionZone_0)\) as the
distance from \(\Station_0\) to a closest point in \( \Boundary
\ReceptionZone_0 \) and \(\LargeRadius(\Station_0, \ReceptionZone_0)\) as the
distance from \(\Station_0\) to a farthest point in \( \Boundary
\ReceptionZone_0 \).

Fix \( \MinDist = \min \{ \dist{\Station_0, \Station_i} \mid i > 0  \} \).
An extreme scenario for establishing a lower bound on
\(\SmallRadius(\Station_0, \ReceptionZone_0)\) would be to place
\(\Station_0\) in \((0, 0)\) and all other \( n - 1 \) stations in \(
(\MinDist, 0) \).
This introduces the \UPN{} \( \cA_{\SmallRadius} = \langle \{ (0, 0),
(\MinDist, 0), \dots, (\MinDist, 0) \}, \bar{1}, \Noise, \beta \rangle \).
The point \(q_{\SmallRadius}\) whose distance to \(\Station_0\) realizes
\(\SmallRadius(\Station_0, \ReceptionZone_0)\) is thus located at \((d, 0)\)
for some \( 0 < d < \MinDist \).
On the other hand, an extreme scenario for establishing an upper bound on
\(\LargeRadius(\Station_0, \ReceptionZone_0)\) would be to place
\(\Station_0\) in \((0, 0)\), \(\Station_1\) in \( (\MinDist, 0) \), and all
other \( n - 2 \) stations in \( (\infty, 0) \) so that their energy at the
vicinity of \(\Station_0\) is ignored.
This introduces the \UPN{} \( \cA_{\LargeRadius} = \langle \{ (0, 0),
(\MinDist, 0), (\infty, 0) \dots, (\infty, 0) \}, \bar{1}, \Noise, \beta
\rangle \).
The point \(q_{\LargeRadius}\) whose distance to \(\Station_0\) realizes
\(\LargeRadius(\Station_0, \ReceptionZone_0)\) is thus located at \((-D, 0)\)
for some \( D > 0 \).

For the sake of analysis, we shall replace the background noise \(\Noise\) in
the above scenarios with a new station \(\Station_n\) located at \( (\MinDist,
0) \) whose power is \( \Noise \cdot \MinDist^2 \).
More formally, the \UPN{} \(\cA_{\SmallRadius}\) is replaced by the network
\[
\cA'_{\SmallRadius}
= \left\langle \{ (0, 0), (\MinDist, 0), \dots, (\MinDist, 0), (\MinDist, 0)
\}, (1, \dots, 1, \Noise \cdot \MinDist^2), 0, \beta \right\rangle
\]
and the \UPN{} \(\cA_{\LargeRadius}\) is replaced by the network
\[
\cA'_{\LargeRadius}
= \left\langle \{ (0, 0), (\MinDist, 0), (\infty, 0) \dots, (\infty, 0),
(\MinDist, 0) \}, (1, \dots, 1, \Noise \cdot \MinDist^2), 0, \beta
\right\rangle ~ .
\]
Note that the energy of the new station \(\Station_n\) at point \((x, 0)\) is
greater than \(\Noise\) for all \( 0 < x < \MinDist \); exactly \(\Noise\) for
\( x = 0 \); and smaller than \(\Noise\) for all \( x < 0 \).
Therefore the value of \(\SmallRadius(\Station_0, \ReceptionZone_0)\)
(respectively, \(\LargeRadius(\Station_0, \ReceptionZone_0)\)) under
\(\cA'_{\SmallRadius}\) (resp., \(\cA'_{\LargeRadius}\)) is smaller (resp.,
greater) than that under \(\cA_{\SmallRadius}\) (resp.,
\(\cA_{\LargeRadius}\)).
In the remainder of this section we establish a lower bound (resp., an upper
bound) on the former.

\ShortVersion 
\begin{AvoidOverfullParagraph}
\ShortVersionEnd 
In the context of \(\cA'_{\SmallRadius}\), we would like to compute the value
of \( d > 0 \) that solves the equation \(
\SINR_{\cA'_{\SmallRadius}}(\Station_0, (d, 0)) = \beta \), which means that
\LongVersion 
\[
\frac{d^{-2}}{(n-1 + \Noise \cdot \MinDist^2) (\MinDist - d)^{-2}} ~=~ \beta~,
\]
\LongVersionEnd 
\ShortVersion 
\( d^{-2} / \left((n-1 + \Noise \cdot \MinDist^2) (\MinDist - d)^{-2}\right) 
= \beta~,\)
\ShortVersionEnd 
or equivalently,
\( (\MinDist - d)^2 = d^2 \beta (n - 1 + \Noise \cdot \MinDist^2) \), or, 
\LongVersion 
\[ d ~=~ \frac{\MinDist}{\sqrt{\beta (n - 1 + \Noise \cdot \MinDist^2)} + 1} ~.
\]
\LongVersionEnd 
\ShortVersion 
\( d = \MinDist / 
\left( \sqrt{\beta (n - 1 + \Noise \cdot \MinDist^2)} + 1 \right).\)
\ShortVersionEnd 
Hence \( \SmallRadius(\Station_0, \ReceptionZone_0) \geq
\frac{\MinDist}{\sqrt{\beta (n - 1 + \Noise \cdot \MinDist^2)} + 1} \).
\ShortVersion 
\end{AvoidOverfullParagraph}
\ShortVersionEnd 

\ShortVersion 
\begin{AvoidOverfullParagraph}
\ShortVersionEnd 
In the context of \(\cA'_{\LargeRadius}\), we would like to compute the value
of \( D > 0 \) that solves the equation \(
\SINR_{\cA'_{\LargeRadius}}(\Station_0, (-D, 0)) = \beta \), which means that
\LongVersion 
\[
\frac{D^{-2}}{(1 + \Noise \cdot \MinDist^2) (\MinDist + D)^{-2}} ~=~ \beta ~,
\]
\LongVersionEnd 
\ShortVersion 
\( D^{-2} / \left( (1+\Noise\cdot\MinDist^2) (\MinDist+D)^{-2} \right) = 
\beta , \)
\ShortVersionEnd 
or equivalently,
\( (\MinDist + D)^2 = D^2 \beta (1 + \Noise \cdot \MinDist^2) \), or
\LongVersion 
\[ 
D ~=~ \frac{\MinDist}{\sqrt{\beta (1 + \Noise \cdot \MinDist^2)} - 1} ~ .
\]
\LongVersionEnd 
\ShortVersion 
\( D = 
\MinDist / \left( \sqrt{\beta (1 + \Noise \cdot \MinDist^2)} - 1 \right) . \)
\ShortVersionEnd 
Hence \( \LargeRadius(\Station_0, \ReceptionZone_0) \leq
\frac{\MinDist}{\sqrt{\beta (1 + \Noise \cdot \MinDist^2)} - 1} \).
\ShortVersion 
\end{AvoidOverfullParagraph}
\ShortVersionEnd 

Theorem~\ref{theorem:ExplicitBounds} follows from the above bounds and from
the following observation.

\begin{observation*}
The inequality \( \frac{\sqrt{a + c} + 1}{\sqrt{b + c} - 1} \leq
\frac{\sqrt{a} + 1}{\sqrt{b} - 1} \) holds for any choice of reals \( a \geq b
> 1 \) and \( c > 0 \).
\end{observation*}
} 
\newcommand{\SectionFatnessPartB}{
\subsection{An improved bound on the fatness parameter}
\label{section:FatnessImplicit}
In this section we prove Theorem~\ref{gtheorem:Fatness} by establishing the
following theorem.

\begin{theorem} \label{theorem:BoundingBoundedBalls}
The fatness parameter of \(\ReceptionZone_0\) with respect to \(\Station_0\)
satisfies
\[
\FatnessParameter(\Station_0, \ReceptionZone_0)
~ \leq ~ \frac{\sqrt{\beta} + 1}{\sqrt{\beta} - 1}
~ = ~ O(1) ~ .
\]
\end{theorem}

Theorem~\ref{theorem:BoundingBoundedBalls} is proved in three steps.
First, in
\LongVersion 
Section~\ref{section:OneDimensionTwoStations}
\LongVersionEnd 
\ShortVersion 
Appendix~\ref{section:OneDimensionTwoStations}
\ShortVersionEnd 
we bound the ratio \( \LargeRadius / \SmallRadius \) in a setting of two
stations in a one dimensional space.
This is used in
\LongVersion 
Section~\ref{section:PositiveColinear}
\LongVersionEnd 
\ShortVersion 
Appendix~\ref{section:PositiveColinear}
\ShortVersionEnd 
to establish the desired bound for a special type of \UPN{}s called
\emph{positive colinear} networks.
We conclude in
\LongVersion 
Section~\ref{section:GeneralUniformEnergy},
\LongVersionEnd 
\ShortVersion 
Appendix~\ref{section:GeneralUniformEnergy},
\ShortVersionEnd 
where we reduce the general case to the case of positive colinear networks.

\subsubsection{Two stations in a one dimensional space}
\label{section:OneDimensionTwoStations}
Let \(\cA\) be a network consisting of two stations \(\Station_0,
\Station_1\) with no background noise (i.e., \( \Noise = 0 \)).
Consider the embedding of \(\cA\) in the Euclidean one dimensional space
\(\Reals\) and assume without loss of generality that \(\Station_0\) is
located at \( a_0 = 0 \) and \(\Station_1\) is located at \( a_1 = 1 \)
(recall that this is made possible due to Lemma~\ref{lemma:Transformation}).
Suppose that \(\Station_0\) admits a unit transmitting power \( \Power_0 = 1 \)
while the transmitting power of \(\Station_1\) is any \( \Power_1 \geq 1 \).
Let \( \RightMost = \max \{ p > 0 \mid \SINR_{\cA}(\Station_0, p) \geq \beta
\} \) and let \( \LeftMost = \min \{ p < 0 \mid \SINR_{\cA}(\Station_0, p)
\geq \beta \} \) (see \Figure{}~\ref{figure:TwoNodesOnTheLine}).
It is easy to verify that \( \ReceptionZone_0 = [\LeftMost, \RightMost] \) and
that \( \SmallRadius = \RightMost \) and \( \LargeRadius = -\LeftMost \).

\newcommand{\FigureTwoNodesOnTheLine}{
\begin{figure}
\begin{center}
\LongVersion 
\includegraphics[width=0.35\textwidth]{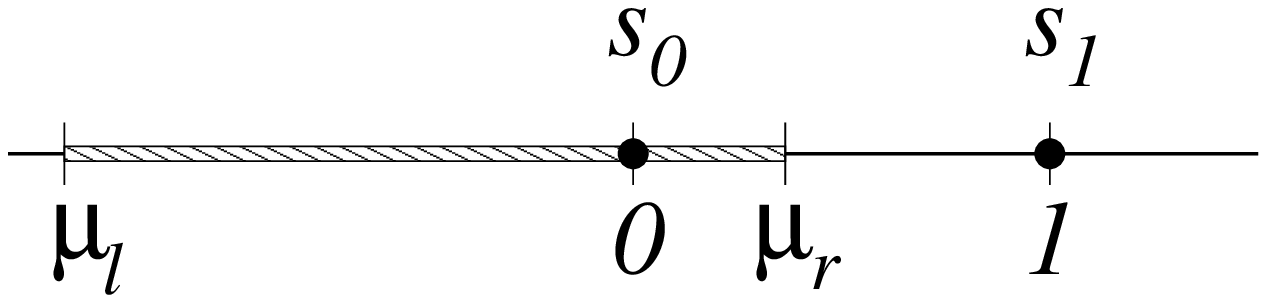}
\LongVersionEnd 
\ShortVersion 
\includegraphics[width=0.25\textwidth]{figs/two-nodes-on-the-line}
\ShortVersionEnd 
\end{center}
\caption{\label{figure:TwoNodesOnTheLine}
The embedding of \(\Station_0\) and \(\Station_1\) in a one dimensional
space.
}
\end{figure}
} 
\LongVersion 
\FigureTwoNodesOnTheLine{}
\LongVersionEnd 

\begin{lemma} \label{lemma:OneDimensionTwoStations}
The network \(\cA\) satisfies \( \LargeRadius / \SmallRadius \leq
\frac{\sqrt{\beta} + 1}{\sqrt{\beta} - 1} \), with equality attained when \(
\Power_1 = 1 \).
\end{lemma}
\begin{proof}
The boundary points \(\RightMost\) and \(\LeftMost\) of \(\ReceptionZone_0\)
are the solutions to the quadratic equation
\[
\frac{(x - 1)^2}{\Power_1 x^2} = \beta
\quad \Longleftrightarrow \quad
(\beta \Power_1 - 1) x^2 + 2x - 1 = 0 ~ .
\]
Solving this equation, we get
\begin{align*}
\RightMost
& = \frac{-2 + \sqrt{4 \beta \Power_1}}{2 \beta \Power_1 - 2}
= \frac{\sqrt{\beta \Power_1} - 1}{\beta \Power_1 - 1} \\
\LeftMost
& = \frac{-2 - \sqrt{4 \beta \Power_1}}{2 \beta \Power_1 - 2}
= - \frac{\sqrt{\beta \Power_1} + 1}{\beta \Power_1 - 1} ~ .
\end{align*}
Therefore the ratio \( \LargeRadius / \SmallRadius \) satisfies
\[
\frac{\LargeRadius}{\SmallRadius} = \frac{\sqrt{\beta \Power_1} +
1}{\sqrt{\beta \Power_1} - 1} \leq \frac{\sqrt{\beta} + 1}{\sqrt{\beta} - 1}
\]
as desired.
\end{proof}

\subsubsection{Positive colinear networks}
\label{section:PositiveColinear}
In this section we switch back to the Euclidean plane \(\Reals^2\) and
consider a special type of \UPN{}s.
A network \( \cA = \langle \{\Station_0, \dots, \Station_{n - 1}\}, \bar{1},
\Noise, \beta \rangle \) is said to be \emph{positive colinear} if \(
\Station_0 = (0, 0) \) and \( \Station_i = (a_i, 0) \) for some \( a_i > 0 \)
for every \( 1 \leq i \leq n - 1 \).
Positive colinear networks play an important role in the subsequent analysis
due to the following lemma.
(Refer to \Figure{}~\ref{figure:NNodesOnTheLine} for illustration.)

\newcommand{\FigureNNodesOnTheLine}{
\begin{figure}
\begin{center}
\LongVersion 
\includegraphics[width=0.4\textwidth]{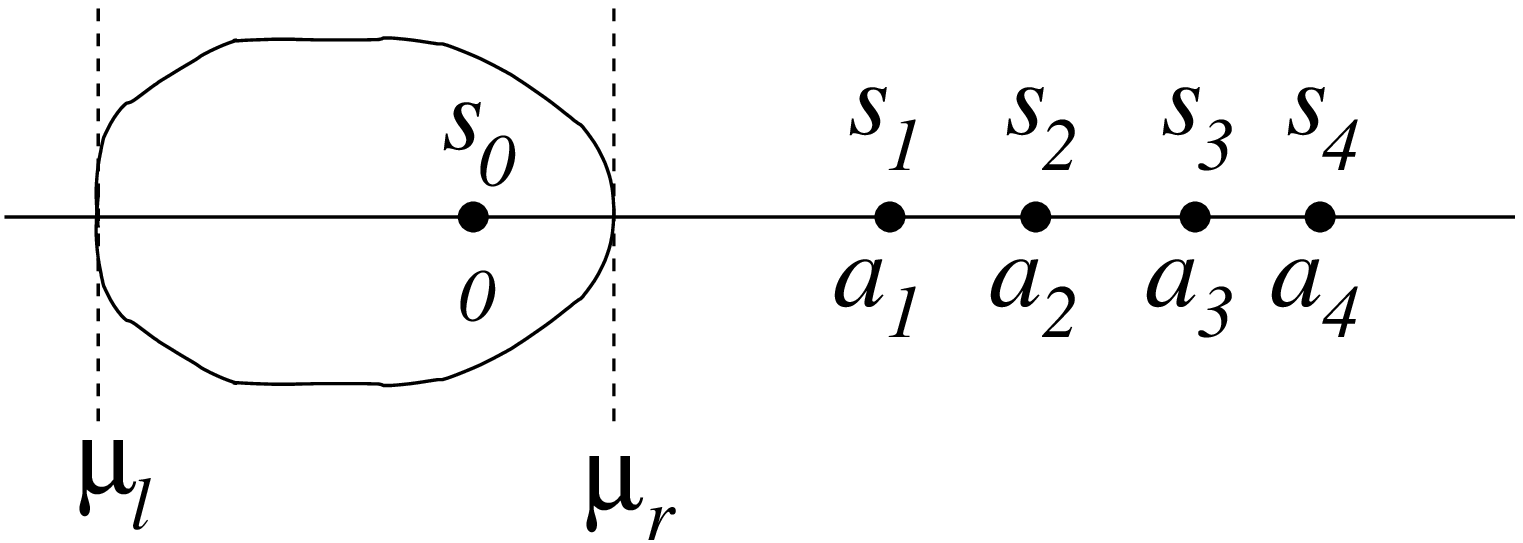}
\LongVersionEnd 
\ShortVersion 
\includegraphics[width=0.3\textwidth]{figs/n-nodes-on-the-line}
\ShortVersionEnd 
\end{center}
\caption{\label{figure:NNodesOnTheLine}
A positive colinear network.
}
\end{figure}
} 
\LongVersion 
\FigureNNodesOnTheLine{}
\LongVersionEnd 

\ShortVersion 
\begin{AvoidOverfullParagraph}
\ShortVersionEnd 
\begin{lemma} \label{lemma:PositiveColinear}
Let \(\cA\) be a positive colinear \UPN{}.
Fix \( \RightMost = \max \{ r > 0 \mid \SINR_{\cA}(\Station_0, (r, 0))
\geq \beta \} \) and \( \LeftMost = \min \{ r < 0 \mid \SINR_{\cA}(\Station_0,
(r, 0)) \geq \beta \} \).
Then
\[
\SmallRadius = \RightMost ~ ,
\]
\[
\LargeRadius = -\LeftMost ~ ,
\]
and
\[
-\frac{\LeftMost}{\RightMost} \leq \frac{\sqrt{\beta} + 1}{\sqrt{\beta} - 1} ~
.
\]
\end{lemma}
\ShortVersion 
\end{AvoidOverfullParagraph}
\ShortVersionEnd 

Before we can establish Lemma~\ref{lemma:PositiveColinear}, we would like to
prove some basic properties of positive colinear networks.
First, we argue that the reception zone \(\ReceptionZone_0\) of \(\Station_0\)
under the positive colinear network \(\cA\) is contained in the halfplanes
intersection \( \{ (x, y) \mid \LeftMost \leq x \leq \RightMost \} \).
To see why this is true, suppose towards deriving contradiction that the point
\( (x, y) \in \ReceptionZone_0 \) for some \( x > \RightMost \) or \( x <
\LeftMost \).
Due to symmetry considerations, we conclude that the point \((x, -y)\) is also
in \(\ReceptionZone_0\).
By the convexity of \(\ReceptionZone_0\), it follows that \( (x, 0) \in
\ReceptionZone_0 \), in contradiction to the definitions of \(\RightMost\) and
\(\LeftMost\).

\begin{corollary}
If \( (x, y) \in \ReceptionZone_0 \), then \( \LeftMost \leq x \leq \RightMost
\).
\end{corollary}

We now turn to prove that \( \SmallRadius = \RightMost \).
To do so, we will prove that the ball of radius \(\RightMost\) centered at
\(\Station_0\) is contained in \(\ReceptionZone_0\).
In fact, by the convexity of \(\ReceptionZone_0\), it is sufficient to show
that the point \( p(\theta) = (\RightMost \cos \theta, \RightMost \sin \theta)
\) is in \(\ReceptionZone_0\) for all \( 0 \leq \theta \leq \pi \).
Since the network is positive colinear, it follows that
\(\Interference_{\cA}(\Station_0, p(\theta))\) attains its maximum for \(
\theta = 0 \).
Therefore the fact that \( p(0) = (\RightMost, 0) \in \ReceptionZone_0 \)
implies that \( p(\theta) \in \ReceptionZone_0 \) for all \( 0 \leq \theta \leq
\pi \) as desired.

\begin{corollary} \label{corollary:SmallRadius}
The positive colinear network \(\cA\) satisfies \( \SmallRadius = \RightMost
\).
\end{corollary}

Next, we prove that \(\LargeRadius\) is realized by the point \((\LeftMost,
0)\).
Indeed, by the triangle inequality, all points at distance \(d\) from
\(\Station_0\) are at distance \( \leq d + a_i \) from \( \Station_i = (a_i,
0) \), with equality attained for the point \((-d, 0)\).
Thus the minimum interference to \(\Station_0\) under \(\cA\) among all
points at distance \(d\) from \(\Station_0\) is attained at the point \((-d,
0)\).
Therefore, by the definition of \(\LeftMost\), there cannot exist any point
\( p \in \ReceptionZone_0 \) such that \( \dist{p, \Station_0} > -\LeftMost
\).

\begin{corollary} \label{corollary:LargeRadius}
The positive colinear network \(\cA\) satisfies \( \LargeRadius = -\LeftMost
\).
\end{corollary}

It remains to establish the bound on the ratio \( - \LeftMost / \RightMost =
\LargeRadius / \SmallRadius \).
Fix \( d = \min \{ a_i \mid 1 \leq i \leq n - 1 \} \), that is, the leftmost
station other than \(\Station_0\) is located at \((d, 0)\).
Clearly, \( \RightMost < d \).
We denote the energy of station \(\Station_i\) at \((\RightMost, 0)\) by \(
\mathcal{E}_i = (a_i - \RightMost)^{-2} \).
We construct a new network \( \cA' = \langle S', \Power', 0, \beta \rangle \)
consisting of \(\Station_0\) and \(n\) new stations \( \Station'_1, \dots,
\Station'_n \), all located at \((d, 0)\).
For \( 1 \leq i \leq n - 1 \), we set the transmitting power \(\Power'_i\) of
the new station \(\Station'_{i}\) so that the energy it produces at
\((\RightMost, 0)\) is \(\mathcal{E}_i\).
The transmitting power \(\Power'_{n}\) of the new station \(\Station'_{n}\) is
set so that the energy it produces at \((\RightMost, 0)\) is \(\Noise\).
This accounts to
\[
\Power'_i =
\left\{
\begin{array}{ll}
\mathcal{E}_i \cdot (d - \RightMost)^2 & \text{ for } 1 \leq i \leq n - 1;
\text{ and} \\
\Noise \cdot (d - \RightMost)^2 & \text{ for } i = n ~ .
\end{array}
\right.
\]

The network \(\cA'\) falls into the setting of
\LongVersion 
Section~\ref{section:OneDimensionTwoStations}:
\LongVersionEnd 
\ShortVersion 
Appendix~\ref{section:OneDimensionTwoStations}:
\ShortVersionEnd 
the stations \( \Station'_1, \dots, \Station'_n \) share the same location,
thus they can be considered as a single station with transmitting power \(
\sum_{i = 1}^{n} \Power'_i \).
We define \( \RightMost' = \max \{ r > 0 \mid \SINR_{\cA'}(\Station_0, (r, 0))
\geq \beta \} \) and \( \LeftMost' = \min \{ r < 0 \mid
\SINR_{\cA'}(\Station_0, (r, 0)) \geq \beta \} \), so that the restriction of
the reception zone of \(\Station_0\) under \(\cA'\) to the \(x\)-axis is
\( [\LeftMost', \RightMost'] \).
Lemma~\ref{lemma:OneDimensionTwoStations} implies that \( -\LeftMost' /
\RightMost' \leq \frac{\sqrt{\beta} + 1}{\sqrt{\beta} - 1} \).
The remainder of the proof relies on establishing the following two bounds: \\
(1) \( \SINR_{\cA'}(\Station_0, (r, 0)) \leq \SINR_{\cA}(\Station_0,
(r, 0)) \) for all \( \RightMost \leq r < d \); and \\
(2) \( \SINR_{\cA'}(\Station_0, (r, 0)) \geq \SINR_{\cA}(\Station_0, (r, 0))
\) for all \( r \leq \RightMost \), \( r \neq 0 \). \\
By combining bounds (1) and (2), we conclude that \( \RightMost' \leq
\RightMost \) and \( \LeftMost' \leq \LeftMost \), which completes the proof
of Lemma~\ref{lemma:PositiveColinear}.

To establish bounds (1) and (2), consider some point \( p = (r, 0) \), where
\( r < d \), \( r \neq 0 \).
For every \( 1 \leq i \leq n - 1 \), we have
\[
\Energy(\Station_i, p) = \frac{1}{(a_i - r)^2} ~ ,
\quad \text{while} \quad
\Energy(\Station'_i, p) = \frac{\Power'_i}{(d - r)^2} = \frac{(d -
\RightMost)^2}{(d - r)^2 (a_i - \RightMost)^2} ~ .
\]
Comparing between the former expression and the latter, we get
\LongVersion 
\[ 
\Energy(\Station_i, p) ~\geq~ \Energy(\Station'_i, p) ~,
\]
\LongVersionEnd 
\ShortVersion 
\(  \Energy(\Station_i, p) \geq \Energy(\Station'_i, p) , \)
\ShortVersionEnd 
or equivalently,
\LongVersion 
\[
(d - r) (a_i - \RightMost) ~\geq~ (d - \RightMost) (a_i - r) ~ .
\]
\LongVersionEnd 
\ShortVersion 
\( (d - r) (a_i - \RightMost) \geq (d - \RightMost) (a_i - r) .\)
\ShortVersionEnd 
Rearranging, we get
\LongVersion 
\[
d a_i - d \RightMost - a_i r + r \RightMost ~\geq~ d a_i
- d r - a_i \RightMost + r \RightMost ~, 
\]
\LongVersionEnd 
\ShortVersion 
\( d a_i - d \RightMost - a_i r + r \RightMost \geq d a_i
- d r - a_i \RightMost + r \RightMost , \)
\ShortVersionEnd 
or
\LongVersion 
\[
\RightMost (a_i - d) ~\geq~ r (a_i - d) ~ ,
\]
\LongVersionEnd 
\ShortVersion 
\( \RightMost (a_i - d) \geq r (a_i - d) , \)
\ShortVersionEnd 
where the last inequality holds if and only \( a_i = d \), which, by
definition, implies that \( \Energy(\Station_i, p) = \Energy(\Station'_i, p)
\), or \( \RightMost \geq r \).
Therefore the contribution of \(\Station'_i\) to the interference to
\(\Station_0\) at \( p = (0, r) \) is not larger then that of \(\Station_i\)
as long as \( r \leq \RightMost \) and not smaller than that of \(\Station_i\)
as long as \( \RightMost \leq r < d \).
On the other hand, the energy of \(\Station'_n\) at \( p = (r, 0) \) is
not larger than the background noise \(\Noise\) for all \( d \leq \RightMost
\) and not smaller than \(\Noise\) for all \( \RightMost \leq r < d \).
Bounds (1) and (2) follow.

\subsubsection{A general uniform power network}
\label{section:GeneralUniformEnergy}
We are now ready to prove the main theorem of 
\LongVersion 
Section~\ref{section:Fatness}.
\LongVersionEnd 
\ShortVersion 
Appendix~\ref{appendix:Fatness}
\ShortVersionEnd 

\begin{proof}[Proof of Theorem~\ref{theorem:BoundingBoundedBalls}.]
Let \( \cA = \langle S, \bar{1}, \Noise, \beta \rangle \), where \( S = \{
\Station_0, \dots, \Station_{n -1} \} \) and \( \beta > 1 \) is a constant, be
an arbitrary \UPN{}.
We employ Lemma~\ref{lemma:Transformation} to assume that \(\Station_0\) is
located at \((0, 0)\) and that \( \max \{ \dist{\Station_0, q} \mid q \in
\ReceptionZone_0 \} \) is realized by a point \(q\) on the negative
\(x\)-axis.
By definition, we have \( q = (-\LargeRadius, 0) \).

We construct a new \UPN{} \( \cA' = \langle \{\Station_0, \Station'_1, \dots,
\Station'_{n - 1}\}, \bar{1}, \Noise, \beta \rangle \), obtained from \(\cA\)
by rotating each station \(\Station_i\) around the point \(q\) until it
reaches the positive \(x\)-axis (see \Figure{}~\ref{figure:NNodesInTwoDim}).
More formally, the location of \(\Station_0\) remains unchanged and \(
\Station'_i = (a'_i, 0) \), where \( a'_i = \dist{\Station_i, q} -
\LargeRadius \) for every \( 1 \leq i \leq n - 1 \).
Since \(\Station_0\) is heard at \(q\) under \(\cA\), it follows that \(
\LargeRadius = \dist{\Station_0, q} < \dist{\Station_i, q} \) for every \( 1
\leq i \leq n - 1 \), hence \( a'_i > 0 \) and \(\cA'\) is a positive colinear
network.
Clearly, \( \dist{\Station'_i, q} = \dist{\Station_i, q} \) for every \( 1
\leq i \leq n - 1 \).

Let \(\ReceptionZone'_0\) denote the reception zone of \(\Station_0\) under
\(\cA'\).
Fix \( \SmallRadius' = \max \{ r > 0 \mid \Ball(\Station_0, r) \subseteq
\ReceptionZone'_0 \} \) and \( \LargeRadius' = \min \{ r > 0 \mid
\Ball(\Station_0, r) \supseteq \ReceptionZone'_0 \} \).
Let \( \RightMost' = \max \{ r > 0 \mid \SINR_{\cA'}(\Station_0, (r, 0))
\geq \beta \} \) and let \( \LeftMost' = \min \{ r < 0 \mid
\SINR_{\cA'}(\Station_0, (r, 0)) \geq \beta \} \).
Lemma~\ref{lemma:PositiveColinear} guarantees that \( \SmallRadius' =
\RightMost' \), \( \LargeRadius' = -\LeftMost' \), and \(
\frac{\LargeRadius'}{\SmallRadius'} \leq \frac{\sqrt{\beta} + 1}{\sqrt{\beta}
- 1} \).
We shall establish the proof of Theorem~\ref{theorem:BoundingBoundedBalls} by
showing that \( \LargeRadius' = \LargeRadius \) and \( \SmallRadius' \leq
\SmallRadius \).
The former is a direct consequence of Lemma~\ref{lemma:PositiveColinear}:
since \( \SINR_{\cA'}(\Station_0, q) = \SINR_{\cA}(\Station_0, q) = \beta \),
it follows that \( \max \{ \dist{\Station_0, p} \mid p \in \ReceptionZone'_0 \}
\) is realized at \( p = q \).

It remains to prove that \( \SmallRadius' \leq \SmallRadius \).
We shall do so by showing that \( \Ball(\Station_0, \SmallRadius') \subseteq
\ReceptionZone_0 \).
Fix \( \rho_i = \dist{\Station_i, q} \) for every \( 1 \leq i \leq n - 1 \).
We argue that the ball \(\Ball(\Station_0, \SmallRadius')\) is strictly
contained in the ball \(\Ball(q, \rho_i)\) for every \( 1 \leq i \leq n - 1
\).
To see why this is true, observe that \( -\Delta < 0 < \SmallRadius' =
\RightMost' < a'_i \), hence the ball centered at \( q = (-\Delta, 0) \) of
radius \( \rho_i = \Delta + a'_i \) strictly contains the ball of radius
\(\SmallRadius'\) centered at \( \Station_0 = (0, 0) \).

Consider an arbitrary point \( p \in \Ball(\Station_0, \SmallRadius') \).
We can now rewrite
\[
\dist{\Station'_i, (\SmallRadius', 0)}
= a'_i - \SmallRadius'
= \min \{ \dist{t, t'} \mid t \in \Ball(\Station_0, \SmallRadius'), t' \in
\Boundary \Ball(q, \rho_i) \}
\]
for every \( 1 \leq i \leq n - 1 \).
Recall that \( \Station_i \in \Boundary \Ball(q, \rho_i) \), thus \(
\dist{\Station_i, p} \geq \dist{\Station'_i, (\SmallRadius', 0)} \).
Therefore \( \Interference_{\cA}(\Station_0, p) \leq
\Interference_{\cA'}(\Station_0, (\SmallRadius', 0)) \) and \(
\SINR_{\cA}(\Station_0, p) \geq \SINR_{\cA'}(\Station_0, (\SmallRadius', 0)) =
\beta \).
It follows that \( p \in \ReceptionZone_0 \), which completes the proof.
\end{proof}

\newcommand{\FigureNNodesInTwoDim}{
\begin{figure}
\begin{center}
\LongVersion 
\includegraphics[width=0.8\textwidth]{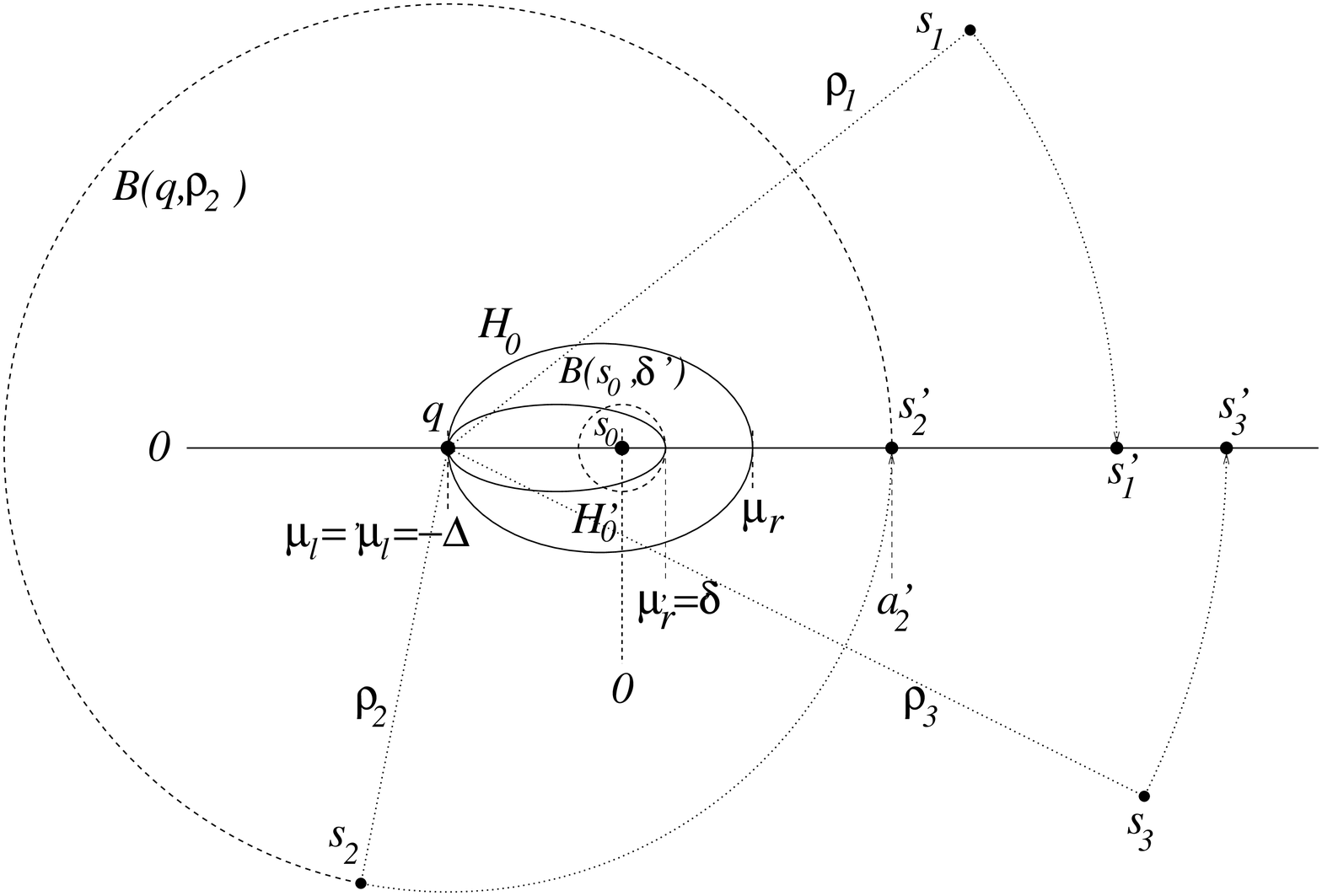}
\LongVersionEnd 
\ShortVersion 
\includegraphics[width=0.6\textwidth]{figs/n-nodes-in-two-dim}
\ShortVersionEnd 
\end{center}
\caption{\label{figure:NNodesInTwoDim}
\(\cA'\) is obtained from \(\cA\) by relocating each station \(\Station_i\) on
the \(x\)-axis.
}
\end{figure}
} 
\LongVersion 
\FigureNNodesInTwoDim{}
\LongVersionEnd 
} 
\LongVersion 
\SectionFatnessPartA{}

\begin{theorem} \label{theorem:ExplicitBounds}
In a uniform energy network \( \cA = \langle S, \bar{1}, \Noise, \beta \rangle
\), where \( S = \{\Station_0, \dots, \Station_{n - 1}\} \) and \( \beta > 1
\) is a constant, if the minimum distance from \(\Station_0\) to any other
station is \( \MinDist > 0 \), then
\[
\SmallRadius(\Station_0, \ReceptionZone_0) \geq \frac{\MinDist}{\sqrt{\beta (n
- 1 + \Noise \cdot \MinDist^2)} + 1}
\quad \text{and} \quad
\LargeRadius(\Station_0, \ReceptionZone_0) \leq \frac{\MinDist}{\sqrt{\beta (1
+ \Noise \cdot \MinDist^2)} - 1} ~ .
\]
The fatness parameter of \(\ReceptionZone_0\) with respect to \(\Station_0\)
thus satisfies
\[
\FatnessParameter(\Station_0, \ReceptionZone_0)
\leq {\frac{\MinDist}{\sqrt{\beta (1 + \Noise \cdot \MinDist^2)} - 1}} \bigg{/}
{\frac{\MinDist}{\sqrt{\beta (n - 1 + \Noise \cdot \MinDist^2)} + 1}}
\leq \frac{\sqrt{\beta (n - 1)} + 1}{\sqrt{\beta} - 1}
= O (\sqrt{n}) ~ .
\]
\end{theorem}

\SectionFatnessPartB{}
\LongVersionEnd 

\section{Handling approximate point location queries}
\label{section:ApproximatePointLocationQueries}

\newcommand{\SectionApproximatePointLocationQueries}{
\subsection{The construction of \QuerDS{}}
\label{section:QueriesForPolyZones}
In this section we describe the construction of \QuerDS{}.
Let \(\GridSpace\) be a positive real to be determined later on.
The data structure \QuerDS{} is based upon imposing a
\(\GridSpace\)-spaced \emph{grid}, denoted by \(\Grid_{\GridSpace}\), on the
Euclidean plane.
The grid is aligned so that the point \(s\) is a grid vertex.
The Euclidean plane is partitioned to grid \emph{cells} with respect to
\(\Grid_{\GridSpace}\) in the natural manner, where ties are broken such that
each cell contains all points on its south edge except its south east corner
and all points on its west edge except its north west corner (the cell does
contain its south west corner).
Given some cell \(C\), we define its \emph{9-cell}, denoted by \( \sharp
C \), as the collection of \( 3 \times 3 \) cells containing \(C\) and the
eight cells surrounding it.

The grid cells will be classified to three types corresponding to the zones
\(\PolyZone^{+}\), \(\PolyZone^{-}\), and \(\PolyZone^{?}\):
cells of type \(\Type^{+}\) are fully contained in \(\PolyZone\);
cells of type \(\Type^{-}\) do not intersect \(\PolyZone\); and
cells of type \(\Type^{?}\) are suspect of partially overlapping
\(\PolyZone\), i.e., having some points in \(\PolyZone\) and some points
not in \(\PolyZone\).
A query on point \( p \in \Reals^{2} \) is handled merely by computing the
cell to which \(p\) belongs and returning its type.
Our analysis relies on bounding the number (and thus the total area) of
\(\Type^{?}\) cells.

By definition, the zone \(\PolyZone\) contains a ball of radius
\(\SmallRadiusBound\) and it is contained in a ball of radius
\(\LargeRadiusBound\), both centered at \(s\).
Clearly, the area of \(\PolyZone\) is lower bounded by the area of any
ball it contains.
Since \(\PolyZone\) is convex, it follows that its perimeter is upper
bounded by the perimeter of any ball that contains it.
Therefore the zone \(\PolyZone\) satisfies
\begin{equation} \label{equation:BoundingAreaAndPerimeter}
\Area(\PolyZone) \geq \pi \SmallRadiusBound^{2}
\quad \text{and} \quad
\Perimeter(\PolyZone) \leq 2 \pi \LargeRadiusBound ~ .
\end{equation}

We will soon present an iterative process, referred to as the \emph{Boundary
Reconstruction Process (BRP)}, which identifies the \(\Type^{?}\) cells.
The union of the \(\Type^{?}\) cells form the zone \(\PolyZone^{?}\) that
contains \(\PolyZone\)'s boundary \( \Boundary \PolyZone = \{ (x, y) \in
\Reals^2 \mid \QuerPoly(x, y) = 0 \} \).
In fact, the zone \(\PolyZone^{?}\) is isomorphic to a ring and in
particular, it partitions \( \Reals^{2} - \PolyZone^{?} \) to a zone
\emph{enclosed} by \(\PolyZone^{?}\) and a zone \emph{outside}
\(\PolyZone^{?}\).
The cells in the former zone (respectively, latter zone) are subsequently
classified as \(\Type^{+}\) cells (resp., \(\Type^{-}\) cells).
We shall conclude by bounding the area of \(\PolyZone^{?}\), showing that it
is at most an \(\epsilon\)-fraction of the area of \(\PolyZone\).

The main ingredient of BRP is a procedure referred to as the \emph{segment
test}.
On input segment \(\sigma\) (which may be open or closed in each endpoint), the
segment test returns the number of distinct intersection points of \(\Boundary
\PolyZone\) and \(\sigma\).
(Since \(\PolyZone\) is convex, this number is either \(0\), \(1\), or
\(2\).)
The segment test is implemented to run in time \(O (m^2)\) by employing Sturm's
condition of the projection of the polynomial \(\QuerPoly(x, y)\) on
\(\sigma\) and by direct calculations of the \(\SINR\) function in the
endpoints of \(\sigma\).
Typically, the segment test will be invoked on segments consisting of edges
of the grid \(\Grid_{\GridSpace}\).

Note that if \(\sigma\) is tangent to \(\Boundary \PolyZone\), then the
segment test reports a single intersection point.
To distinguish this (extreme) case from the (common) case where
\(\Boundary \PolyZone\) crosses \(\sigma\) in a single point, we can append
three other segments to \(\sigma\), thus closing a virtual square, and apply
the segment test to these three new segments.
Since \(\Boundary \PolyZone\) is a closed curve, if \(\Boundary
\PolyZone\) crosses \(\sigma\) and enters the virtual square, then it must
exit it at some point.
On the other hand, by the convexity of \(\PolyZone\), we know that if
\(\sigma\) is tangent to \(\Boundary \PolyZone\), then such a virtual square
cannot intersect \(\Boundary \PolyZone\) at any other point.
(Of course, one should consider two possible such squares, one on each side of
\(\sigma\).)

We now turn to describe BRP.
Informally, the process traverses the boundary of \(\PolyZone\) in the
clockwise direction and identifies the grid cells that intersect it (with some
slack).
Let \(C_{s}\) be the grid cell that contains the point \(s\).
(We will choose the parameter \(\GridSpace\) to ensure that \( \GridSpace <
\SmallRadiusBound / \sqrt{2} \) so that \(C_{s}\) and the three other cells
that share the vertex \(s\) are fully contained in \(\PolyZone\).)
BRP begins by identifying the unique cell \(C_1\) north to \(C_{s}\) (\(C_1\)
and \(C_{s}\) are in the same grid column) which contains a point of
\(\Boundary \PolyZone\) along its west edge.
Note that all grid vertices between \(s\) and the south west corner of \(C_1\)
are in \(\PolyZone\), while the north west corner of \(C_1\) and all the
grid vertices to its north are not in \(\PolyZone\).
The computation of \(C_1\) is performed by direct calculations of the \(\SINR\)
function at grid vertices north of \(s\) in a binary search fashion, starting
with a grid vertex at distance at most \(\LargeRadiusBound\) from \(s\), and
ending with a grid vertex at distance at least \(\SmallRadiusBound\) from
\(s\), so that the total number of \(\SINR\) calculations is \( O (\log
(\LargeRadiusBound / \SmallRadiusBound)) \).

Let \(q\) denote the (unique) intersection point of \(\Boundary
\PolyZone\) and the west edge of \(C_1\).
Consider some continuous (injective) curve function \( \phi : [0, 2 \pi)
\rightarrow \Boundary \PolyZone \) that traverses \(\Boundary \PolyZone\)
in the clockwise direction, aligned so that \( \phi(0) = q \).
For the sake of formality, we extend the domain of \(\phi\) to \( [0, \infty)
\) by setting \( \phi(z) = \phi(z - \lfloor z / (2 \pi) \rfloor \cdot 2 \pi)
\) for every \( z > 2 \pi \).
Let \( z_1 = 0 \).
Given the cell \( C_{j - 1} \) and the real \( z_{j - 1} \in [0, 2 \pi) \), we
define \( z_j = \inf \{ z > z_{j - 1} \mid \phi(z) \notin \sharp C_{j - 1} \}
\).
(Informally, \(\phi(z_j)\) is the first point out of \( \sharp C_{j - 1} \)
encountered along a clockwise traversal of \(\Boundary \PolyZone\) that
begins at \( \phi(z_{j - 1}) \).)

If \( z_j \geq 2 \pi \) (which means that the process have completed a full
encirclement of \(\Boundary \PolyZone\)), then we fix \( m = j \) and BRP
is over.
Assume that \( z_j < 2 \pi \).
If \( \phi(z_j) \notin \sharp C_{j - 1} \), then the cell \(C_j\) is defined
to be the cell containing \(\phi(z_j)\).
Otherwise, the cell \(C_j\) is defined to be the cell containing the point \(
\phi(z_j + \delta) \) for sufficiently small \( \delta > 0 \).
BRP then continues, gradually constructing the collection of \(\Type^{?}\)
cells, consisting of all cells in the 9-cell of \(C_j\) for every \( 1 \leq j
< m \).
The choice of cells \( C_1, \dots, C_{m - 1} \) is illustrated in
\Figure{}~\ref{figure:Grid}.

\newcommand{\FigureGrid}{
\begin{figure}
\begin{center}
\LongVersion 
\includegraphics[width=0.5\textwidth]{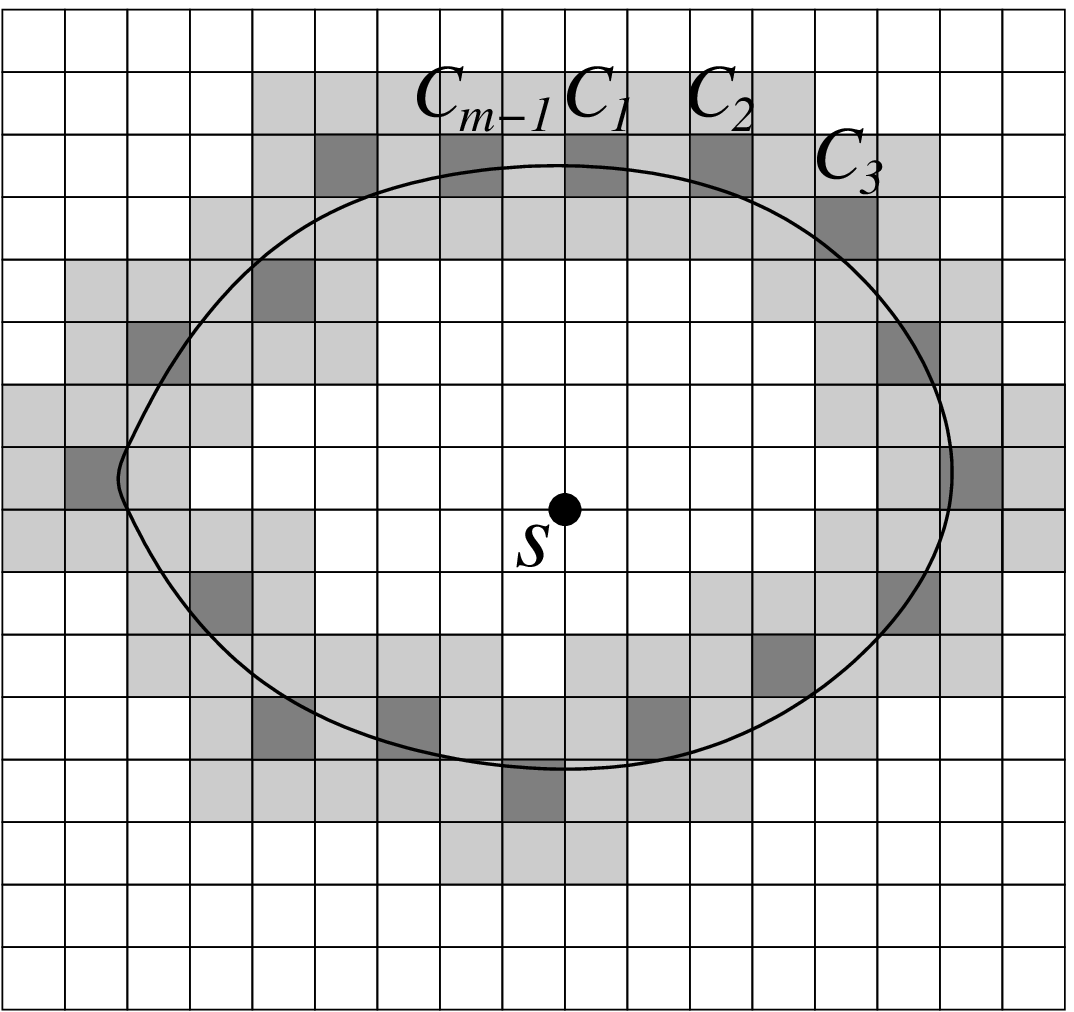}
\LongVersionEnd 
\ShortVersion 
\includegraphics[width=0.4\textwidth]{figs/grid}
\ShortVersionEnd 
\caption{
\label{figure:Grid}
The cell collection \( C_1, \dots, C_{m - 1} \) (dark gray) on top of the
boundary of \(\PolyZone\) (bold curve).
The \(\Type^{?}\) cells are the union of \( C_1, \dots, C_{m - 1} \) and the
\(8\) cells surrounding each one of them (in light gray).
}
\end{center}
\end{figure}
} 
\LongVersion 
\FigureGrid{}
\LongVersionEnd 

It should be clarified that from an algorithmic perspective, we do not
explicitly compute the real sequence \( z_1, \dots, z_m \), but rather the
cell sequence \( C_1, \dots, C_{m - 1} \).
This is done by \(O (1)\) applications of the segment test for every 9-cell
involved in the process.
Since \(\Boundary \PolyZone\) is a closed curve, and since \(\PolyZone\)
is convex, these applications are sufficient to identify the grid edges (and
vertices) crossed by (or tangent to) \(\Boundary \PolyZone\), and hence to
compute the desired cell sequence \( C_1, \dots, C_{m - 1} \).

Next, we bound the number of \(\Type^{?}\) cells.
In every iteration of BRP, we introduce at most \(9\) new \(\Type^{?}\) cells,
hence the total number of \(\Type^{?}\) cells is at most \( 9 (m - 1) \).
Recall our choice of reals \( z_1, \dots, z_m \).
As \( \phi(z_{j - 1}) \) lies on the boundary of \( C_{j - 1} \) and
\(\phi(z_j)\) lies on the boundary of \( \sharp C_{j - 1} \), we conclude that
\( \dist{\phi(z_j), \phi(z_{j - 1})} \geq \GridSpace \) for every \( 1 < j
\leq m \).
Therefore in each iteration of BRP, at least \(\GridSpace\) units of length are
``consumed'' from \(\Perimeter(\PolyZone)\).
By inequality~(\ref{equation:BoundingAreaAndPerimeter}), we have \( m \leq
\lceil \Perimeter(\PolyZone) / \GridSpace \rceil \leq \lceil 2 \pi
\LargeRadiusBound / \GridSpace \rceil \), thus the number of \(\Type^{?}\)
cells is at most \( 9 (m - 1) < 18 \pi \LargeRadiusBound / \GridSpace \).

Since the area of each grid cell is \(\GridSpace^2\), it follows that the
total area of \(\PolyZone^{?}\) (which is the union of the \(\Type^{?}\)
cells) is smaller than \( 18 \pi \LargeRadiusBound \GridSpace \).
In order to guarantee that \( \Area(\PolyZone^{?}) \leq \epsilon \cdot
\Area(\PolyZone) \), we employ
inequality~(\ref{equation:BoundingAreaAndPerimeter}) once more and demand that
\( 18 \pi \LargeRadiusBound \GridSpace \leq \epsilon \pi \SmallRadiusBound^2
\).
Therefore it is sufficient to fix
\[
\GridSpace = \frac{\epsilon \SmallRadiusBound^2}{18 \LargeRadiusBound} ~ ,
\]
which means that the number of \(\Type^{?}\) cells is \( O ((\LargeRadiusBound
/ \SmallRadiusBound)^2 \epsilon^{-1}) \).

Let \(\ColumnsCollection\) be the collection of grid columns with at least one
\(\Type^{?}\) cell.
Clearly, \( |\ColumnsCollection| = O ((\LargeRadiusBound /
\SmallRadiusBound)^2 \epsilon^{-1}) \).
Each column \(\chi\) in \(\ColumnsCollection\) contains at most \(6\)
\(\Type^{?}\) cells.
Consider some cell \(C\) in \(\chi\) which is not a \(\Type^{?}\) cell.
If there is a \(\Type^{?}\) cell to the north of \(C\) and a \(\Type^{?}\)
cell to its south, then \(C\) is a \(\Type^{+}\) cell; otherwise, \(C\) is a
\(\Type^{-}\) cell.
It follows that the data structure \QuerDS{} can be represented as a
vector with one entry per each grid column in \(\ColumnsCollection\) (\( O
((\LargeRadiusBound / \SmallRadiusBound)^2 \epsilon^{-1}) \) entries
altogether), where the entry corresponding to the grid column \( \chi \in
\ColumnsCollection \) stores the \(\Type^{?}\) cells of \(\chi\) (at most
\(6\) of them).
On input point \( p \in \Reals^2 \), we merely have to compute the grid cell
to which \(p\) belongs and (possibly) look up at the relevant entry of
\QuerDS{}.

\subsection{Approximate point location queries in the SINR diagram}
\label{section:QueriesForReceptionZones}
\LongVersion 
\begin{AvoidOverfullParagraph}
\LongVersionEnd 
In this section we explain the relevance of the construction presented in
\LongVersion 
Section~\ref{section:QueriesForPolyZones}
\LongVersionEnd 
\ShortVersion 
Appendix~\ref{section:QueriesForPolyZones}
\ShortVersionEnd 
to \(\epsilon\)-approximate point location queries in the SINR diagram and
establish Theorem~\ref{gtheorem:ApproximatePointLocationQueries}.
Consider some \UPN{} \( \langle S, \bar{1}, \Noise, \beta \rangle \), where \(
S = \{\Station_0, \dots, \Station_{n - 1}\} \) and \( \beta > 1 \) is a
constant.
Recall that the reception zone \( \ReceptionZone_i = \{ (x, y) \in \Reals^2
\mid \HearPoly_i(x, y) \leq 0 \} \) for every \( 0 \leq i \leq n - 1 \), where
\(\HearPoly_{i}(x, y)\) is a \(2\)-variate polynomial of degree at most \( 2 n
\) that is strictly negative for all internal points \((x, y)\) of
\(\ReceptionZone_i\)
\LongVersion 
(see Section~\ref{section:WirelessNetworks}).
\LongVersionEnd 
\ShortVersion 
(see Section~\ref{section:Preliminaries}).
\ShortVersionEnd 
Assuming that the location of \(\Station_i\) is not shared by any other
station (if it is, then \( \ReceptionZone_i = \{\Station_i\} \) and point
location queries are answered trivially), we know that \(\Station_i\) is an
internal point of \(\ReceptionZone_i\).
Furthermore, Theorem~\ref{gtheorem:Convexity} guarantees that the reception
zone \(\ReceptionZone_i\) is a bounded convex zone and
Theorem~\ref{theorem:ExplicitBounds} provides us with a lower bound
\(\SmallRadiusBound\) on \(\SmallRadius(\Station_i, \ReceptionZone_i)\), and
an upper bound \(\LargeRadiusBound\) on \(\LargeRadius(\Station_i,
\ReceptionZone_i)\) such that \( \LargeRadiusBound / \SmallRadiusBound = O
(\sqrt{n}) \).
\LongVersion 
\end{AvoidOverfullParagraph}
\LongVersionEnd 

In fact, we can obtain much tighter bounds on \(\SmallRadius(\Station_i,
\ReceptionZone_i)\) and \(\LargeRadius(\Station_i, \ReceptionZone_i)\).
Let \(r\) be some positive real and assume that we are promised that \(
\SmallRadius(\Station_i, \ReceptionZone_i) = O (r) \) and that \(
\LargeRadius(\Station_i, \ReceptionZone_i) = \Omega (r) \).
Theorem~\ref{theorem:BoundingBoundedBalls} guarantees that \(
\LargeRadius(\Station_i, \ReceptionZone_i) / \SmallRadius(\Station_i,
\ReceptionZone_i) = O (1) \), hence both \(\SmallRadius(\Station_i,
\ReceptionZone_i)\) and \(\LargeRadius(\Station_i, \ReceptionZone_i)\) are
\(\Theta (r) \).
Such a positive real \(r\) is found via an iterative binary-search-like
process that directly computes the values of the \(\SINR\) function of
\(\Station_i\) at points to the, say, north of \(\Station_i\), starting with a
point at distance \(\LargeRadiusBound\) form \(\Station_i\), and ending with a
point at distance at least \(\SmallRadiusBound\) from \(\Station_i\).
Since \( \LargeRadiusBound / \SmallRadiusBound = O (\sqrt{n}) \), it follows
that this process is bound to end within \( O (\log n) \) iterations.
Each iteration takes \( O (n) \) time, thus the improved bounds for
\(\SmallRadius(\Station_i, \ReceptionZone_i)\) and \(\LargeRadius(\Station_i,
\ReceptionZone_i)\) are computed in time \( O (n \log n) \).

Given a performance parameter \( 0 < \epsilon < 1 \), we apply the technique
presented in
\LongVersion 
Section~\ref{section:QueriesForPolyZones}
\LongVersionEnd 
\ShortVersion 
Appendix~\ref{section:QueriesForPolyZones}
\ShortVersionEnd 
to \(\ReceptionZone_i\) and its corresponding polynomial \(\HearPoly_i\) with
the improved bounds on \(\SmallRadius(\Station_i, \ReceptionZone_i)\) and
\(\LargeRadius(\Station_i, \ReceptionZone_i)\) and construct in time \( O (n^2
\epsilon^{-1}) \) a data structure \(\QuerDS_i\) of size \( O (\epsilon^{-1})
\) that partitions the Euclidean plane to disjoint zones \( \Reals^2 =
\ReceptionZone_{i}^{+} \cup \ReceptionZone_{i}^{-} \cup \ReceptionZone_{i}^{?}
\) such that
(1) \( \ReceptionZone_{i}^{+} \subseteq \ReceptionZone_i \);
(2) \( \ReceptionZone_{i}^{-} \cap \ReceptionZone_i = \emptyset \); and
(3) \(\ReceptionZone_{i}^{?}\) is bounded and its area is at most an
\(\epsilon\)-fraction of \(\ReceptionZone_i\).
Given a query point \( p \in \Reals^2 \), \(\QuerDS_i\) answers in constant
time whether \(p\) is in \(\ReceptionZone_{i}^{+}\),
\(\ReceptionZone_{i}^{-}\), or \(\ReceptionZone_{i}^{?}\).
(We construct a separate data structure \(\QuerDS_i\) for every \( 0 \leq i
\leq n - 1 \).)

Recall that by Observation~\ref{observation:ReceptionSetBasicProperties},
point \(p\) cannot be in \(\ReceptionZone_i\) unless it is closer to
\(\Station_i\) than it is to any other station in \(S\).
Thus for such a point \(p\) there is no need to query the data structure
\(\QuerDS_j\) for any \( j \neq i \).
A Voronoi diagram of linear size for the \(n\) stations is constructed in \( O
(n \log n) \) preprocessing time, so that given a query point \( p \in \Reals^2
\), we can identify the closest station \(\Station_i\) in time \( O (\log n)
\) and invoke the appropriate data structure \(\QuerDS_i\).

Combining the Voronoi diagram with the data structures \(\QuerDS_i\) for all
\( 0 \leq i \leq n - 1 \), we obtain a data structure \DataStructure{} of size
\( O (n \epsilon^{-1}) \), constructed in \( O (n^3 \epsilon^{-1}) \)
preprocessing time, that decides in time \( O (\log n) \) whether the query
point \(p\) is in \(\ReceptionZone_{i}^{+}\) for some \(i\), in
\(\ReceptionZone_{i}^{?}\) for some \(i\), or neither, which means that \( p
\in \ReceptionZone^{-} = \bigcap_{i = 0}^{n - 1} \ReceptionZone_{i}^{-} \).
Theorem~\ref{gtheorem:ApproximatePointLocationQueries} follows.
} 

\ShortVersion 
Our goal in this section is to prove
Theorem~\ref{gtheorem:ApproximatePointLocationQueries}.
In fact, our technique for approximate point location queries is suitable for
a more general framework of zones (and diagrams).
Let \(\QuerPoly(x, y)\) be a \(2\)-variate polynomial of degree \(m\) and
suppose that the zone
\[
\PolyZone = \left\{ (x, y) \in \Reals^2 \mid \QuerPoly(x, y) \leq 0
\right\}
\]
is a thick set and that \(\QuerPoly(x, y)\) is strictly negative for all
internal points \((x, y)\) of \(\PolyZone\).
Moreover, suppose that we are given an internal point \(s\) of
\(\PolyZone\), a lower bound \(\SmallRadiusBound\) on \(\SmallRadius(s,
\PolyZone)\), and an upper bound \(\LargeRadiusBound\) on \(\LargeRadius(s,
\PolyZone)\).
Let \( 0 < \epsilon < 1 \) be a predetermined performance parameter.
We construct in \( O (m^2 (\LargeRadiusBound / \SmallRadiusBound)^2
\epsilon^{-1}) \) preprocessing time a data structure \QuerDS{} of size
\( O ((\LargeRadiusBound / \SmallRadiusBound)^2 \epsilon^{-1}) \).
\QuerDS{} essentially partitions the Euclidean plane into three
disjoint zones \( \Reals^2 = \PolyZone^{+} \cup \PolyZone^{-} \cup
\PolyZone^{?} \) such that \\
(1) \( \PolyZone^{+} \subseteq \PolyZone \); \\
(2) \( \PolyZone^{-} \cap \PolyZone = \emptyset \); and \\
(3) \(\PolyZone^{?}\) is bounded and its area is at most an
\(\epsilon\)-fraction of the area of \(\PolyZone\). \\
Given a query point \( p \in \Reals^2 \), \QuerDS{} answers in constant
time whether \(p\) is in \(\PolyZone^{+}\), \(\PolyZone^{-}\), or
\(\PolyZone^{?}\).

In Appendix~\ref{section:QueriesForPolyZones}
we describe the construction of \QuerDS{}.
In Appendix~\ref{section:QueriesForReceptionZones}
we explain how the reception zones and the SINR diagram fall into the above
framework and establish
Theorem~\ref{gtheorem:ApproximatePointLocationQueries}.
\ShortVersionEnd 

\LongVersion 
Our goal in this section is to prove
Theorem~\ref{gtheorem:ApproximatePointLocationQueries}.
In fact, our technique for approximate point location queries is suitable for
a more general framework of zones (and diagrams).
Let \(\QuerPoly(x, y)\) be a \(2\)-variate polynomial of degree \(m\) and
suppose that the zone
\[
\PolyZone = \left\{ (x, y) \in \Reals^2 \mid \QuerPoly(x, y) \leq 0
\right\}
\]
is a thick set and that \(\QuerPoly(x, y)\) is strictly negative for all
internal points \((x, y)\) of \(\PolyZone\).
Moreover, suppose that we are given an internal point \(s\) of
\(\PolyZone\), a lower bound \(\SmallRadiusBound\) on \(\SmallRadius(s,
\PolyZone)\), and an upper bound \(\LargeRadiusBound\) on \(\LargeRadius(s,
\PolyZone)\).
Let \( 0 < \epsilon < 1 \) be a predetermined performance parameter.
We construct in \( O (m^2 (\LargeRadiusBound / \SmallRadiusBound)^2
\epsilon^{-1}) \) preprocessing time a data structure \QuerDS{} of size
\( O ((\LargeRadiusBound / \SmallRadiusBound)^2 \epsilon^{-1}) \).
\QuerDS{} essentially partitions the Euclidean plane into three
disjoint zones \( \Reals^2 = \PolyZone^{+} \cup \PolyZone^{-} \cup
\PolyZone^{?} \) such that \\
(1) \( \PolyZone^{+} \subseteq \PolyZone \); \\
(2) \( \PolyZone^{-} \cap \PolyZone = \emptyset \); and \\
(3) \(\PolyZone^{?}\) is bounded and its area is at most an
\(\epsilon\)-fraction of the area of \(\PolyZone\). \\
Given a query point \( p \in \Reals^2 \), \QuerDS{} answers in constant
time whether \(p\) is in \(\PolyZone^{+}\), \(\PolyZone^{-}\), or
\(\PolyZone^{?}\).

In Section~\ref{section:QueriesForPolyZones}
we describe the construction of \QuerDS{}.
In Section~\ref{section:QueriesForReceptionZones}
we explain how the reception zones and the SINR diagram fall into the above
framework and establish
Theorem~\ref{gtheorem:ApproximatePointLocationQueries}.

\SectionApproximatePointLocationQueries{}
\LongVersionEnd 

\clearpage
\bibliographystyle{plain}
\bibliography{sinr}

\ShortVersion 
\clearpage

\pagenumbering{roman}
\appendix

\renewcommand{\theequation}{A-\arabic{equation}}
\setcounter{equation}{0}

\centerline{\large{\bf APPENDIX}}

\section{A brief overview of some basic notions in point set topology}
\label{appendix:TopologyOverview}
\TopologyOverview{}

\section{Additional material for the convexity proof}

\subsection{Proof of Lemma~\ref{lemma:StarConvex}}
\label{appendix:ProofLemmaStarConvex}
\ProofLemmaStarConvex{}
\qed

\subsection{Proof of Proposition~\ref{proposition:OppositeSign}}
\label{appendix:ProofPropositionOppositeSign}
\ProofPropositionOppositeSign{}
\qed

\subsection{Proof of Proposition~\ref{proposition:BallBoundariesIntersect}}
\label{appendix:ProofPropositionBallBoundariesIntersect}
\ProofPropositionBallBoundariesIntersect{}
\qed

\subsection{Adding background noise}
\label{appendix:AddingNoise}
\SectionAddingNoise{}

\section{The fatness of the reception zones}
\label{appendix:Fatness}
\SectionFatnessPartA{}
\SectionFatnessPartB{}

\section{Handling approximate point location queries}
\label{appendix:ApproximatePointLocationQueries}
\SectionApproximatePointLocationQueries{}

\clearpage
\renewcommand{\thepage}{}

\centerline{\large{\bf FIGURES}}
\bigskip

\FigureExample{}
\FigureSinrUdgThree{}
\FigureNonConvex{}
\FigureApproxQuery{}
\FigureFatness{}
\FigureStarShape{}
\FigureRoots{}
\FigureBTwoIsInsideBOne{}
\FigureBOneAndBTwoIntersect{}
\FigureInductiveStep{}
\FigureAddingBackgroundNoise{}
\FigureTwoNodesOnTheLine{}
\FigureNNodesOnTheLine{}
\FigureNNodesInTwoDim{}
\FigureGrid{}

\ShortVersionEnd 

\end{document}